\documentclass[aos,preprint]{imsart}
\RequirePackage[OT1]{fontenc}
\RequirePackage{amsthm,amsmath}
\RequirePackage[authoryear]{natbib}
\RequirePackage[colorlinks,citecolor=blue,urlcolor=blue]{hyperref}

\usepackage{bm}
\usepackage{graphicx}
\usepackage{amssymb,amsmath,amsthm,mathrsfs}
\usepackage{epstopdf}
\usepackage{algorithm}
\usepackage{amsfonts,delarray}
\usepackage{algorithm,algcompatible,amsmath}
\usepackage{rotating,color}
\usepackage{natbib,multirow}
\usepackage{titlesec,textcase}
\usepackage{longtable}
\usepackage{bbm}
\usepackage{rotating}
\usepackage{times,url}
\usepackage{hyperref}
\usepackage{dblfloatfix}
\usepackage[flushleft]{threeparttable}
\usepackage{arydshln}
\usepackage{graphicx}

\startlocaldefs

\newtheorem{lemma}{Lemma}

\newtheorem{condition}{Condition}

\newtheorem{theorem}{Theorem}
\newtheorem{remark}{Remark}

\newcommand{\bbeta}{{\boldsymbol{\beta}}}
\newcommand{\bmu}{{\boldsymbol{\mu}}}
\newcommand{\bw}{{\boldsymbol{w}}}

\DeclareMathOperator*{\argmin}{arg\,min}
\DeclareMathOperator*{\argmax}{arg\,max}

\definecolor{rp}{RGB}{83,54,106}

\def\boxit#1{\vbox{\hrule\hbox{\vrule\kern6pt\vbox{\kern6pt#1\kern6pt}\kern6pt\vrule}\hrule}}

\endlocaldefs

\begin{document}

\begin{frontmatter}
	\title{Community detection with nodal information}
	\runtitle{Community detection with nodal information}
\begin{aug}
	\author{\fnms{Haolei} \snm{Weng}\ead[label=e1]{hw2375@columbia.edu}} 
	\and 
	\author{\fnms{Yang} \snm{Feng}\thanksref{t2}\ead[label=e2]{yang.feng@columbia.edu}}
\thankstext{t2}{Partially supported by NSF CAREER Grant DMS-1554804.}

\runauthor{H. Weng and Y. Feng}

\affiliation{Columbia University}
\address{Department of Statistics, \\ Columbia University, \\New York, NY, 10027. \\
\printead{e1}\\
\phantom{E-mail:\ }\printead*{e2}}
\end{aug}

\begin{abstract}
Community detection is one of the fundamental
problems in the study of network data. Most existing community detection
approaches only consider edge information as inputs, and the output
could be suboptimal when nodal information is available. In such
cases, it is desirable to leverage nodal information for the
improvement of community detection accuracy. Towards this goal, we
propose a flexible network model incorporating nodal information, and
develop likelihood-based inference methods.  For the proposed methods,
we establish favorable asymptotic properties as well as efficient
algorithms for computation. Numerical experiments show the
effectiveness of our methods in utilizing nodal information across a
variety of simulated and real network data sets.
\end{abstract}



\end{frontmatter}

\section{Introduction.}\label{sec1}

Networked systems are ubiquitous in  modern society. Examples include worldwide web, gene regulatory networks, and social networks. Network analysis has attracted a lot of research attention from social science, physics, computer science and mathematical science. There have been some interesting findings regarding the network structures, such as small world phenomena and power-law degree distributions \citep{newman2003structure}. One of the fundamental problems in network analysis is detecting and characterizing community structure in networks. Communities can be intuitively understood as groups of nodes which are densely connected within groups while sparsely connected between groups\footnote{More rarely, one can encounter communities of the opposite meaning in disassortative mixing networks.}. Identifying network communities not only helps better understand structural features of the network, but also offers practical benefits. For example, communities in social networks tend to share similar interest, which could provide useful information to build recommendation systems. 

Existing community detection methods can be roughly divided into algorithmic and model-based ones \citep{zhao2012consistency}. Algorithmic methods typically define an objective function such as modularity \citep{newman2006modularity},  which measures the goodness of a network partition, and design algorithms to search for the solution of the corresponding optimization problem. See \cite{fortunato2010community} for a thorough discussion of various algorithms. Unlike algorithmic approaches, model-based methods first construct statistical models that are assumed to generate the networks under study, and then develop statistical inference tools to learn the latent communities. Some popular models include stochastic block model \citep{holland1983stochastic}, degree-corrected stochastic block model \citep{dasgupta2004spectral, karrer2011stochastic} and mixed membership stochastic block model \citep{airoldi2009mixed}. 

In recent years, there have been increasingly active researches towards understanding the theoretical performances of community detection methods under different types of models. Regarding the stochastic block model, consistency results have been proved for likelihood based approaches, including maximum likelihood \citep{celisse2012consistency, choi2012stochastic}, profile likelihood \citep{bickel2009nonparametric}, pseudo likelihood \citep{amini2013pseudo} and variational inference \citep{celisse2012consistency, bickel2013asymptotic}, among others. Some of the existing results are generalized to degree-corrected block models \citep{zhao2012consistency}. Another line of theoretical works focuses on methods of moments. See \cite{rohe2011spectral, lei2014consistency, jin2015fast, qin2013regularized, joseph2013impact} on theoretical analysis of spectral clustering for detecting communities in block models. Spectral clustering \citep{zhang2014detecting} and tensor spectral method \citep{anandkumar2014tensor} have also been used to detect overlapping communities 
under mixed membership models. In addition, carefully constructed convex programming has been shown to enjoy provable guarantees for community detection \citep{chen2012clustering, amini2014semidefinite, cai2015robust, guedon2015community,chen2015convexified}. See also the interesting theoretical works of community detection under minimax framework \citep{zhang2015minimax, gao2015achieving}. Finally, there exists a different research theme focusing on detectability instead of consistency \citep{decelle2011asymptotic, krzakala2013spectral, saade2014spectral, abbe2015community}.

All the aforementioned methods are  based on only the observations of the edge connections in the networks. In the real world, however, networks often appear with additional nodal information. For example, social networks such as Facebook and Twitter contain users' personal profile information. A citation network has the authors' names, keywords, and abstracts of papers. Since nodes in the same communities tend to share similar features, we can expect that nodal attributes are in turn indicative of community structures. Combining both sources of edge and nodal information opens the possibility for more accurate community discovery. Many efficient heuristic algorithms are proposed in recent years to accomplish this goal \citep{akoglu2012pics, ruan2013efficient, chang2010hierarchical,  nallapati2008link, yang2013community}. However, not much theory has been established to understand the statistical properties. See \cite{binkiewicz2014covariate, zhang2013community} for some theoretical developments. In this paper, we aim to give a thorough study of the community detection with nodal information problem. Our work first introduces a flexible modeling framework tuned for community detection when edge and nodal information coexist. Under a specific model, we then study three likelihood methods and derive their asymptotic properties. Regarding the computation of the estimators, we resort to a convex relaxation (semidefinite programming) approach to obtain a preliminary community estimate serving as a good initialization fed into ``coordinate" ascent type iterative algorithms, to help locate the global optima. Various numerical experiments demonstrate that our methods can  accurately discover community structures by making efficient use of nodal information.

The rest of the paper is organized as follows. Section 2 introduces our network model with nodal information. We then propose likelihood based methods and derive the corresponding asymptotic properties in Section 3. Section 4 is devoted to the design and analysis of practical algorithms. Simulation examples and real data analysis are presented in Section 5. We conclude the paper with a discussion in Section 6. All the technical proofs are collected in the Appendix.

\section{Network Modeling with Nodal Information.}
A network is usually represented by a graph $G(V, E)$, where $V=\{1,2,\dots,n\}$ is the set of nodes and $E$ is the set of edges. Throughout the paper, we will focus on the networks in which the corresponding graphs are undirected and contain no self-edges. The observed edge information can be recorded in the adjacency matrix $A \in \{0,1\}^{n \times n}$, where $A_{ij}=A_{ji}=1$ if and only if $(i,j) \in E$. Suppose the network can be divided into $K$ non-overlapping communities. Let $\bm{c}=(c_1,\dots, c_n)$ be the community assignment vector, with $c_i$ denoting the community membership of node $i$ and taking values in $\{1, 2, \dots, K\}$. Additionally, the available nodal information is formulated in a covariate matrix $X=(\bm{x}_1, \dots, \bm{x}_n)^T \in \mathbb{R}^{n\times p}$, where $\bm{x}_i \in \mathbb{R}^p$ is the $i$-th node's covariate vector. The goal is to estimate $\bm{c}$ from the observations $A$ and $X$. 
\subsection{Conditional Independence.}\label{model:ci}

We treat $A, X$ and $\bm{c}$ as random and posit a statistical model for them. Before  introducing the model, we would like to elucidate the main motivation. For the purpose of community detection, we  follow the standard two-step inference procedure:
\begin{itemize}
\item[(1)] Derive parameter estimator $\hat{\bm{\theta}}$ based on $P(A, X; \bm{\theta})$,
\item[(2)] Perform posterior inference according to $P(\bm{c} \mid A,X; \hat{\bm{\theta}})$.
\end{itemize}
Under this framework, we now make a conditional independence assumption: \emph{$A \perp X \mid  \bm{c}$}. Admittedly,  the assumption imposes a strong constraint that given the community membership, what nodes are like (described by covariates $X$) does not affect how they are connected (encoded in $A$). On the other hand, this assumption is consistent with our belief that knowing nodal information can help identify community structure $\bm{c}$. More importantly, this assumed conditional independence turns out to simplify the above two steps to a great extent. First,  for the parameter estimation step, the conditional independence assumption implies that
\begin{align}
P(A,X;  \bm{\theta})=\sum_{\bm{c}} P(A\mid \bm{c})P(X\mid \bm{c}) P(\bm{c})=P(X; \bm{\theta}_1) \sum_{\bm{c}} P(A\mid \bm{c}; \bm{\theta}_2) P(\bm{c} \mid X;  \bm{\theta}_3),\label{mle:one}
\end{align}
where $\bm{\theta}=(\bm{\theta}_1,\bm{\theta}_2, \bm{\theta}_3)$ indexes a family of generative models (not restricted to parametric forms). Regarding the second step, conditional independence leads to 
\begin{align}
P(\bm{c} \mid A, X)& =  \frac{P(A\mid \bm{c}) P(X \mid \bm{c})P(\bm{c})}{\sum_{\bm{c}} P(A\mid \bm{c}) P(X \mid \bm{c})P(\bm{c})} =\frac{P(A\mid \bm{c}) P(\bm{c} \mid X)P(X)}{\sum_{\bm{c}} P(A\mid \bm{c}) P(\bm{c} \mid X)P(X)}  \label{mle:two}  \\
& = \frac{P(A\mid \bm{c}) P(\bm{c} \mid X)}{\sum_{\bm{c}} P(A\mid \bm{c}) P(\bm{c} \mid X)}.\nonumber
\end{align}
From  \eqref{mle:one} and \eqref{mle:two}, we observe that the distribution $P(X)$ is a ``nuisance" in the two-step procedure. Hence we are able to avoid modeling and estimating the marginal distribution of $X$. As a result, the effort can be saved for the inference of $P(A\mid \bm{c})$ and $P(\bm{c} \mid X)$.

\subsection{Node-coupled Stochastic Block Model.} \label{nsbm}

The conditional independence and follow-up arguments in Section \ref{model:ci} pave the way to a flexible framework of models for networks with nodal covariates: specifying the two conditionals $P(A\mid \bm{c})$ and $P(\bm{c}\mid X)$. A similar modeling strategy was proposed in \cite{newman2015structure}, along with detailed empirical results. Unlike them, we will consider a different model and present a thorough study from both theoretical and computational perspectives. Note that the conditional distribution $P(A\mid \bm{c})$ only involves the edge information of the network, while $P(\bm{c} \mid X)$ is often encountered in the standard regression setting for i.i.d. data. This motivates us to consider the following model.

\vspace{0.2cm}
\emph{Node-coupled Stochastic Block Model} (NSBM):
\begin{itemize}
\item[(a)] $P(A\mid \bm{c})=\underset{ i <j}{\prod} B_{c_ic_j}^{A_{ij}}(1-B_{c_ic_j})^{1-A_{ij}}$,
\item[(b)] $P(\bm{c} \mid X)=\underset{ i}{\prod} \frac{\exp(\bm{\beta}^T_{c_i}\bm{x}_i)}{\sum_{k=1}^K\exp({\bm{\beta}^T_{k}\bm{x}_i})}$,
\end{itemize}
where $B=(B_{ab}) \in [0,1]^{K\times K}$ is symmetric, $\bm{\beta}=(\bm{\beta}_1,\dots, \bm{\beta}_K) \in \mathbb{R}^{K  p}$. The distribution $P(A\mid \bm{c})$ in (a) follows the stochastic block model (SBM), which, as the fundamental model, has been extensively studied in the literature. The SBM implies that the distribution of an edge between node $i$ and $j$ only depends on their community membership $c_i$ and $c_j$. The nodes from the same community are stochastically equivalent. The element $B_{ab}$ in the matrix $B$ represents the probability of edge connection between a node in community $a$ and a node in community $b$. The $P(\bm{c} \mid X)$ in (b) simply takes a multi-logistic regression form, where we will assume $\bm{\beta}_K=\bm{0}$ for identifiability. Simple as it looks, we would like to point out some advantages of NSBM:
\begin{itemize}
\item The parameters in NSBM can be estimated by combining the estimation of $B$ under (a) and $\bm{\beta}$ under (b), as we shall elaborate in Section \ref{algo:pp}.
\item The coefficient $\bm{\beta}$ reflects the contribution of each nodal covariate for identifying community structures. This information can help us better understand the implication of the network communities. 
\item The probability $p(c=k \mid \bm{x})=\frac{\exp({\bm{\beta}^T_{k}\bm{x}})}{\sum_{k=1}^K\exp({\bm{\beta}^T_{k}\bm{x}})}$ can be used to predict a new node's community membership $c$ based on its covariates $\bm{x}$, without waiting for it to form network connections.
\item Both $P(A\mid \bm{c})$ and $P(\bm{c} \mid X)$ can be readily generalized to fit more complicated structures.
\end{itemize}
\begin{remark}
As illustrated in Section \ref{model:ci}, under the conditional independence assumption $A \perp X \mid  \bm{c}$, 	it is sufficient  to consider the conditional likelihood $P(A,\bm{c}\mid X)$ instead of the full version $P(A,\bm{c},X)$. In particular, we study the maximum likelihood estimate,  maximum variational likelihood estimate, and the maximum profile likelihood estimate based on the conditional likelihood in the next section. However, we emphasize that the conditional independence assumption is not part of NSBM, though it was used to motivate the model. In the next section, we will treat this assumption as working independence  to derive the likelihood based estimates. Hence, the three estimates are in fact based on pseudo likelihood. With a slight abuse of terminology and for simplicity, we still call them the aforementioned likelihood names in the rest of the paper.
\end{remark}

\section{Statistical Inference under NSBM.} \label{infer:nsbm}

For community detection, our main goal is to find an accurate community assignment estimator $\hat{\bm{c}}$ for the underlying true communities $\bm{c}$. Theoretically, we would like to  study the consistency of community detection for a given method. We adopt the notions of consistency from \cite{bickel2009nonparametric} and \cite{zhao2012consistency}:
\begin{eqnarray*}
(\mbox{\emph{strong consistency}})&&P(\hat{\bm{c}}=\bm{c}) \rightarrow 1, \mbox{~as~} n \rightarrow \infty, \\
(\mbox{\emph{weak consistency}})&&\forall \epsilon>0, P\Big(\frac{1}{n}\sum_{i=1}^n \mathbbm{1}(\hat{c}_i \neq c_i) < \epsilon \Big) \rightarrow 1, \mbox{~as~} n \rightarrow \infty.
\end{eqnarray*}
As the network size increases to infinity, with probability approaching 1, strong consistency requires perfect recovery of the true community structure, while weak consistency only needs the mis-classification rate to be arbitrarily small. Note that since community structure is invariant under a permutation of the community labels in $\{1, 2, \dots, K\}$, the consistency notations above as well as the estimators to be introduced should always be interpreted up to label permutations.

In the asymptotic setting where the network size $n \rightarrow \infty$, holding the parameter $B \in [0,1]^{K\times K}$ unchanged implies that the total number of edges present in the network is of order $O(n^2)$. Such networks are unrealistically dense. To study under a more realistic asymptotic framework, we allow $B$ to change with $n$. In particular, we consider a sequence of submodels where $B =\rho_n \bar{B}$ with $\bar{B}$ fixed and $\rho_n=P(A_{ij}=1) \rightarrow 0$ as $n \rightarrow \infty$. The same asymptotic formulation was studied in \cite{bickel2009nonparametric, zhao2012consistency, bickel2013asymptotic}. In this way, the parameter $\rho_n$ directly represents the sparsity level of the network. For the consistency results to be derived in the subsequent sections, we will specify the sufficient conditions on the order of $\rho_n$. 

As pointed out in Section \ref{nsbm}, the parameter $\bm{\beta} $ in NSBM is associated with the contribution of each nodal covariate for discovering communities. Measuring the importance of each nodal attribute to the community structure may provide insightful information about the network. For that purpose, in addition to community detection, we will study the asymptotics of the estimators for $\bm{\beta}$ as well. Since the parameter $B$ is not of current interest, we will skip the theoretical analysis of the corresponding estimators.

\subsection{Consistency of Maximum Likelihood Method.} \label{cmle}

In Section \ref{model:ci}, we pointed out the appealing implication of the assumed conditional independence for likelihood based inference procedure. We now evaluate this procedure under the asymptotic framework we introduced at the beginning of Section \ref{infer:nsbm}. Towards that end, we define the following maximum likelihood based estimators\footnote{Recall that the likelihood formulation is the pseudo version as pointed out in Remark 1. Similar explanations hold for the other two likelihood based methods presented in the subsequent sections. }:
\begin{align}
(\hat{\bm{\beta}}, \hat{B})&=\argmax_{\substack{\bm{\beta}_K=\bm{0},~ \bm{\beta} \in \mathbb{R}^{Kp} \\  B \in [0,1]^{K\times K}, B^T = B}}\sum_{\bm{c}}~ \underset{ i <j}{\prod} B_{c_ic_j}^{A_{ij}}(1-B_{c_ic_j})^{1-A_{ij}}\cdot \underset{ i}{\prod} \frac{\exp(\bm{\beta}^T_{c_i}\bm{x}_i)}{\sum_{k=1}^K\exp({\bm{\beta}^T_{k}\bm{x}_i})},  \label{mle:form}\\
\hat{\bm{c}} &= \argmax_{\bm{c} \in \{1,\dots, K\}^n} \underset{ i <j}{\prod} \hat{B}_{c_ic_j}^{A_{ij}}(1-\hat{B}_{c_ic_j})^{1-A_{ij}}\cdot \underset{ i}{\prod} \frac{\exp({\hat{\bm{\beta}}^T_{c_i}\bm{x}_i})}{\sum_{k=1}^K\exp({\hat{\bm{\beta}}^T_{k}\bm{x}_i})}.  \label{mle:comd}
\end{align}

The estimators defined above are the realizations of the two-step procedure we mentioned at the beginning of Section \ref{model:ci}. We are mainly interested in studying the consistency of $\hat{\bm{\beta}}$ and $\hat{\bm{c}}$. 

First, we would like to introduce several technical conditions. 

\begin{condition}\label{condition::idenfiability}
	$\bar{B}$ has no two identical columns. 
\end{condition}
If the probability matrix $\bar{B}$ has two identical columns, then there exist at least two communities unidentifiable with each other. In practice, it makes sense to combine those communities into a bigger one. 

\begin{condition}\label{condition::c-x}
	$(c_1, \bm{x}_1), \dots, (c_n, \bm{x}_n) \overset{iid}{\sim} (c,  \bm{x})$ with $\mathbb{E}(\bm{x}\bm{x}^T)\succ 0$, where $\succ 0$ represents the matrix being positive definite. 
\end{condition}

Condition \ref{condition::c-x} ensures the coefficient vector $\bm{\beta}$ is uniquely identifiable. 

\begin{condition}\label{condition::x-tail}
There exist constants $\kappa_1$ and $\kappa_2$ such that for sufficiently large $t$, we have 
	$$P(\|\bm{x}\|_2 > t) \leq \kappa_1 e^{-\kappa_2 t}.$$
	\end{condition}
Condition \ref{condition::x-tail} imposes a sub-exponential tail bound on $\|\bm{x}\|_2$, which is equivalent to sub-exponential tail assumption on each component of $\bm{x}$, via a simple union bound argument. This covers many different types of covariates like discrete, Gaussian and exponential.

\begin{theorem}\label{thm:mle}
Assume the data $(A, X)$ follows NSBM and  Conditions \ref{condition::idenfiability}, \ref{condition::c-x} and \ref{condition::x-tail} hold. In addition, assume $\frac{n\rho_n}{\log n} \rightarrow \infty$ as $n\rightarrow \infty$.  Then, we have  as $n \rightarrow \infty$
\begin{eqnarray*}
P(\hat{\bm{c}} = \bm{c}) \rightarrow 1, \quad \sqrt{n}(\hat{\bm{\beta}}-\bm{\beta}) \overset{d}{\rightarrow} N({\bf 0}, I^{-1}(\bm{\beta})),
\end{eqnarray*}
where $I(\bm{\beta})$ is the Fisher information for the multi-logistic regression problem of regressing $\bm{c}$ on $X$. 
\end{theorem}
The key condition $\frac{n\rho_n}{\log n} \rightarrow \infty$ requires that the expected degree of every node to grow faster than the order of $\log n$. The same condition has been used in \cite{bickel2009nonparametric} and \cite{zhao2012consistency} to derive strong consistency under SBM. Under the  conditions of the theorem, the maximum likelihood method not only gives us a strong consistent community assignment estimate $\hat{\bm{c}}$, but also a coefficient estimate $\hat{\bm{\beta}}$ which is as efficient as if the true label $\bm{c}$ were known.

\subsection{Consistency of Variational Method.}\label{cvim}

The maximum likelihood method studied in Section \ref{cmle} has been shown to have nice theoretical properties. However, the likelihood function form in \eqref{mle:form} renders the computation of the estimators $(\hat{\bbeta},\hat{B})$ intractable. In particular, it is computationally infeasible to even evaluate the likelihood function value at a non-degenerate point  (when $n$ is not too small), due to the marginalization over all possible membership assignments. To address this computation issue, we propose a tractable variational method, and demonstrate that it enjoys equally favorable asymptotic properties as the maximum likelihood approach. This is motivated by the works about variational methods under SBM \citep{daudin2008mixture, celisse2012consistency, bickel2013asymptotic}. Throughout this section, we will use the generic symbol $P(\cdot)$ to denote joint distributions and $\bm{\theta}=(\bbeta, B)$. To begin with, recall the well known identity:
\begin{eqnarray*}
\log P(A, X; \bm{\theta})=\mathbb{E}_{Q}[ \log P(A,X, \bm{c}; \bm{\theta})-\log Q(\bm{c}) ]+D[Q(\bm{c})~||~ P(\bm{c}\mid A,X; \bm{\theta})],
\end{eqnarray*}
where $Q(\cdot)$ denotes any joint distribution of $\bm{c}$; the expectation $\mathbb{E}_{Q}(\cdot)$ is taken with respect to $\bm{c}$ under $Q(\bm{c})$; $D[\cdot || \cdot]$ is the Kullback-Leibler divergence. Since $D[Q(\bm{c})~||~ P(\bm{c}\mid A,X; \bm{\theta})]\geq 0$ and the equality holds when $Q(\bm{c})=P(\bm{c}\mid A,X; \bm{\theta})$, it is not hard to verify the following variational equality,
\begin{align}\label{vi:g}
\max_{\bm{\theta}}\log P(A, X; \bm{\theta})=\max_{\bm{\theta}, Q(\cdot)} \mathbb{E}_{Q}[ \log P(A,X, \bm{c}; \bm{\theta})-\log Q(\bm{c}) ].
\end{align}
Hence, to compute the maximum likelihood value, we can equivalently solve the optimization problem on the right hand side of \eqref{vi:g}. Note that iteratively optimizing over $\bm{\theta}$ and $Q(\cdot)$ leads to the EM algorithm \citep{dempster1977maximum}. However, the calculation of $P(\bm{c}\mid A, X; \bm{\theta})$ at each iteration of EM is computationally intensive. Instead of optimizing over the full distribution space of $Q(\cdot)$, variational methods aim to solve an approximate optimization problem, by searching over a subset of all possible $Q(\cdot)$. In particular, we consider the mean-field variational approach \citep{jordan1999introduction}, 
\begin{align}\label{vi:mf}
\max_{\bm{\theta}, Q \in \mathcal{Q}}  \mathbb{E}_{Q}[ \log P(A,X, \bm{c}; \bm{\theta})-\log Q(\bm{c}) ],
\end{align}
where $\mathcal{Q}=\{Q: Q(\bm{c})=\prod_{i=1}^n q_{ic_i}, \sum_k q_{ik}=1, 1\leq i \leq n\}$. The subset $\mathcal{Q}$ contains all the distributions under which the elements of $\bm{c}$ are mutually independent. The independence structure turns out to make the computation in \eqref{vi:mf} manageable. We postpone the detailed calculations to Section \ref{algo:pp}, and focus on the asymptotic analysis in this section.  Denote the maximizer in \eqref{vi:mf} by $(\check{\bm{\beta}}, \check{B})$ and 
\begin{align}\label{vi:mf2}
\check{\bm{c}}= \argmax_{\bm{c} \in \{1,\dots, K\}^n} \underset{ i <j}{\prod} \check{B}_{c_ic_j}^{A_{ij}}(1-\check{B}_{c_ic_j})^{1-A_{ij}}\cdot \underset{ i}{\prod} \frac{\exp({\check{\bm{\beta}}^T_{c_i}\bm{x}_i})}{\sum_{k=1}^K\exp({\check{\bm{\beta}}^T_{k}\bm{x}_i})}.
\end{align}

\begin{theorem}\label{thm:mfvi}
Suppose the conditions in Theorem \ref{thm:mle} hold. Then as $n \rightarrow \infty$
\begin{align*}
P(\check{\bm{c}} = \bm{c}) \rightarrow 1, \quad \sqrt{n}(\check{\bm{\beta}}-\bm{\beta}) \overset{d}{\rightarrow} N({\bf 0}, I^{-1}(\bm{\beta})),
\end{align*}
where $I(\bm{\beta})$ is the Fisher information for the multi-logistic regression problem of regressing $\bm{c}$ on $X$. 
\end{theorem}
As we can see, under the same conditions as the maximum likelihood method, the variational approach can deliver equally good estimators, at least in the asymptotic sense. In other words, the approximation made by the variational method does not degrade the asymptotic performance. This should be attributed to the condition $\frac{n\rho_n}{\log n}\rightarrow \infty$, which guarantees the network has sufficient edge information for doing approximate inference.

\subsection{Consistency of Maximum Profile Likelihood Method.}\label{cmplm}

The two methods presented in Sections \ref{cmle} and \ref{cvim} are implementations of the two-step procedure we discussed in Section \ref{model:ci}: first estimating parameters based on the likelihood function and then doing posterior inference using the estimated distribution. In this section, we introduce a one-step method that outputs the parameter and community assignment estimates simultaneously.  The method solves the following problem,
\begin{align}\label{thm:eqeq}
(\tilde{\bm{\beta}}, \tilde{B},\tilde{\bm{c}})=\argmax_{\substack{\bm{\beta}_K=\bm{0}, ~\bm{\beta} \in \mathbb{R}^{Kp}  \\ \bm{c} \in \{1,\dots, K\}^n\\ B \in [0,1]^{K\times K}, B^T = B}  }  \underset{ i <j}{\prod} B_{c_ic_j}^{A_{ij}}(1-B_{c_ic_j})^{1-A_{ij}}\cdot \underset{ i}{\prod} \frac{\exp(\bm{\beta}^T_{c_i}\bm{x}_i)}{\sum_{k=1}^K\exp({\bm{\beta}^T_{k}\bm{x}_i})}.
\end{align}
In the above formulation, we treat the latent variables $\bm{c}$ as parameters and obtain the estimators as the maximizer of the joint likelihood function. This enables us to avoid the cumbersome marginalization encountered in the maximum likelihood method. This approach is known as maximum profile likelihood \citep{bickel2009nonparametric, zhao2012consistency}. \cite{bickel2009nonparametric} showed strong consistency under stochastic block model, and \cite{zhao2012consistency} generalized the results to degree-corrected block models. Following similar ideas, we will investigate this method in the node-coupled stochastic block model. For theoretical convenience, we consider a slightly different formulation:
\begin{align}\label{pmle:eq}
(\tilde{\bm{\beta}}, \tilde{B},\tilde{\bm{c}})=\argmax_{\substack{\bm{\beta}_K=\bm{0}, ~\bm{\beta} \in \mathbb{R}^{Kp}  \\ \bm{c} \in \{1,\dots, K\}^n\\ B \in \mathbb{R}^{K\times K}, B^T = B}  }  \underset{ i <j}{\prod} e^{-B_{c_ic_j}}B_{c_ic_j}^{A_{ij}} \cdot \underset{ i}{\prod} \frac{\exp(\bm{\beta}^T_{c_i}\bm{x}_i)}{\sum_{k=1}^K\exp({\bm{\beta}^T_{k}\bm{x}_i})},
\end{align}
where the Bernoulli distribution in \eqref{thm:eqeq} is replaced by Poisson distribution. In our asymptotic setting $\rho_n \rightarrow 0$, the difference becomes negligible.

\begin{theorem}\label{thm:pmle}
Assume the data $(A, X)$ follows NSBM and  Conditions \ref{condition::idenfiability} and \ref{condition::c-x} hold.
\begin{itemize}
\item[(i)]  If $n\rho_n \rightarrow \infty$ and $\mathbb{E}\|\bm{x}\|_2^{\alpha} < \infty ~(\alpha >1)$, then there exists a constant $\gamma >0$ such that, as $n \rightarrow \infty$
\begin{align*}
P\Big(\frac{1}{n}\sum_{i=1}^n \mathbbm{1}(\tilde{c}_i \neq c_i) \leq \gamma (n\rho_n)^{-1/2} \Big) \rightarrow 1,\quad \| \tilde{\bm{\beta}}-\bm{\beta} \|_2=O_p((n\rho_n)^{\frac{1-\alpha}{2\alpha}}) . 
\end{align*}
\item[(ii)] Assume Condition \ref{condition::x-tail} is satisfied.  If $\frac{n \rho_n}{\log n} \rightarrow \infty$, then as $n \rightarrow \infty$
\begin{align*}
P(\tilde{\bm{c}} = \bm{c}) \rightarrow 1, \quad \sqrt{n}(\tilde{\bm{\beta}}-\bm{\beta}) \overset{d}{\rightarrow} N({\bf 0}, I^{-1}(\bm{\beta})),
\end{align*}
where $I(\bm{\beta})$ is the Fisher information for the multi-logistic regression problem of regressing $\bm{c}$ on $X$. \end{itemize}
\end{theorem}
We see that part (ii) in Theorem \ref{thm:pmle} is identical to Theorems \ref{thm:mle} and \ref{thm:mfvi}. Hence the maximum profile likelihood method is equivalently good as the previous two, in certain sense. The conclusions in part (i) shed lights on how the network edges and nodal covariates affect the consistency results. Under the scaling $n\rho_n \rightarrow \infty$, $\tilde{\bm{c}}$ is only weak consistent. And the higher moment  $\|\bm{x}\|_2$ has, the faster convergence rate $\check{\bm{\beta}}$ can achieve. Suppose all moments of $\|\bm{x}\|_2$ exist, then we would have $\sqrt{n\rho_n}\|\check{\bm{\beta}}-\bm{\beta}\|_2=O_p(1)$. Since $\rho_n \rightarrow 0$, this convergence rate is slower than and may be arbitrarily close to the one in part (ii) when $\frac{n\rho_n}{\log n} \rightarrow \infty$.

\section{Practical Algorithms.}\label{algo:pp}

In Section \ref{infer:nsbm}, we have studied three likelihood based community detection methods and shown their superb asymptotic performances. In this section, we design and analyze specialized algorithms, for computing the variational estimators defined by \eqref{vi:mf}, \eqref{vi:mf2} and the maximum profile likelihood estimators in \eqref{pmle:eq}. As discussed in Section \ref{cvim}, the maximum likelihood estimators are computationally infeasible, hence omitted here. The key challenge lies on the fact that the likelihood based functions in \eqref{vi:mf}, \eqref{vi:mf2} and \eqref{pmle:eq} are all non-convex.  Multiple local optima may exist and the global solution is often impossible to accurately allocate. To address this issue, we first obtain a ``well behaved" preliminary estimator via convex optimization, and then feed it as an initialization into ``coordinate" ascent type iterative schemes. The idea is that the carefully chosen initialization may help the followed-up iterations to escape ``bad" local optima and arrive ``closer" (better approximation) to the ideal global solution. As we shall see in the numerical studies, the results with  a ``well behaved" initial estimator  are significantly better than  those with a random initialization. These two steps will be  discussed in detail  in Sections \ref{ivco:two} and \ref{sds:one}, respectively.

\subsection{Initialization via Convex Optimization.}\label{ivco:two}

The convex optimization we consider in this section is semidefinite programming (SDP). Different formulations of SDP have been shown to yield good community detection performances in \cite{chen2012clustering, amini2014semidefinite, cai2015robust, montanari2015semidefinite, guedon2015community}, among others. One illuminating interpretation of SDP is to think of it as a convex relaxation of the maximum likelihood method. For example, starting from a specialized stochastic block model, one can derive SDP by approximating the corresponding likelihood function. See \cite{chen2012clustering, amini2014semidefinite, cai2015robust} for the detailed arguments. However, under NSBM, because of the nodal covariates term, it is not straightforward to generalize the convex relaxation arguments. We hence resort to a different understanding of SDP elaborated in \cite{guedon2015community}. The key idea is to construct SDP based on the observations directly, with the goal of having the true community assignment $\bm{c}$ to be the solution of a ``population" version of the SDP under construction. Then with a few conditions, we would like to show that the solution of the SDP is ``close" to the solution of its ``population" version, i.e., the true community assignment. In particular, we consider the following semidefinite programming problem,
\begin{align}\label{sdp:formula}
\hat{Z}=\argmax_{Z}& \langle A +\gamma_n XX^T, Z \rangle \\
\mbox{subject to }& Z \succeq 0, Z \in \mathbb{R}^{n \times n}  \nonumber \\
&  0\leq Z_{ij}\leq 1, 1\leq i, j \leq n \nonumber \\
&  \sum_{ij}Z_{ij}=\lambda_n,  \nonumber 
\end{align}
where $\gamma_n,\lambda_n>0$ are two tuning parameters; $\langle \cdot , \cdot \rangle$ denotes the inner product of two matrices. We then obtain the communities by running K-means on $\hat{Z}$ (treating each row of $\hat{Z}$ as a data point in $\mathbb{R}^n$). We show that the approach can produce a consistent community assignment estimate as presented in the following theorem. 

\begin{theorem}\label{thm:sdp}
Assume part (a) in NSBM holds and $(c_1, \bm{x}_1), \dots, (c_n, \bm{x}_n) \overset{iid}{\sim} (c,  \bm{x})$. Let $\{\bar{\bm{c}}_i\}_{i=1}^n$ be the community estimates from running K-means on $\hat{Z}$ defined in \eqref{sdp:formula}. If $\min_a \bar{B}_{aa} > \max_{a\neq b} \bar{B}_{ab}, n\rho_n \rightarrow \infty$ and $\|\bm{x}\|_2$ is sub-Gaussian, then by choosing $\gamma_n =o(\rho_n)$ and $\lambda_n=\sum_{k=1}^K(\sum_{i=1}^n\mathbbm{1}(c_i=k))^2$, we have
\begin{align*}
\frac{1}{n}\sum_{i=1}^n \mathbbm{1}(\bar{c}_i \neq c_i) \overset{P}{\rightarrow} 0.
\end{align*}
\end{theorem}

The proof of Theorem \ref{thm:sdp} will provide a clear picture on why \eqref{sdp:formula} is constructed in that way. But to avoid digression, we defer the proof to the Appendix. As seen from Theorem \ref{thm:sdp}, though the SDP works under more assumptions than previously discussed likelihood based methods, the conditions are not very stringent. The crucial assumption $\min_a \bar{B}_{aa} > \max_{a\neq b} \bar{B}_{ab}$ requires denser edge connections within communities than between them, which is satisfied by most real networks. Furthermore, the two tuning parameters $(\gamma_n,\lambda_n)$ coupled with \eqref{sdp:formula} need to be chosen appropriately. The tuning $\gamma_n$ trades off the information from two different sources: network edge and nodal covariates. The way we incorporate nodal covariates has the same spirit as \cite{binkiewicz2014covariate} does in spectral clustering. From the simulation studies in the next section, we shall see that a flexible choice of $\gamma_n$ can lead to satisfactory results. The choice of  parameter $\lambda_n$ in Theorem \ref{thm:sdp} depends on the unknown truth in a seemingly restrictive way. However, we will demonstrate through simulations that the community detection results are quite robust to the choice of $\lambda_n$. 

The convex optimization problem \eqref{sdp:formula} can be readily solved by standard semidefinite programming solvers such as SDPT3 \citep{tutuncu2003solving}. However, those solvers are based on interior-point methods, and are computationally expensive when the network size $n$ is more than a few hundred. To overcome this limit, we apply the alternating direction method of multipliers (ADMM) to develop a more scalable algorithm for solving \eqref{sdp:formula}. We start by a brief description of the generic ADMM algorithm with the details available in the excellent tutorial by \cite{boyd2011distributed}. In general, ADMM solves problems in the form 
\begin{align}\label{admm:form}
&\mbox{minimize}~~~~ f(\bm{y})+h(\bm{z}) \\
&\mbox{subject~to}~~~B\bm{y}+D\bm{z}=\bm{w}, \nonumber 
\end{align}
where $\bm{y} \in \mathbb{R}^m, \bm{z} \in \mathbb{R}^n, B \in \mathbb{R}^{q \times m}, D \in \mathbb{R}^{q \times n}, \bw \in \mathbb{R}^q;$ and $f(\bm{y}), h(\bm{z})$ are two convex functions. The algorithm takes the following iterations at step $t$. 
\begin{align}
\bm{y}^{t+1}&=\argmin_{\bm{y}} \Big(  f(\bm{y})+(\xi/2)\|B\bm{y}+D\bm{z}^t-\bm{w}+\bm{u}^t\|_2^2   \Big), \label{admm:iter1}\\
\bm{z}^{t+1}&= \argmin_{\bm{z}}\Big(h(\bm{z})+(\xi/2)\|B\bm{y}^{t+1}+D\bm{z}-\bm{w}+\bm{u}^t\|_2^2     \Big), \label{admm:iter2}\\
\bm{u}^{t+1}&=\bm{u}^t+B\bm{y}^{t+1}+D\bm{z}^{t+1}-\bm{w}, \label{admm:iter3}
\end{align}
with $\xi>0$ being a step size constant.

To use this framework, we reformulate \eqref{sdp:formula} as:
\begin{align*}
\mbox{minimize}~~~&l(Z\succeq 0)+l(0\leq Y_{ij}\leq 1, 1\leq i, j \leq n)+l\big(\sum_{ij}W_{ij}=\lambda_n\big)\\
&- \langle A +\gamma_n XX^T, Z \rangle \\
\mbox{subject to}~~~&Y=Z, Y=W,
\end{align*}
where $Z,Y,W \in \mathbb{R}^{n\times n}; l(Z\succeq 0)$ equals 0 if $Z\succeq 0$ and $+\infty$ otherwise; similar definitions hold for other $l(\cdot)$. If we set $\bm{y}=({\mbox{vec}}(Y), {\mbox{vec}}(Y))^T \in \mathbb{R}^{2n^2}, \bm{z}=({\mbox{vec}}(W),{\mbox{vec}}(Z))^T \in \mathbb{R}^{2n^2}, B=-D={I}_{2n^2}\in \mathbb{R}^{2n^2\times 2n^2}, \bm{w}=\bm{0}\in \mathbb{R}^{2n^2}$, where {vec}($\cdot$) denotes the vectorized version of a matrix, then the problem above becomes an instance of \eqref{admm:form}. The corresponding iterations have the following expressions:
\begin{align*}
Y^{t+1}&=\argmin_{0\leq Y_{ij}\leq 1} \Big( \|Y-W^t+U^t\|_F^2+ \|Y-Z^t+V^t\|_F^2    \Big), \\
W^{t+1}&= \argmin_{\sum_{ij}W_{ij}=\lambda_n} \|Y^{t+1}-W+U^t\|_F^2, \\
Z^{t+1}&=\argmin_{Z\succeq 0 } \Big(- \langle A +\gamma_n XX^T, Z \rangle +(\xi/2)\|Y^{t+1}-Z+V^t\|_F^2   \Big),\\
U^{t+1}&=U^t+Y^{t+1}-W^{t+1}, V^{t+1}=V^t+Y^{t+1}-Z^{t+1},
\end{align*}
where $\|\cdot\|_F$ denotes the Frobenius norm. It is not hard to see that each iteration above has a closed form update with the details summarized in Algorithm 1.\footnote{In Step (c), $P\Lambda P^T$ denotes the spectral decomposition; $\Lambda_+$ represents the truncated (keep positive elements) version of $\Lambda$.}

\begin{algorithm}
\caption{Solving \eqref{sdp:formula} via ADMM}
\begin{algorithmic}
\State Input: initialize $Z^0=A+\gamma_n XX^T, W^0=Y^0=U^0=V^0=0$, number of iterations $\mathcal{T}$, step size $\xi$. 
\State For $t=0,\dots, \mathcal{T}-1$
\begin{itemize}
\item[](a) $Y^{t+1} =\min \{\max \{0,  \frac{1}{2}(W^t+Z^t-U^t-V^t)\},1  \}$.
\item[](b) $W^{t+1}=Y^{t+1}+U^t + n^{-2}[{\lambda_n -\sum_{ij}(Y_{ij}^{t+1}+U_{ij}^t)}]\bm{1}\bm{1}^T$.
\item[](c) $Z^{t+1}=P\Lambda_+P^T$, where $Y^{t+1}+V^t+\xi^{-1}\cdot (A+\gamma_n XX^T)=P\Lambda P^T$.
\item[](d) $U^{t+1}=U^t+Y^{t+1}-W^{t+1}, V^{t+1}=V^t+Y^{t+1}-Z^{t+1}$.
\end{itemize}
\State Output $Z^{\mathcal{T}}$. 
\end{algorithmic}
\end{algorithm}

\subsection{Coordinate Ascent Scheme.}\label{sds:one}

As we may see, the problem formulations in \eqref{vi:mf} and \eqref{pmle:eq} are not suitable for gradient or Hessian based iterative algorithms, because they either involve discrete variables or have non-trivial constraints. The variables involved in those optimization problems can be divided into community assignment related and others. Naturally, we will adopt the iterative scheme that alternates between these two types of variables. 

\subsubsection{Computing Variational Estimates.}

To compute the variational estimates in \eqref{vi:mf}, we follow the EM style iterative fashion by maximizing the objective function in \eqref{vi:mf} with respect to $\bm{\theta}$ and $Q \in \mathcal{Q} $ alternatively. Specifically, we are solving

\begin{align}
\bm{\beta}^{t+1}=&\argmax_{\bm{\beta}\in \mathbb{R}^{pK}, \bm{\beta}_K=\bm{0}} \sum_{i} \Big[ \big(\sum_k q^t_{ik}\bm{\beta}_k\big)^T\bm{x}_i-\log\big(\sum_{k}e^{\bm{\beta}_k^T\bm{x}_i}\big) \Big], \label{em:two}
\end{align}
\begin{align}
B^{t+1}=&\argmax_{B\in \mathbb{R}^{K \times K}, B^T = B} \sum_{ab} \Big [\log B_{ab}\cdot \sum_{i<j} A_{ij}q^t_{ia}q^t_{jb}+\log(1-B_{ab}) \cdot \sum_{i<j}(1-A_{ij})q^t_{ia}q^t_{jb}  \Big],  \label{em:one}
\end{align}
\begin{align}\label{em:three} 
\{q^{t+1}_{ik}\}= \argmax_{\{q_{ik}\}}&\sum_{ab} \Big [\log B^{t+1}_{ab}\cdot \sum_{i<j} A_{ij}q_{ia}q_{jb}+\log(1-B^{t+1}_{ab}) \cdot \sum_{i<j}(1\\&-A_{ij})q_{ia}q_{jb}  \Big]
+\sum_{i}\sum_k q_{ik}(\bm{\beta}^{t+1}_k)^T\bm{x}_i-\sum_i \sum_k q_{ik}\log q_{ik}.  \nonumber 
\end{align}
Note that the objective function in \eqref{em:two} takes a similar form as the log-likelihood function of multi-logistic regression model. We hence use Newton-Raphson algorithm in the same way as we fit multi-logistic regression model, to compute the update in \eqref{em:two}. This corresponds to Step (a) of Algorithm 2, in which we have used the name ``FitMultiLogistic" there to denote the full step with a bit abuse of notation. In addition, the update in \eqref{em:one} has an explicit solution, which corresponds to Step (b) in Algorithm 2.  Regarding the update for $\{q_{ik}\}_{ik}$ in \eqref{em:three}, unfortunately, the optimization is non-convex and does not have analytical solutions.  We then implement an inner blockwise coordinate ascent loop to solve it. In particular, we update $\{q_{ik}\}_{k=1}^K$ one at a time:
\begin{align*}
\{q_{ik}\}_k=\argmax_{\{q_{ik}\}_k} &\sum_{k} q_{ik}  \cdot \Big[ \sum_{b} \sum_{j \neq i} \Big(A_{ij}q_{jb} \cdot \log B_{kb}+(1-A_{ij})q_{jb} \cdot \log(1\\&-B_{kb})\Big) \Big]+  
 \sum_k q_{ik} \bm{\beta}^T_k\bm{x}_i - \sum_k q_{ik}\log q_{ik}.
\end{align*}
It is straightforward to show that the update above has closed forms:
\begin{align*}
q_{ik}=\frac{e^{a_k}}{\sum_{k=1}^K e^{a_k}}, \quad a_k= \bm{\beta}_k^T\bm{x}_i+ \sum_{b} \sum_{j \neq i} q_{jb}\cdot \Big(A_{ij} \log B_{kb}+(1-A_{ij}) \log(1-B_{kb})\Big).
\end{align*}
This yields Step (c) for Algorithm 2. After computing $\{q^{\mathcal{T}}_{ik}\}, \bm{\beta}^{\mathcal{T}}, B^{\mathcal{T}}$ via Algorithm 2, we  calculate the community assignment estimate $\check{\bm{c}}$ based on \eqref{vi:mf2}. This could be done by coordinate ascent iterations, like Step (c) in Algorithm 3 (to be introduced in Section \ref{subsec:compute-MPLE}. Alternatively, we can use the following approximated posterior distribution $\{q^{\mathcal{T}}_{ik}\}$:  
\begin{align}
\check{c}_i=\argmax_{1\leq k \leq K}q^{\mathcal{T}}_{ik}, \quad  1 \leq i \leq n. \label{approx:c}
\end{align}
In numerical studies, we adopt the approach in \eqref{approx:c}, which is computationally more efficient. 
\begin{algorithm}
\caption{Solving \eqref{vi:mf} via iterating between $(\bm{\beta}, B)$ and $Q$.}
\begin{algorithmic}
\State Input: initialize $\{q^0_{ik}\}$, number of iterations $\mathcal{T}$ 
\State For $t=0,\dots, \mathcal{T}-1$
\begin{itemize}
\item[](a) $\bm{\beta}^{t+1}=$ FitMultiLogistic$(X, \{q^t_{ik}\})$
\item[](b) $B^{t+1}_{ab}=\frac{\sum_{i< j} A_{ij} q^t_{ia}q^t_{jb}}{\sum_{i< j} q^t_{ia}q^t_{jb}}$
\item[](c) Update $\{q_{ik}^{t+1}\}$ via Repeating 
\begin{itemize}
\item[] For $i=1,\dots, n$
\begin{itemize}
\item[] $\log q_{ik} \propto (\bm{\beta}_k^{t+1})^T\bm{x}_i +\sum_b \sum_{j\neq i} q_{jb}\cdot \big(A_{ij}\log B^{t+1}_{kb}+ (1-A_{ij})\log (1-B^{t+1}_{kb}) \big)$ 
\end{itemize}
\end{itemize}
\end{itemize}
\State Output $\{q_{ik}^{\mathcal{T}}\}, \bm{\beta}^{\mathcal{T}}, B^{\mathcal{T}}$.
\end{algorithmic}
\end{algorithm}

\subsubsection{Computing Maximum Profile Likelihood Estimates.\label{subsec:compute-MPLE}}
Similarly to the variational estimates, we maximize the likelihood function in \eqref{pmle:eq} with respect to $(\bm{\beta},B)$ and $\bm{c}$ iteratively. In other words, we  solve
\begin{align}
\bm{\beta}^{t+1}=&\argmax_{\bm{\beta}\in \mathbb{R}^{pK}, \bm{\beta}_K=\bm{0}} \sum_{i} \Big[ \bm{\beta}^T_{c^t_i}\bm{x}_i-\log\big(\sum_{k}e^{\bm{\beta}_k^T\bm{x}_i}\big) \Big], \label{pmle:two}
\end{align}
\begin{align}
B^{t+1}_{ab}=&\argmax_{B_{ab}}~ \log B_{ab} \cdot \sum_{i<j} A_{ij}\mathbbm{1}(c^t_i=a, c^t_j=b)-B_{ab} \cdot \sum_{i<j} \mathbbm{1}(c^t_i=a, c^t_j=b),   \label{pmle:one}
\end{align}
\begin{align}
 \bm{c}^{t+1}=&\argmax_{\bm{c} \in \{1,\dots, K\}^n} \sum_{ab} \Big [\log B^{t+1}_{ab}\cdot \sum_{i<j} A_{ij}\mathbbm{1}(c_i=a, c_j=b) \label{pmle:three}\\
&-B^{t+1}_{ab} \cdot \sum_{i<j}\mathbbm{1}(c_i=a, c_j=b)  \Big] +\sum_{i}(\bm{\beta}^{t+1}_{c_i})^T\bm{x}_i. \nonumber 
\end{align}
Here, solving \eqref{pmle:two} is equivalent to computing the maximum likelihood estimate of multi-logistic regression. This is carried out in Step (a) of Algorithm 3. In addition, it is straightforward to see that Step (b) in Algorithm 3 is the solution to \eqref{pmle:one}.  For computing $\bm{c}^{t+1}$ in \eqref{pmle:three}, we update its element one by one, as shown by Step (c) in Algorithm 3.

\begin{algorithm}
\caption{Solving \eqref{pmle:eq} via iterating between $(\bm{\beta}, B)$ and $\bm{c}$.}
\begin{algorithmic}
\State Input: initialize $\{c^0_i\}$, number of iterations $\mathcal{T}$
\State For $t=0, \dots, \mathcal{T}-1$
\begin{itemize}
\item[](a) $\bm{\beta}^{t+1}=$ FitMultiLogistic$(X, \bm{c}^{t})$
\item[](b) $B^{t+1}_{ab}=\frac{\sum_{i < j} A_{ij} \mathbbm{1}(c^t_i=a, c^t_j=b)}{\sum_{i < j} \mathbbm{1}(c^t_i=a, c^t_j=b)}$
\item[](c) Update $\{c^{t+1}_i \}$ via Repeating
\begin{itemize}
\item[] For $i=1, \dots, n$
\begin{itemize}
\item[] $c_i=\argmax_{1 \leq k \leq K} (\bm{\beta}_{k}^{t+1})^T\bm{x}_i+\sum_b\sum_{j\neq i}\mathbbm{1}(c_j=b)\cdot (A_{ij}\log B^{t+1}_{kb}-B^{t+1}_{kb})$ 
\end{itemize}
\end{itemize}
\end{itemize}
\State Output $\{c^{\mathcal{T}}_i\}, \bm{\beta}^{\mathcal{T}}$.
\end{algorithmic}
\end{algorithm}

\subsubsection{Variational Estimates vs. Maximum Profile Likelihood Estimates.}

So far we have studied the theoretical properties of variational and maximum profile likelihood estimates, and developed algorithms to compute them. The results in Sections \ref{cvim} and \ref{cmplm} demonstrate that they have the same asymptotic performance under $\frac{n\rho_n}{\log n} \rightarrow \infty$. We now compare the corresponding  algorithms. By taking a close look at Algorithms 2 and 3, we observe that the three steps in the two algorithms share a lot of similarities. Algorithm 2 is essentially a ``soft" version of Algorithm 3 in the following sense: instead of using the community assignment $c_i$ in Algorithm 3, the steps in Algorithm 2 involve the probability of belonging to every possible community. This might remind us of the comparison between the EM algorithm and K-means under Gaussian mixture models. As we will see in Section \ref{num:key}, variational and maximum profile likelihood methods usually lead to  similar numerical results.

\section{Numerical Experiments.}\label{num:key}

In this section, we conduct a detailed experimental study of the SDP defined in \eqref{sdp:formula}, variational and maximum profile likelihood methods on both simulated and real datasets. We use two quantitative measures for evaluating their community detection performance. 
\begin{itemize}
\item[]\emph{Normalized Mutual Information} \citep{ana2003robust}:
\begin{align*}
\mbox{NMI} = \frac{-2 \sum_{i}  \sum_j n_{ij}\log\Big(\frac{n_{ij}\cdot n}{n_{i\cdot}n_{\cdot j}} \Big)}{\sum_i n_{i\cdot}\log\Big(\frac{n_{i\cdot}}{n}\Big)+\sum_j n_{\cdot j} \log\Big(\frac{n_{\cdot j}}{n}\Big)}.
\end{align*}
\item[]\emph{Adjusted Rand Index} \citep{hubert1985comparing}:
\begin{align*}
\mbox{ARI}=\frac{\sum_{ij} {n_{ij} \choose 2}-\frac{\sum_i {n_{i\cdot}\choose 2} \sum_j {n_{\cdot j}\choose 2}}{{n \choose 2}} }{\frac{1}{2}\sum_i {n_{i\cdot}\choose 2} + \frac{1}{2} \sum_j {n_{\cdot j}\choose 2} - \frac{\sum_i {n_{i\cdot}\choose 2} \sum_j {n_{\cdot j}\choose 2}}{{n \choose 2}} }.
\end{align*}
\end{itemize}
In the above expressions, $n_{i\cdot}$ denotes the true number of nodes in community $i$, $n_{\cdot j}$ represents the number of nodes in the estimated community $j$ and $n_{ij}$ is the number of nodes belonging to community $i$ but estimated to be in community $j$. Both NMI and ARI  are bounded by $1$, with the value of $1$ indicating perfect recovery while $0$ implying the estimation is no better than random guess. See \cite{steinhaeuser2010identifying} for a detailed discussion.

\subsection{Simulation Studies.}

We set $K=2, \rho_n=\frac{3[\log(n)]^{1.5}}{4n}, P(c=1)=P(c=2)=0.5, \bar{B}=\left (
\begin{array}{cc}
1.6 & 0.4 \\
0.4& 1.6
\end{array} \right )$.  We consider the following two different scenarios.
\begin{itemize}
\item[(A).] $p=4, \bm{x} \mid c=1 \sim N(\bmu, {I}_4), \bm{x} \mid c=2 \sim N(-\bmu, {I}_4), \bmu=(0, 0.4, 0.6, 0.8)^T$, where ${I}_4 \in \mathbb{R}^{4\times 4}$ is the identity matrix.
\item[(B).] $p=4, (x_1,x_2) \mid c=1 \sim N(\bmu, \Sigma), (x_1,x_2) \mid c=2 \sim N(-\bmu, \Sigma), \bmu=(0.5, 0.5)^T, \Sigma_{11}=\Sigma_{22}=1, \Sigma_{12}=0.3, x_3 \mid c=1 \sim  \mbox{Bernoulli}(0.6), x_3 \mid c=2 \sim  \mbox{Bernoulli}(0.4), x_4 \mid c=1 \sim \mbox{Uniform}(-0.2, 0.5), x_4 \mid c=2 \sim \mbox{Uniform}(-0.5, 0.2)$; and $(x_1,x_2), x_3$ and $x_4$ are mutually independent. 
\end{itemize}

Note that in Scenario (A), NSBM is the correct model and the first nodal variable is independent of the community assignment; In Scenario (B), NSBM is no longer correct. Under both correct and misspecified models, we would like to: (i) investigate the impacts of the two tuning parameters $(\gamma_n, \lambda_n)$ in the SDP \eqref{sdp:formula}; (ii) examine the effectiveness of the SDP as initialization; (iii) check the performances of variational and maximum profile likelihood methods for utilizing nodal information.

\subsubsection{Tuning Parameters in SDP.}

For both simulation settings, we solve SDP defined in \eqref{sdp:formula} with different tuning parameters via Algorithm 1, with the  number of iterations $\mathcal{T}=100$ and the step size $\xi=1$. We then calculate the NMI of its community detection estimates. Since the ARI gives similar results, we do not show them here for simplicity. The full procedure is repeated 500 times.  

Figure \ref{fig:one} demonstrates the joint impact of the tuning parameters on the SDP performance under Scenario (A). First of all, the comparison of NMI between $\gamma_n=0$ and $\gamma_n>0$ indicates the effectiveness of SDP \eqref{sdp:formula} for leveraging nodal information. We can also see that neither small or large values of $\gamma_n$ lead to optimal performances, verifying the point we discussed in Section \ref{ivco:two} that $\gamma_n$ plays the role of balancing the edge and nodal information. An appropriate choice, as suggested by the four plots, is $\gamma_n = \frac{[\log(n)]^{0.5}}{ n}$, which is consistent with the result of Theorem \ref{thm:sdp}. Regarding the parameter $\lambda_n$, we know from Theorem \ref{thm:sdp} that $\lambda_n=\frac{n^2}{2}$ is the desired choice. Interestingly, Figure \ref{fig:one} shows that a wide range of $\lambda_n$ can give competitive results, as long as the corresponding $\gamma_n$ is properly chosen. For Scenario (B), similar phenomena can be observed in Figure \ref{fig:two}. Note that since the nodal covariates are not as informative as in Scenario (A), the optimal $\gamma_n \approx \frac{0.8[\log(n)]^{0.5}}{ n}$ tends to give more weights to the adjacency matrix.  The results in these two different settings  confirm the implication of Theorem \ref{thm:sdp}, that SDP \eqref{sdp:formula} can work beyond NSBM. 

\begin{figure}[htb]
\centering
\setlength\tabcolsep{1.5pt}
\begin{tabular}{cc}
\includegraphics[width=5.6cm, height=4.8cm]{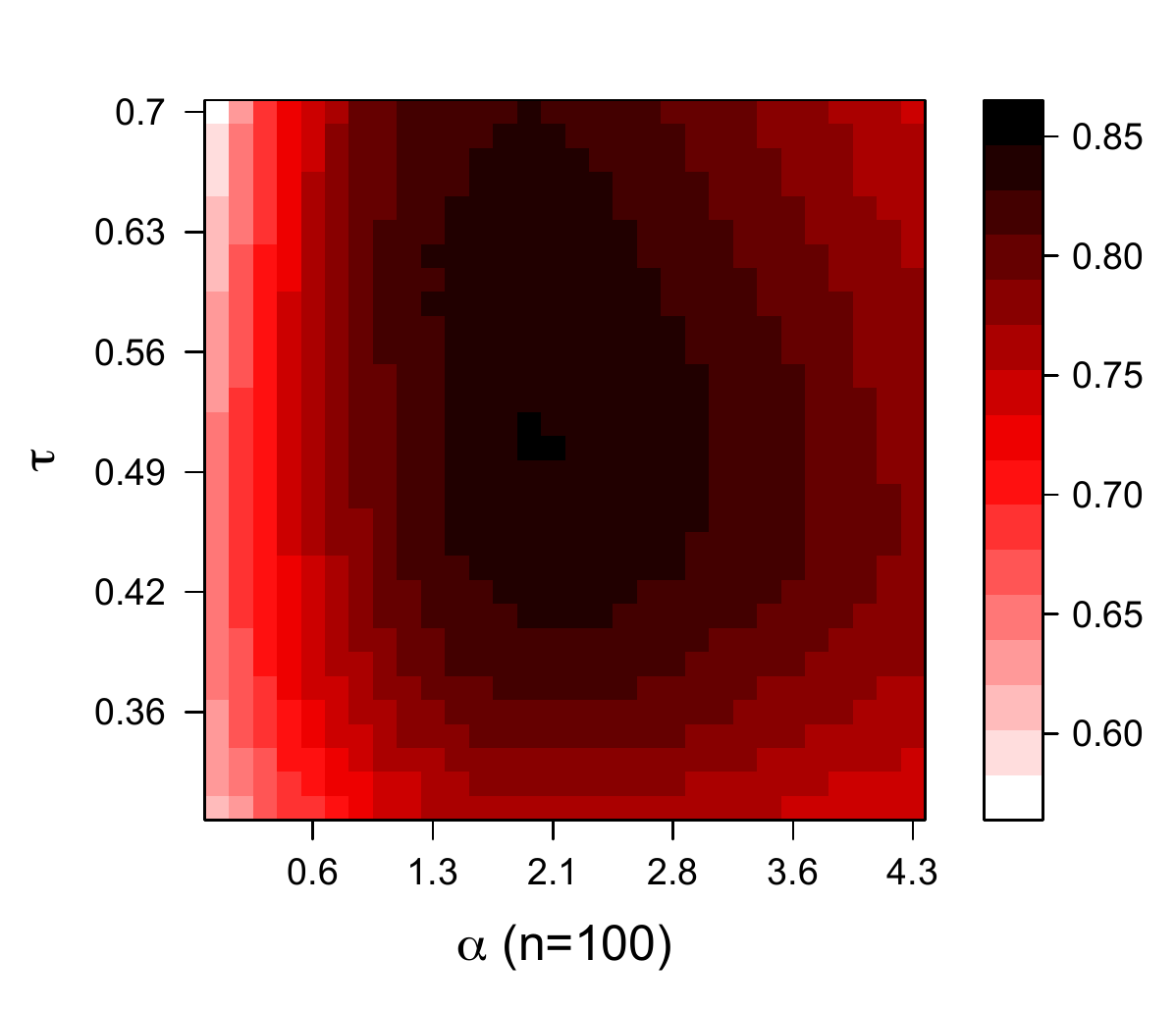} &
\includegraphics[width=5.6cm, height=4.8cm]{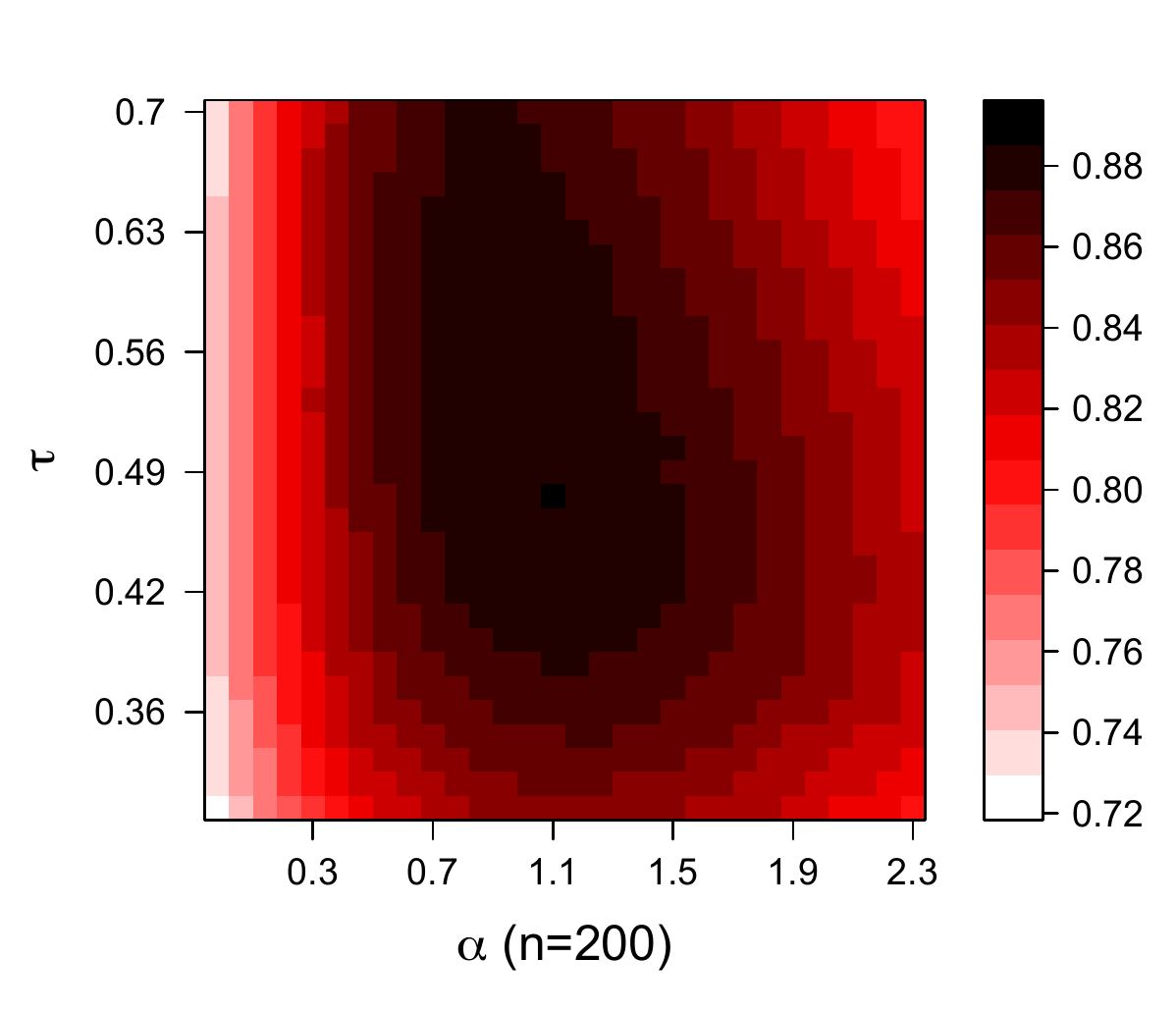} \\ [-0.5cm]
\includegraphics[width=5.6cm, height=4.8cm]{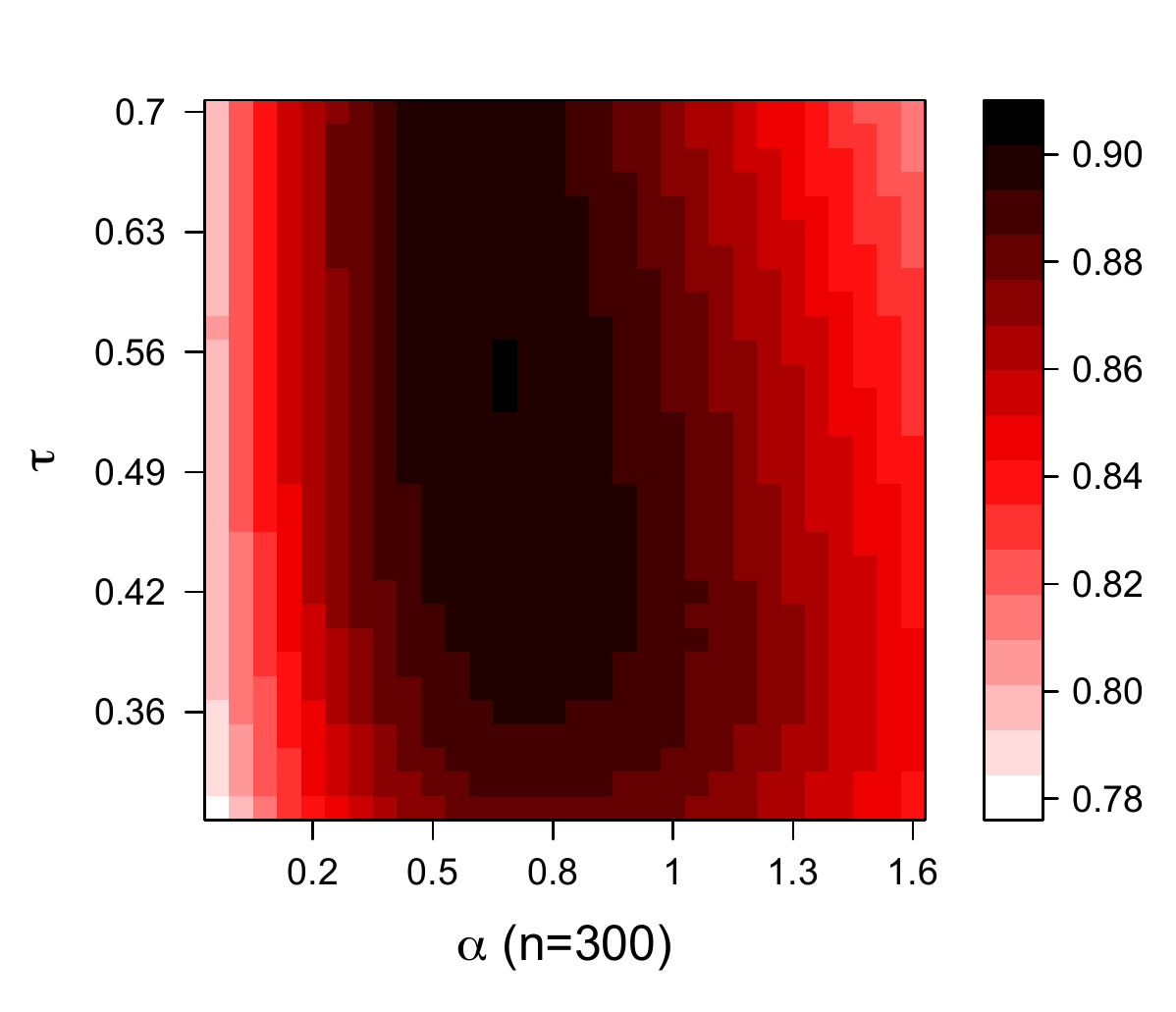} &
\includegraphics[width=5.6cm, height=4.8cm]{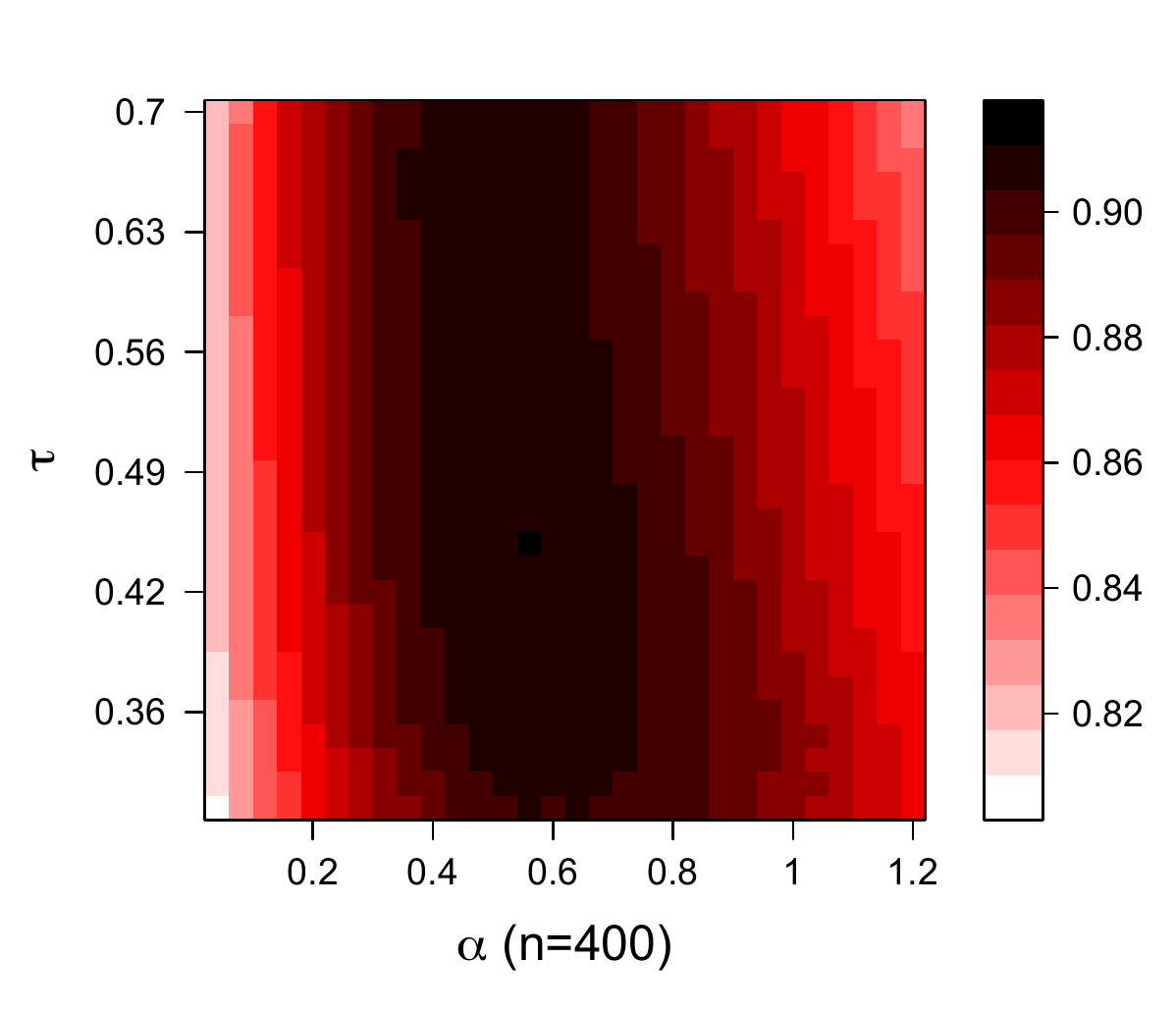} 
\end{tabular}
\vspace{-0.3cm}
\caption{The community detection performance of SDP (measured by NMI), under Scenario (A), with different tuning parameters $(\lambda_n, \gamma_n)$; NMI is averaged over 500 repetitions; We have used the scaled version of the tuning parameters: $\tau= \frac{\lambda_n}{n^2}, \alpha=100\gamma_n$.} \label{fig:one}
\end{figure}

\begin{figure}[htb]
\centering
\setlength\tabcolsep{1.5pt}
\begin{tabular}{cc}
\includegraphics[width=5.6cm, height=4.8cm]{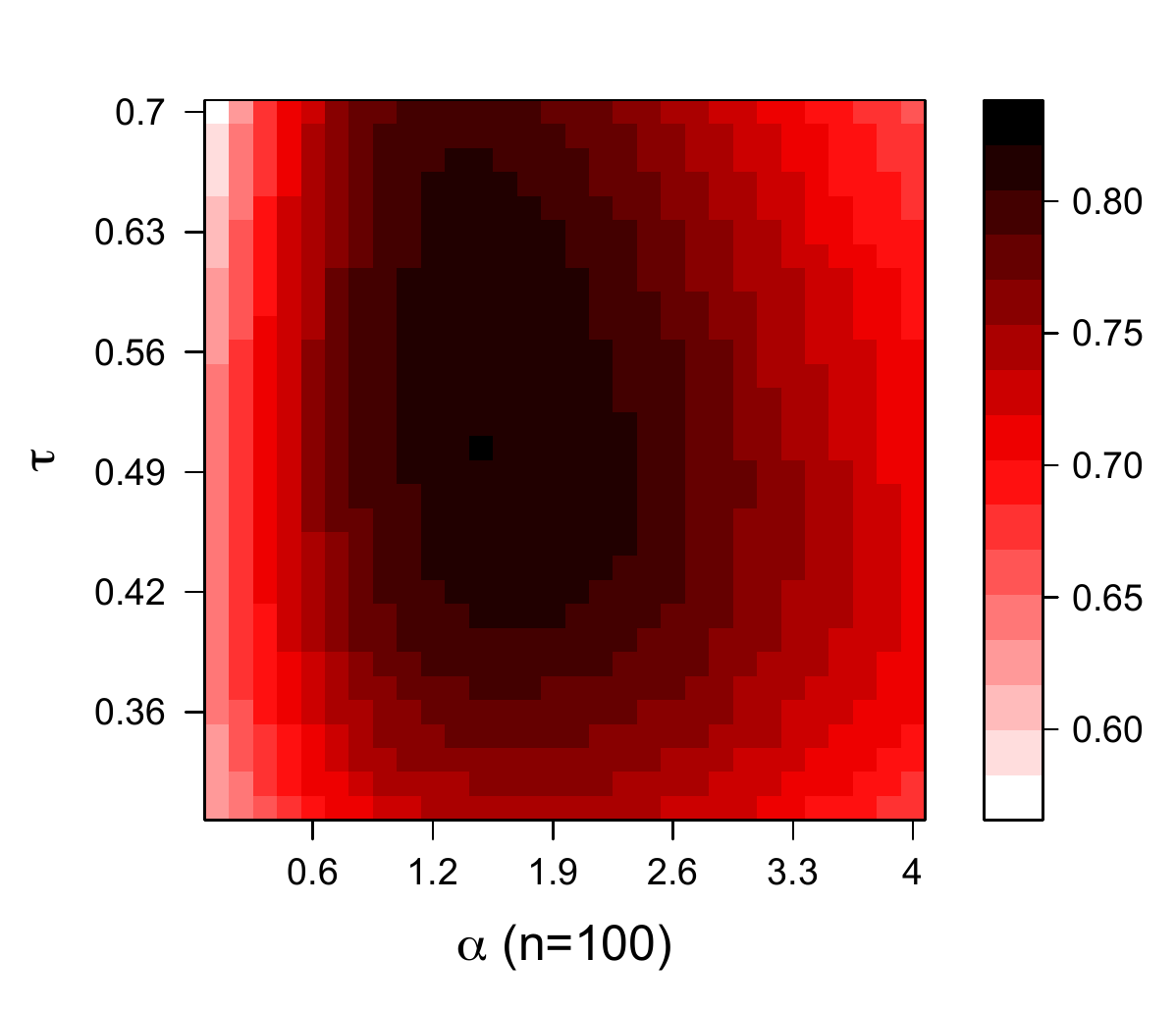} &
\includegraphics[width=5.6cm, height=4.8cm]{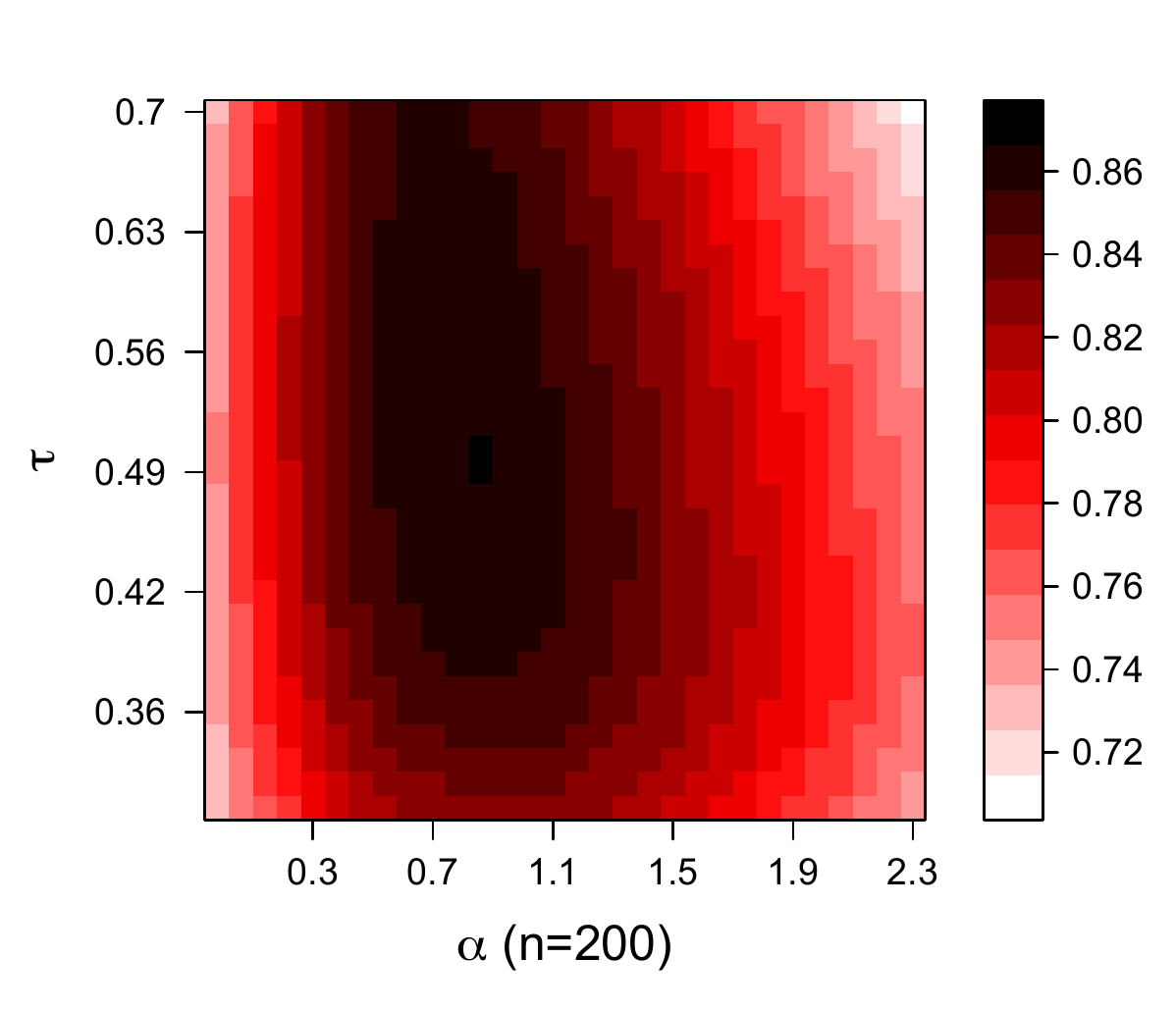} \\ [-0.5cm]
\includegraphics[width=5.6cm, height=4.8cm]{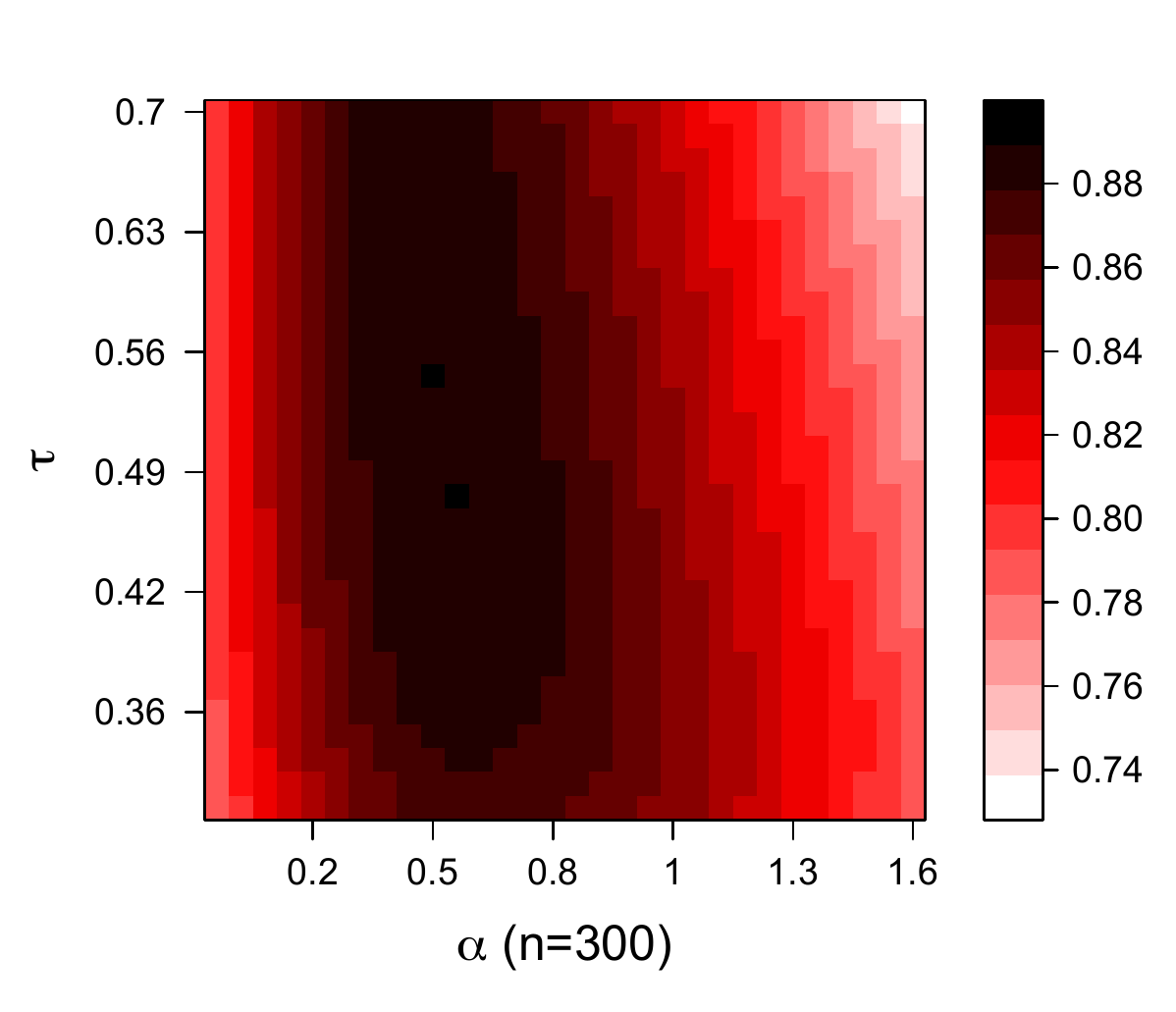} &
\includegraphics[width=5.6cm, height=4.8cm]{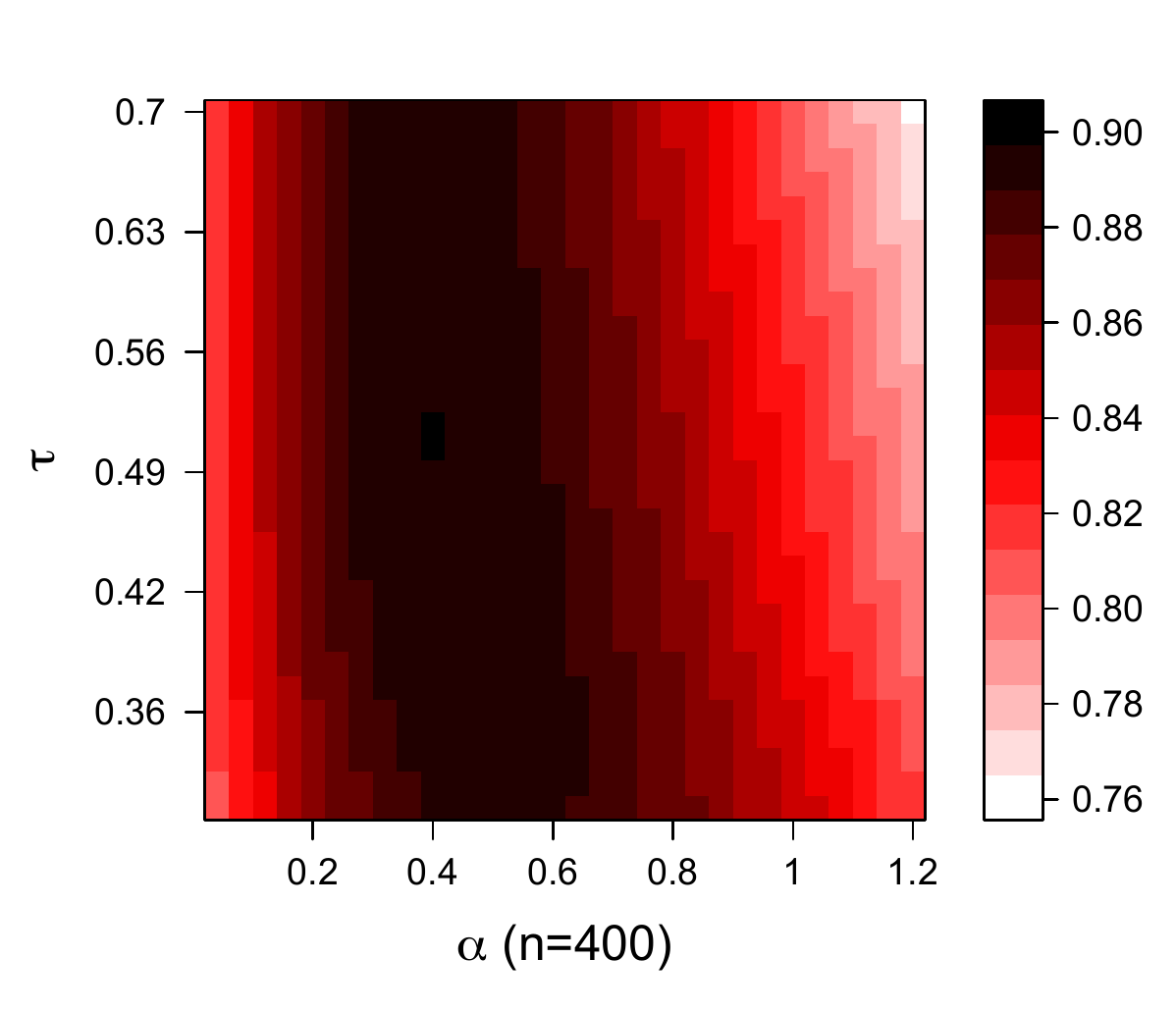} 
\end{tabular}
\vspace{-0.3cm}
\caption{The community detection performance of SDP (measured by NMI), under Scenario (B), with different tuning parameters $(\lambda_n, \gamma_n)$; NMI is averaged over 500 repetitions; We have used the scaled version of the tuning parameters: $\tau= \frac{\lambda_n}{n^2}, \alpha=100\gamma_n$.} \label{fig:two}
\end{figure}

\subsubsection{Community Detection Performance via Variational and Maximum Profile Likelihood Methods.}

We implement variational and maximum profile likelihood methods via Algorithms 2 and 3 respectively, taking the outputs from Algorithm 1 as initialization (called VEM-C and MPL-C respectively). We do not predefine the number of iterations $\mathcal{T}$ in both algorithms, and instead keep iterating until convergence. To investigate the impact of SDP as an initialization, we have additionally implemented both methods with random initialization (called VEM-B and MPL-B respectively): run Algorithms 2 and 3 with random initialization independently multiple times and choose the outputs that give the largest objective function value (e.g., the profile likelihood function). We have also applied both methods for the simulated datasets with nodal attributes removed (called VEM-A and MPL-A respectively). This will be used as a comparison to check the effect of the two methods in incorporating nodal information. We set $\lambda_n=\frac{1}{2}n^2, \gamma_n=\frac{[\log(n)]^{0.5}}{ n}$ for all the implementations of SDP under Scenario (A); and $\lambda_n=\frac{1}{2}n^2, \gamma_n=\frac{0.8 [\log(n)]^{0.5}}{ n}$ under Scenario (B).

Figure \ref{fig:three} shows the community detection results of both methods under Scenario (A). By comparing the four curves in each plot, we can make a list of interesting observations: (1) SDP is a good initialization (MPL-C vs. MPL-B, VEM-C vs. VEM-B); (2) SDP itself already gives reasonable outputs, but the follow-up iterations further improve the performance (SDP vs.. MPL-C, SDP vs. VEM-C); (3) the nodal covariates are helpful for detecting communities, and the two methods have made effective use of it (MPL-A vs. MPL-C, VEM-A vs. VEM-C); (4) the two methods have similar performances when initialized with SDP (MPL-A vs. VEM-A, MPL-C vs. VEM-C). Moreover, we would like to point out the different behavior of the two methods with random initialization. The comparison between the two purple curves (MPL-B vs. VEM-B) implies that compared to the variational method, maximum profile likelihood method has the potential of exploring the parameter space more efficiently, especially when the sample size is large. One possible explanation is that the update of the ``soft" community labels (the distribution $\{q_{ik}\}$) in the variational algorithm may cause it to move very slowly in the parameter space and hence it may take many steps to change a label assignment. Furthermore, note that we can use the asymptotic normality property of the estimators for $\bm{\beta}$ in Theorems \ref{thm:mfvi} and \ref{thm:pmle} to perform variable selection. The results of the Wald test regarding each component of $\bm{\beta}$ are presented in Figure \ref{fig:add}. We see that our methods are able to identify relevant (the last three) and irrelevant (the first) nodal variables.
Regarding Scenario (B), similar observations on the community detection performance can be made from Figure \ref{fig:four}. We thus omit the details. As a final remark, the performances in Scenario (B) indicate that both methods can work to a certain extent of model misspecification.

\begin{figure}[htb]
\centering
\setlength\tabcolsep{-15pt}
\begin{tabular}{cc}
\includegraphics[scale=0.7]{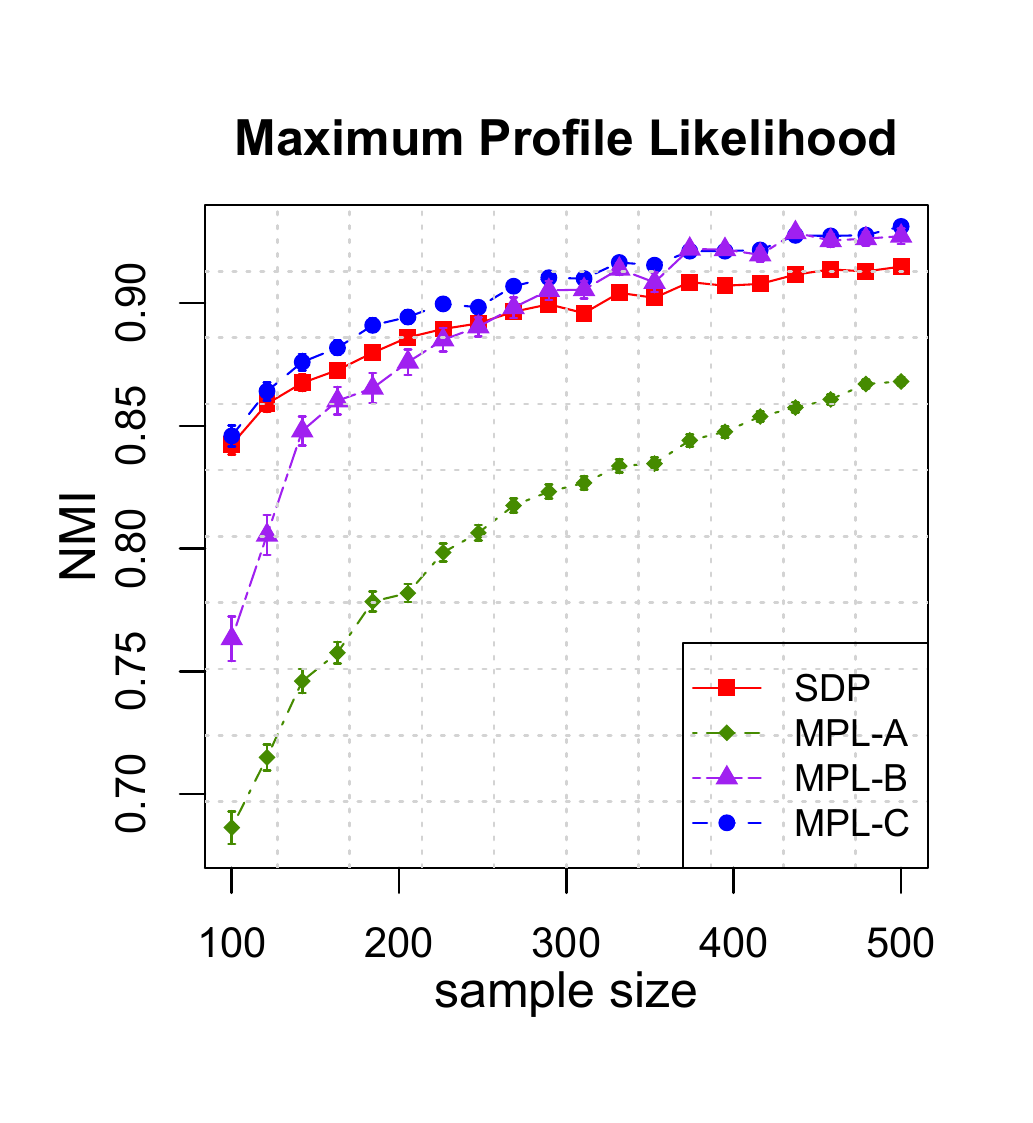} &
\includegraphics[scale=0.7]{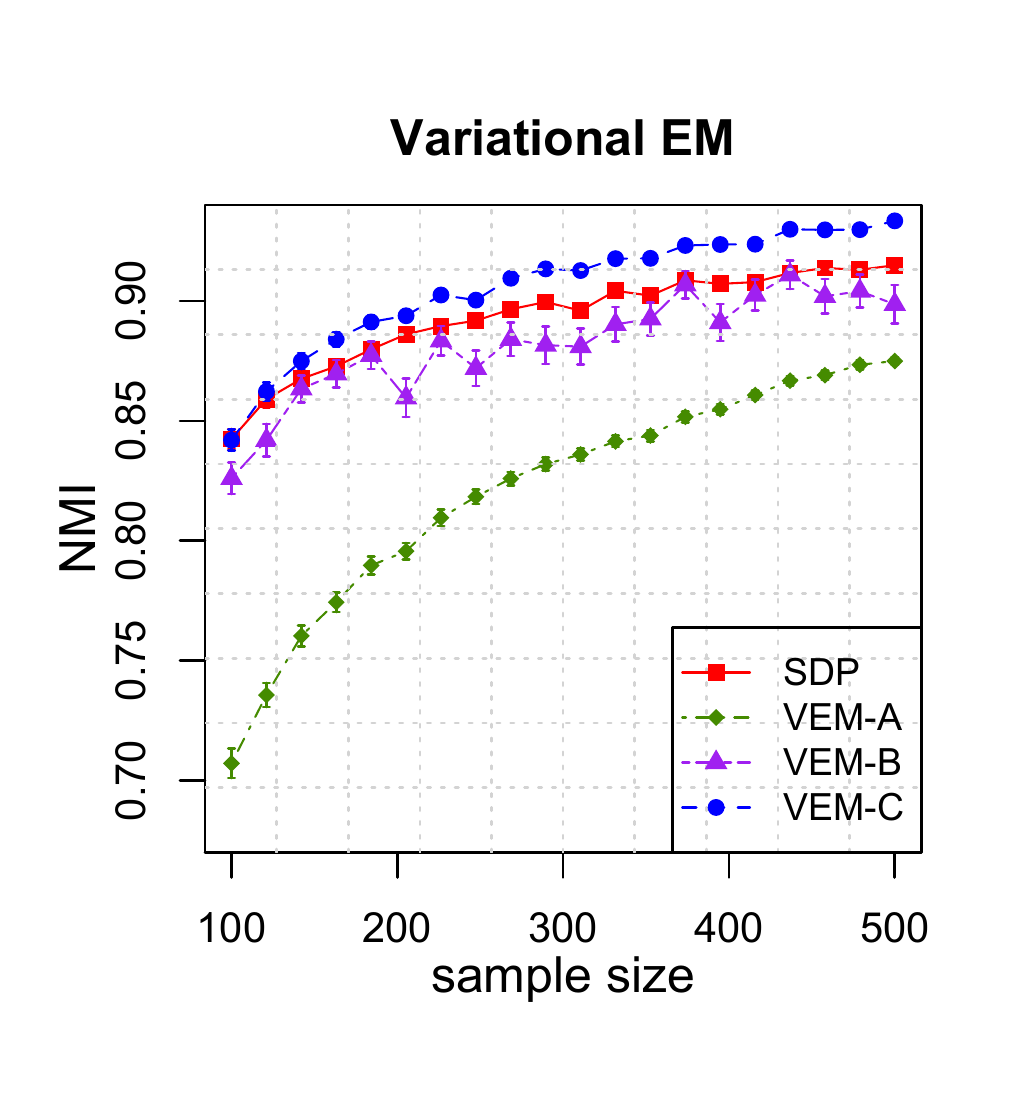} 
\end{tabular}
\vspace{-0.5cm}
\caption{The community detection performances under Scenario (A); The average NMI is calculated over 500 repetitions along with its standard error bar; MPL-A, MPL-B and MPL-C represent the maximum profile likelihood methods with no nodal covariates being used, random initialization, and initialization from SDP, respectively; Similar notations are used for variational method. We have used 15 independent random initializations for maximum profile likelihood method across all the sample size; For variational method, the number of random initializations used starts from 15 for $n=100$ and consecutively increases by 1 for the subsequent sample sizes. } \label{fig:three}
\end{figure}

\begin{figure}[htb]
\centering
\setlength\tabcolsep{-13pt}
\begin{tabular}{cc}
\includegraphics[width=7cm, height=7cm]{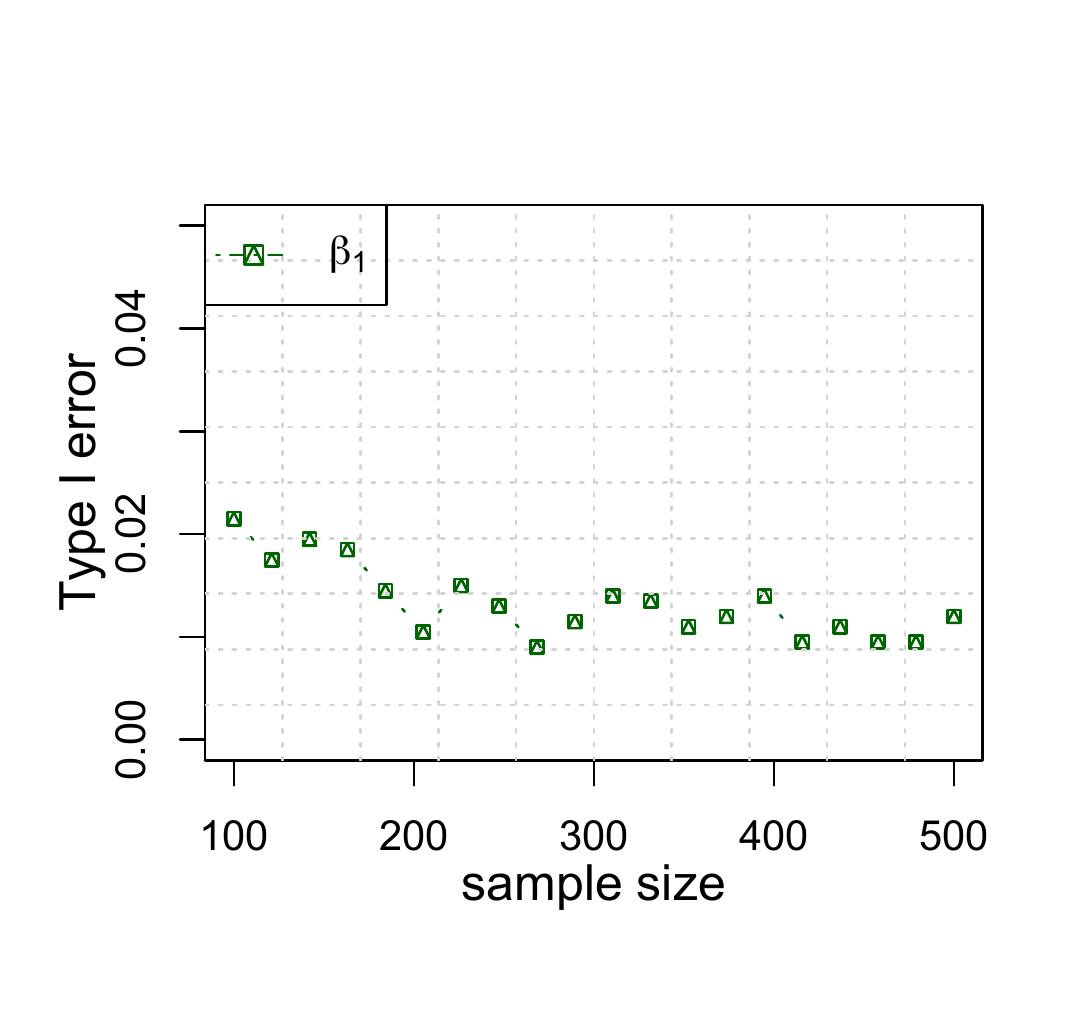} &
\includegraphics[width=7cm, height=7cm]{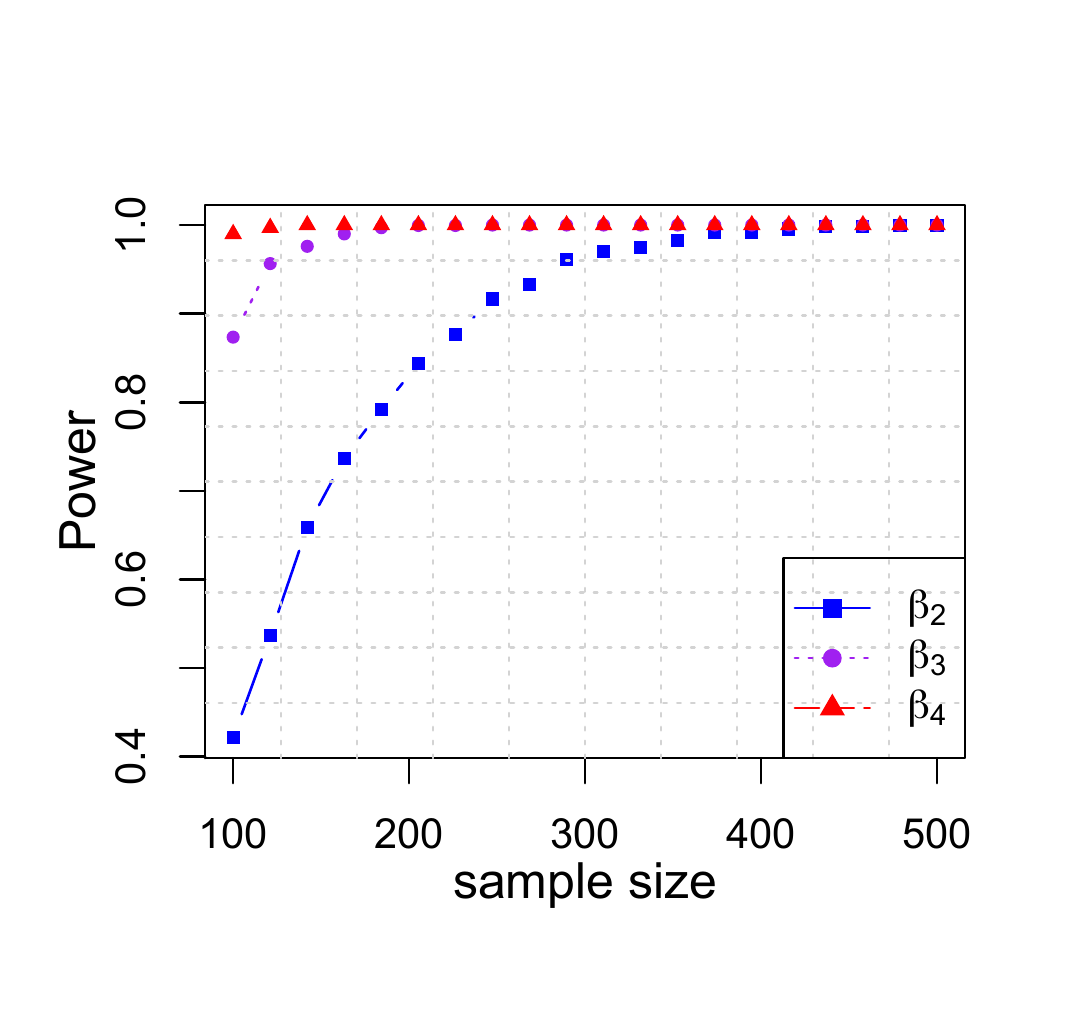} 
\end{tabular}
\vspace{-0.7cm}
\caption{Wald test for each component of $\beta$.  The calculations are averaged over 2000 repetitions. The significance level is set to 0.01. Since both variational and maximum profile likelihood methods give similar results, we only present the result of variational method  for simplicity. } \label{fig:add}
\end{figure}

\begin{figure}[htb]
\centering
\setlength\tabcolsep{0.pt}
\begin{tabular}{cc}
\includegraphics[scale=0.7]{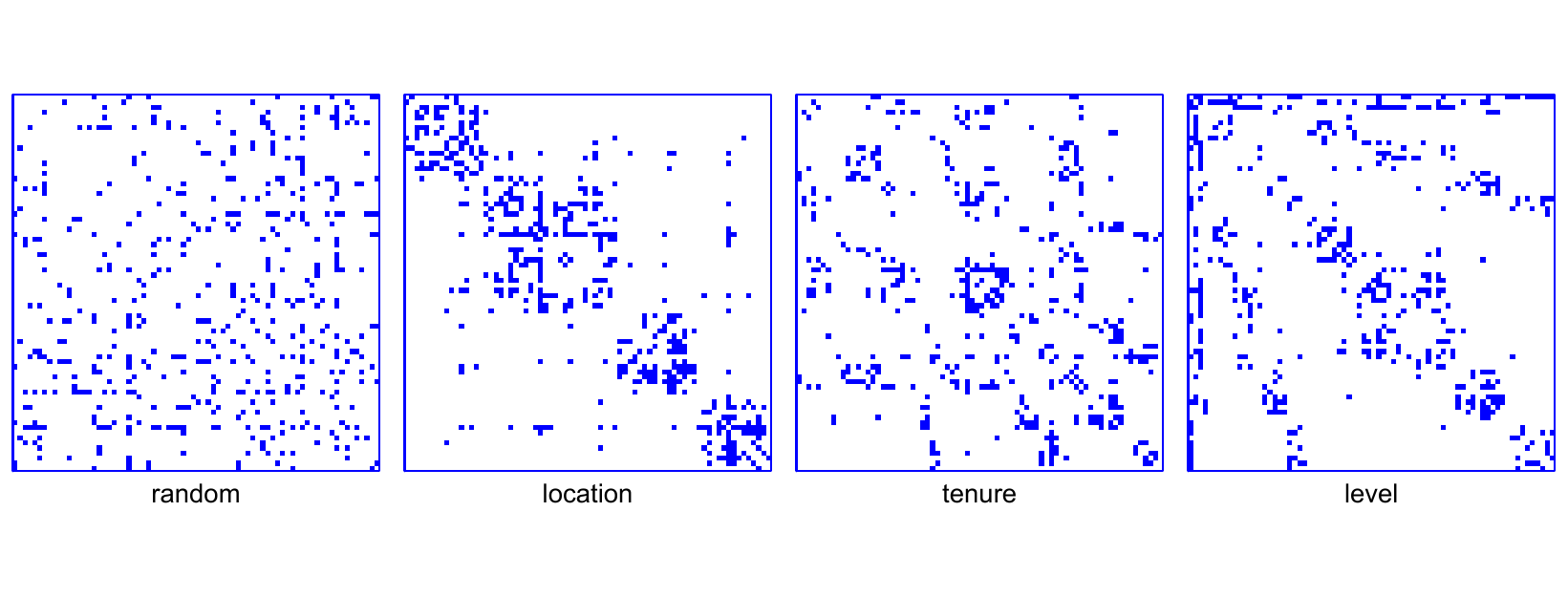} 
\end{tabular}
\vspace{-0.7cm}
\caption{From left to right are the re-ordered adjacency matrices based on random permutation, location, tenure, and organizational level.  } \label{fig:five}
\end{figure}

\begin{figure}[htb]
\centering
\setlength\tabcolsep{-15pt}
\begin{tabular}{cc}
\includegraphics[scale = 0.7]{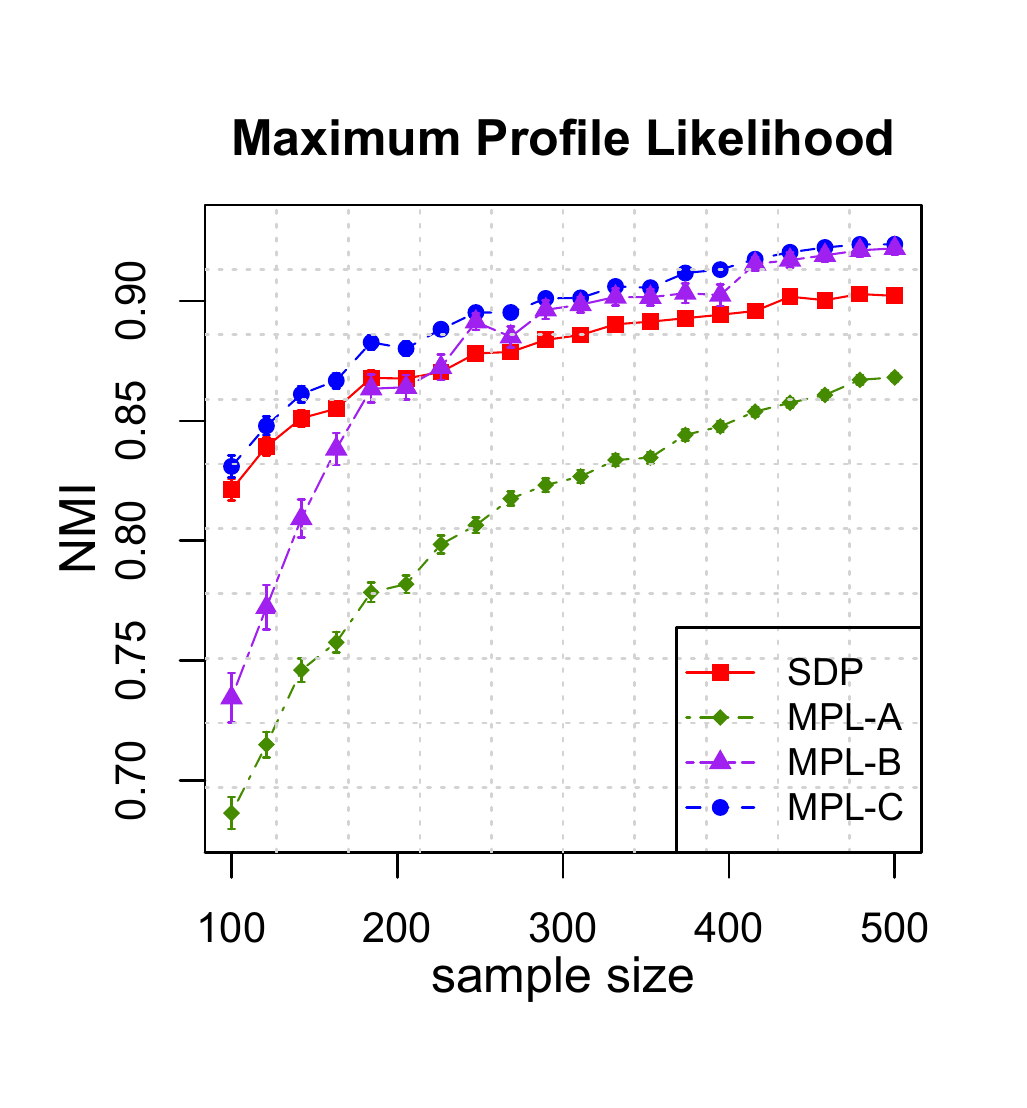} &
\includegraphics[scale = 0.7]{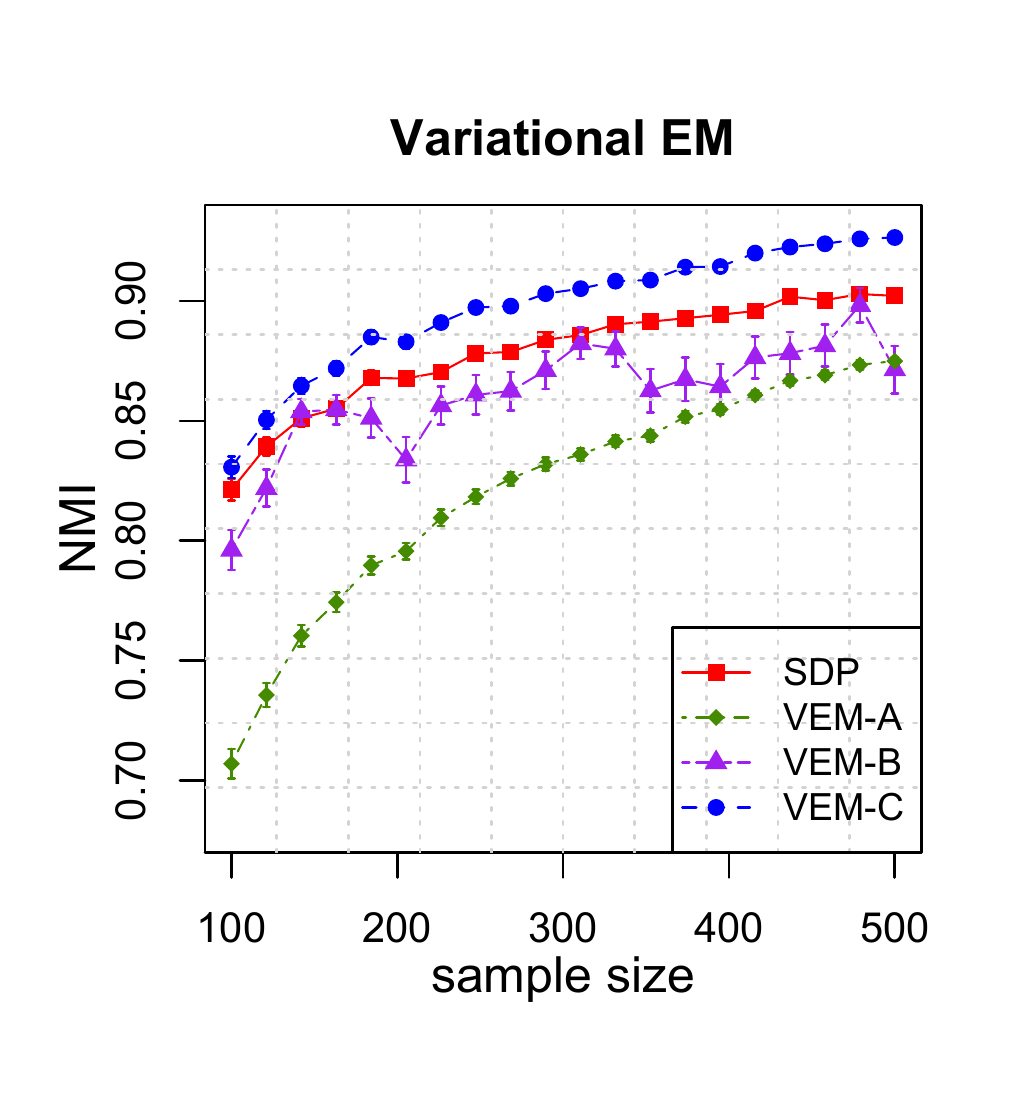} 
\end{tabular}
\vspace{-0.5cm}
\caption{The community detection performances under Scenario (B); All the relevant descriptions are the same as in Figure \ref{fig:three}.} \label{fig:four}
\end{figure}

\subsection{Real Data Analysis.}

The dataset is about a research team consisting of 77 employees in a manufacturing company \citep{cross2004hidden}. A weight $w_{ij}$ is assigned to the edge from employee $i$ to employee $j$ based on the extent to which employee $i$ provides employee $j$ with information $j$ uses to accomplish $j$'s work. There are seven choices for the weights: 0 (I do not know this person/I have never met this person); 1 (Very infrequently); 2 (Infrequently); 3 (Somewhat infrequently); 4 (Somewhat frequently); 5 (Frequently); 6 (Very frequently). In addition to the edge information, the dataset also contains several attributes of each employee: location (1: Paris, 2: Frankfurt, 3: Warsaw, 4: Geneva); tenure (1: 1-12 months, 2: 13-36 months, 3: 37-60 months, 4: 61+ months); the organizational level (1: Global Dept Manager, 2: Local Dept Manager, 3: Project Leader, 4: Researcher). Since the network is a weighted and directed network, we first convert it to a binary network such that there exists an edge from $i$ to $j$ if and only if $w_{ij}>3$. This corresponds to whether the information is provided frequently or not. We then further convert it into an undirected network in the way that the edge between $i$ and $j$ exists if and only if both directed edges from $i$ to $j$ and $j$ to $i$ are present. Finally, we remove three isolated nodes from the network. To explore the intro-organizational community structure, we re-order the adjacency matrix based on random permutation and the attributes. As can be seen from Figure \ref{fig:five}, the attribute ``location" is a very informative indicator of the network's community structure. This should not come as a big surprise, since the same office location usually promotes communication and collaboration between team members. We now use the ``location" as the ground truth for the community assignment and examine the performances of SDP \eqref{sdp:formula}, maximum profile likelihood and variational methods based on the rest of the data. For SDP \eqref{sdp:formula}, we first use spectral clustering on adjacency matrix \citep{lei2014consistency} to estimate the size of the communities and plug the estimates in the formula of $\lambda_n$ in Theorem \ref{thm:sdp} to determine $\lambda_n$. Regarding $\gamma_n$, motivated from the simulation results, we choose $\gamma_n=\frac{\hat{\rho}_n}{\log n}$, where $\hat{\rho}_n=\frac{2\times \mbox{number of edges}}{n^2}$. The maximum profile likelihood and variational methods are initialized by the output from SDP. We can see from Table \ref{tab:one} that by incorporating the nodal information, community detection accuracy has been improved. It is interesting to observe that SDP performs as well as the two likelihood based methods, when nodal covariates are available. Note that we can calculate the mutual information between the ``ground truth" variable ``location" and the other two to see how much community information they contain. Given that both mutual information (0.11 \& 0.03) are pretty small, the magnitude of improvement in Table \ref{tab:one} is reasonable.

\begin{table}[!htb]
\begin{center}
\begin{tabular}{c|ccc|ccc}
\hline
&SDP & MPL & VEM &SDP&MPL&VEM \\
\hline
edge&0.881&0.894&0.894&0.882&0.883&0.883  \\
\hline
edge + nodal &0.920&0.920 &0.920&0.921&0.921&0.921 \\
\hline
\end{tabular}
\end{center}
\caption{The community detection results of SDP, maximum profile likelihood, and variational methods. MPL and VEM denote maximum profile likelihood and variational methods, respectively. NMI is computed on the left part of the table, and ARI on the right. The row indexed by ``edge" shows the results based on the network without nodal information, while the other one ``edge+nodal" contains the results of making use of the two attributes available. } \label{tab:one}
\end{table}

\section{Discussion.}

In this paper, we present a systematic study of the community detection with nodal information problem. We propose a flexible network modeling framework, and analyze three likelihood based methods under a specialized model. Both asymptotic and algorithmic aspects have been thoroughly discussed. The superiority of variational and maximum profile likelihood methods are verified through a variety of numerical experiments. Finally, we would like to highlight several potential extensions and open problems for future work.
\begin{itemize}
\item[1.] The modeling of both the network and nodal covariates can be readily extended to more general families, such as degree-corrected stochastic block model and non-parametric regression, respectively. The corresponding asymptotic results might be derived accordingly.
\item[2.] In the setting with high dimensional covariates, penalized likelihood methods are more appealing for both community detection and variable selection. Theoretical analysis of community detection and variable selection consistency will be necessary.
\item[3.] For very sparse networks, considering $n\rho_n=O(1)$ seems to be a more realistic asymptotic framework. Under such asymptotic setting, community detection consistency is impossible. The effect of nodal covariates becomes more critical. It is of great interest to characterize the impact of the nodal information on community detection. 
\item[4.] In this work, we assume the number of communities $K$ is known. How to select $K$ is an important problem in community detection. Some recent efforts  towards this direction include \cite{saldana2015many, le2015estimating, wang2015likelihood, lei2016goodness}.
\end{itemize}

\appendix
\section{Appendix}

\noindent \emph{\textbf{Notations and Preliminaries}}. Before the proofs, we first introduce some necessary notations. Let $\pi_a=P(c=a), \hat{\pi}_a=\frac{1}{n}\sum_{i=1}^n\mathbbm{1}(c_i=a), 1\leq a \leq K$. Given a community assignment $\bm{e} \in \{1,\dots, K\}^n$, define $O(\bm{e}),V(\bm{e}),T(\bm{e}), \hat{T}(\bm{e}) \in \mathbb{R}^{K \times K}$ and $\bm{f}^0(\bm{e}),   \hat{\bm{f}}(\bm{e}) \in \mathbb{R}^K$ with their elements being
\begin{align*}
&O_{ab}(\bm{e})=\sum_{ij}A_{ij}\mathbbm{1}(e_i=a, e_j=b), \quad V_{ab}(\bm{e})=\frac{\sum_i \mathbbm{1}(e_i=a, c_i=b)}{\sum_i \mathbbm{1}(c_i=b)}, \\
&T_{kl}(\bm{e})=\sum_{ab}\pi_a\pi_b\bar{B}_{ab}V_{ka}(\bm{e})V_{lb}(\bm{e}), \quad \hat{T}_{kl}(\bm{e})=\sum_{ab}\hat{\pi}_a\hat{\pi}_b\bar{B}_{ab}V_{ka}(\bm{e})V_{lb}(\bm{e}), \\
&f^0_k(\bm{e})=\sum_a V_{ka}(\bm{e})\pi_a, \quad \hat{f}_k(\bm{e})=\sum_a V_{ka}(\bm{e})\hat{\pi}_a = \frac{n_k(\bm{e})}{n}.
\end{align*}
Here, $O_{ab}(\bm{e})$ represents the number of edges between communities $a$ and $b$ under assignment $\bm{e}$ (twice the number when $a=b$), $V_{ab}(\bm{e})$ represents the proportion of the nodes in community $b$ under $\bm{c}$ that are mislabeled as community $a$ under $\bm{e}$, $T_{kl}(\bm{e})$ reflects the connection probability (up to a scaling factor $\rho_n$) between a node in community $k$ and another node in community $l$ under $\bm{e}$ with $\hat{T}_{kl}(\bm{e})$ being its empirical version, and $\hat{f}_k(\bm{e})$ represents the proportion of community $k$ under $\bm{e}$ with $f^0_k(\bm{e})$ as its ``population" version. Also denote $n_a(\bm{e})=\sum_{i}\mathbbm{1}(e_i=a), 1\leq a \leq K; F(T, \bm{f} )=\sum_{ab}\Big( T_{ab}\log  \frac{T_{ab}}{f_af_b}-T_{ab} \Big); \mathcal{H}(T,\bm{f})= \sum_{ab} \Big(T_{ab}\log \bar{B}_{ab}-f_af_b\bar{B}_{ab}\Big); \mathcal{V}=\{V \in \mathbb{R}^{K \times K}: \sum_{k}V_{ka}=1,   V_{ka} \geq 0, 1 \leq k, a \leq K\}; \mu_n=n^2\rho_n$. Throughout the proofs, we use $B^0$ and $\bm{\beta}_0$ to represent the true parameters in NSBM while leaving $B$ and $\bm{\beta}$ as the generic parameter. We will $C_1$, $C_2$, $\dots$, to represent positive generic constants whose values may vary across different lines. We will frequently use the notation $\|V(\bm{e})-V(\bm{c})\|_1=\sum_{ab}|V_{ab}(\bm{e})-V_{ab}(\bm{c})|$. Note that $\frac{1}{n}\sum_i \mathbbm{1}(e_i \neq c_i) \leq \frac{1}{2}\|V(\bm{e})-V(\bm{c})\|_1$ (see the derivation on Page 22 of \cite{zhao2012consistency}).  Moreover, we cite two concentration inequality results that will be used in the proof. The first one is Lemma A.1 from \cite{zhao2012consistency}:
\begin{align}\label{concen:one}
P\Big(\max_{\bm{e}}\max_{ab} \Big |\frac{O_{ab}(\bm{e})}{\mu_n}-\hat{T}_{ab}(\bm{e}) \Big | \geq \epsilon \Big) \leq 2K^{n+2}\exp({-\frac{\epsilon^2\mu_n}{8\max_{ab}\bar{B}_{ab}}}), 
\end{align}
for $\epsilon < 3\max_{ab}\bar{B}_{ab}$. The second one is (1.4) in \cite{bickel2015correction}: $\forall~\gamma_0 >0,$
\begin{align}\label{concen:two}
P\Big(\max_{0<\|V(\bm{e})-V(\bm{c})\|_1\leq \delta_n}[\max_{ab}|X_{ab}(\bm{e})-X_{ab}(\bm{c})|-\gamma_0\cdot \|V(\bm{e})-V(\bm{c})\|_1 ] \leq 0\Big ) \rightarrow 1, 
\end{align}
where $X_{ab}(\bm{e})=\frac{O_{ab}(\bm{e})}{\mu_n}-\hat{T}_{ab}(\bm{e})$ and $\delta_n \rightarrow 0, n\rho_n/\log n \rightarrow \infty$.
Finally, we should be aware that the expressions involving community assignment $\bm{e}$ in the proofs are to be interpreted,  up to permutations of community labels in $\{1,2,\dots, K\}$ whenever necessary.

\vspace{0.3cm}
\begin{lemma}\label{strong:curv}
There exist positive constants $c_0, c_1, c_2, c_3>0$, such that
\begin{align*}
&F(T(\textbf{c}), \textbf{f}^0(\textbf{c}))-F(T(\textbf{e}), \textbf{f}^0(\textbf{e})) \geq c_0 \cdot \|V(\textbf{e})-V(\textbf{c})\|_1, \mbox{~~~if~} \|V(\textbf{e})-V(\textbf{c})\|_1 \leq c_1,\\
&\mathcal{H}(T(\textbf{c}), \textbf{f}^0(\textbf{c}))-\mathcal{H}(T(\textbf{e}), \textbf{f}^0(\textbf{e})) \geq c_2 \cdot \|V(\textbf{e})-V(\textbf{c})\|_1, \mbox{~~~if~} \|V(\textbf{e})-V(\textbf{c})\|_1 \leq c_3.
\end{align*}
\end{lemma}
\begin{proof}
We only show the proof for the first inequality since the second one can be derived in a similar way. Note that $F(T(\textbf{e}), \textbf{f}^0(\textbf{e}))$ can be considered as a function of $V(\textbf{e})$. We give it another notation $H(V(\textbf{e}))$, where $H(\cdot)$ is defined on the convex set $\mathcal{V}$. Further define $g(\epsilon; V(\textbf{e}))=H((1-\epsilon)V(\textbf{c})+\epsilon V(\textbf{e}))$, for $0\leq \epsilon \leq 1$. Since $V(\bm{e}), V(\bm{c})\in \mathcal{V}$, $g(\epsilon; V(\textbf{e}))$ is well defined. We first show that $\exists~ \tilde{c}_0, \tilde{c}_1>0$, s.t.
\begin{align}\label{key:point}
g'(\epsilon; V(\textbf{e})) \leq -\tilde{c}_0\cdot \|V(\textbf{e})-V(\textbf{c})\|_1, \mbox{~for any~} 0\leq \epsilon \leq \tilde{c}_1, V(\textbf{e}) \in \mathcal{V} 
\end{align}
where $g'(\epsilon;V(\textbf{e}))$ is the derivative with respect to $\epsilon$; the constants $\tilde{c}_0, \tilde{c}_1$ do not depend on $V(\textbf{e})$. To prove \eqref{key:point}, denote $\tilde{V}_{ka}=(1-\epsilon)V_{ka}(\textbf{c})+\epsilon V_{ka}(\textbf{e}), r_{kl}=\frac{\sum_{ab}\pi_a\pi_b\tilde{V}_{ka}\tilde{V}_{lb}\bar{B}_{ab}}{\sum_{ab}\pi_a\pi_b\tilde{V}_{ka}\tilde{V}_{lb}}$. Then a straightforward calculation gives us,
\begin{align}\label{g:general}
&g'(\epsilon; V(\textbf{e}))=2\epsilon \sum_{kl} \sum_{ab} \pi_a\pi_b \cdot [V_{ka}(\textbf{e})-V_{ka}(\textbf{c})]\cdot [V_{lb}(\textbf{e})-V_{lb}(\textbf{c})] \cdot (\bar{B}_{ab}\log r_{kl}-r_{kl})\\
&+\sum_{kl} \sum_{ab} \pi_a\pi_b([V_{ka}(\textbf{e})-V_{ka}(\textbf{c})]V_{lb}(\textbf{c})+[V_{lb}(\textbf{e})-V_{lb}(\textbf{c})]V_{ka}(\textbf{c}))\cdot (\bar{B}_{ab}\log r_{kl}-r_{kl}). \nonumber
\end{align}
Hence when $\epsilon=0$, the above equation can be simplified as 
\begin{align}\label{gzero}
g'(0; V(\textbf{e}))&=\sum_{ka}\sum_{lb} \pi_a\pi_b \cdot [V_{ka}(\textbf{e})-V_{ka}(\textbf{c})]\cdot V_{lb}(\textbf{c})\cdot (\bar{B}_{ab} \log\bar{B}_{kl}-\bar{B}_{kl}) \\
&\quad +\sum_{lb}\sum_{ka} \pi_a\pi_b \cdot [V_{lb}(\textbf{e})-V_{lb}(\textbf{c})]\cdot V_{ka}(\textbf{c})\cdot (\bar{B}_{ab} \log\bar{B}_{kl}-\bar{B}_{kl}) \nonumber \\
&=\sum_{ka} \sum_b \pi_a\pi_b \cdot (V_{ka}(\textbf{e})-V_{ka}(\textbf{c}))\cdot (\bar{B}_{ab}\log \bar{B}_{kb}-\bar{B}_{kb}) \nonumber \\
&\quad + \sum_{lb} \sum_a \pi_a\pi_b \cdot (V_{lb}(\textbf{e})-V_{lb}(\textbf{c}))\cdot (\bar{B}_{ab}\log \bar{B}_{al}-\bar{B}_{al}) \nonumber \\
&=2\sum_{abl} \pi_a\pi_b(V_{lb}(\textbf{e})-V_{lb}(\textbf{c}))\cdot (\bar{B}_{ab}\log \bar{B}_{al}-\bar{B}_{al}). \nonumber
\end{align}
Note that for any $V(\textbf{e})\in \mathcal{V}$, it holds that $V_{bb}(\textbf{e})-1=-\sum_{l \neq b} V_{lb}(\textbf{e})$. We can then continue the calculation from \eqref{gzero}:
\begin{align}
&~g'(0; V(\textbf{e}))\label{final:zero} \\
=&~2\sum_{ab}\Big [\pi_a\pi_b\cdot (V_{bb}(\textbf{e})-1)\cdot (\bar{B}_{ab}\log\bar{B}_{ab}-\bar{B}_{ab})+\sum_{l \neq b} \pi_a\pi_b \cdot V_{lb}(\textbf{e})\cdot (\bar{B}_{ab}\log \bar{B}_{al}-\bar{B}_{al}) \Big]  \nonumber \\
=&~2\sum_{ab}\sum_{l\neq b} \pi_a\pi_b \cdot V_{lb}(\textbf{e})\cdot[\bar{B}_{ab}\log \bar{B}_{al}-\bar{B}_{al}- \bar{B}_{ab}\log\bar{B}_{ab}+\bar{B}_{ab}]  \nonumber \\
\overset{(a)}{=}&~2\sum_{ab}\sum_{l\neq b} \pi_a\pi_b \cdot V_{lb}(\textbf{e})\cdot \frac{-\bar{B}_{ab}(\bar{B}_{al}-\bar{B}_{ab})^2}{2\tilde{B}^2_{ab}} \overset{(b)}{\leq} -\tilde{c}_2 \cdot \sum_{ab}\sum_{l\neq b}V_{lb}(\textbf{e})\cdot (\bar{B}_{al}-\bar{B}_{ab})^2\nonumber \\
=&~ -\tilde{c}_2 \cdot \sum_{b}\sum_{l\neq b}[ V_{lb}(\textbf{e})\cdot \sum_a (\bar{B}_{al}-\bar{B}_{ab})^2] \overset{(c)}{\leq} -\tilde{c}_3 \cdot \sum_{b}\sum_{l\neq b} V_{lb}(\textbf{e})=-\frac{\tilde{c}_3}{2}\|V(\textbf{e})-V(\textbf{c})\|_1, \nonumber
\end{align}
where $\tilde{B}_{ab}$ is a number between $\bar{B}_{ab}$ and $\bar{B}_{al}$. To obtain (a), we have used the second order Taylor expansion of the function $\bar{B}_{ab}\log x- x$ around its maxima $x=\bar{B}_{ab}$; (b) is simply due to $\min_{a}\pi_a>0, \min_{ab}\bar{B}_{ab}>0$; (c) holds since $\bar{B}$ has no two identical columns. In order to obtain \eqref{key:point}, we need to evaluate $g'(\epsilon; V(\textbf{e}))$ for small $\epsilon$. From \eqref{g:general} and \eqref{gzero}, it is straightforward to see
\begin{align}
&~g'(\epsilon; V(\textbf{e}))-g'(0; V(\textbf{e}))\label{gsum}\\
=&~\underbrace{2\epsilon \cdot \sum_{klab} \pi_a\pi_b \cdot [V_{ka}(\textbf{e})-V_{ka}(\textbf{c})]\cdot [V_{lb}(\textbf{e})-V_{lb}(\textbf{c})]\cdot (\bar{B}_{ab}\log  r_{kl}-r_{kl} )}_{\triangleq G_1} \nonumber \\
&+ \underbrace{\sum_{klab}  \pi_a\pi_b\bar{B}_{ab}([V_{ka}(\textbf{e})-V_{ka}(\textbf{c})]V_{lb}(\textbf{c})+[V_{lb}(\textbf{e})-V_{lb}(\textbf{c})]V_{ka}(\textbf{c}))\cdot \log\frac{r_{kl}}{\bar{B}_{kl}}}_{\triangleq G_2}\nonumber \\
&- \underbrace{\sum_{klab}  \pi_a\pi_b([V_{ka}(\textbf{e})-V_{ka}(\textbf{c})]V_{lb}(\textbf{c})+[V_{lb}(\textbf{e})-V_{lb}(\textbf{c})]V_{ka}(\textbf{c}))\cdot (r_{kl}-\bar{B}_{kl} )}_{\triangleq G_3}. \nonumber 
\end{align}
We now bound $G_1, G_2$ and $G_3$ in the above equation. For $G_1$, note that $| \log r_{kl} | \leq  \max_{ab} |\log \bar{B}_{ab}|$ and $|r_{kl}| \leq \max_{ab} \bar{B}_{ab}$. Therefore, $\exists~\tilde{c}_4>0$ such that
\begin{align}\label{g1:bound}
|G_1|\leq \tilde{c}_4 \epsilon \cdot  \sum_{ka}\sum_{bl} |V_{ka}(\textbf{e})-V_{ka}(\textbf{c})| \cdot |V_{lb}(\textbf{e})-V_{lb}(\textbf{c})| \leq 2K^2 \tilde{c}_4 \epsilon\cdot  \|V(\textbf{e})-V(\textbf{c})\|_1.
\end{align}
Regarding $G_2$, since $\tilde{V}_{ka}-V_{ka}(\textbf{c})=O(\epsilon)$, we have
\begin{align*}
\log\frac{r_{kl}}{\bar{B}_{kl}} =\log\frac{\pi_k\pi_l\bar{B}_{kl}+O(\epsilon)}{\pi_k\pi_l\bar{B}_{kl}+O(\epsilon)}=O(\epsilon).
\end{align*}
So we can bound $G_2$:
\begin{align}\label{g2:bound}
|G_2| \leq O(\epsilon) \cdot \sum_{klab} ( |V_{ka}(\textbf{e})-V_{ka}(\textbf{c})| + |V_{lb}(\textbf{e})-V_{lb}(\textbf{c})| )=O(\epsilon)\cdot \|V(\textbf{e})-V(\textbf{c})\|_1.
\end{align}
Similar arguments can give us $|G_3|= O(\epsilon) \cdot \|V(\textbf{e})-V(\textbf{c})\|_1$. This fact combined with \eqref{final:zero}, \eqref{gsum}, \eqref{g1:bound} and \eqref{g2:bound} completes the proof of \eqref{key:point}. We now consider any $V(\textbf{e})$ such that $\|V(\textbf{e})-V(\textbf{c})\|_1 \leq \tilde{c}_1$. Define ${V}^*(\textbf{e})=V(\textbf{c})+\frac{V(\textbf{e})-V(\textbf{c})}{\|V(\textbf{e})-V(\textbf{c})\|_1}$. It is then straightforward to confirm that $V^*(\bm{e})\in \mathcal{V}$. Hence,
\begin{align*}
&F(T(\textbf{c}),\textbf{f}^0(\textbf{c}))-F(T(\textbf{e}),\textbf{f}^0(\textbf{e}))=H(V(\textbf{c}))-H(V(\textbf{e})) \\
&=g(0; V^*(\textbf{e}))-g(\|V(\textbf{e})-V(\textbf{c})\|_1;V^*(\textbf{e})) \\
&\overset{(d)}{=} -g'(\tilde{\epsilon};V^*(\textbf{e})) \cdot \|V(\textbf{e})-V(\textbf{c})\|_1 \overset{(e)}{\geq} \tilde{c}_0 \cdot \|V^*(\textbf{e})-V(\textbf{c})\|_1  \cdot \|V(\textbf{e})-V(\textbf{c})\|_1 \\
&=\tilde{c}_0 \cdot \|V(\textbf{e})-V(\textbf{c})\|_1,
\end{align*}
where (d) is simply by mean value theorem; $\tilde{\epsilon}$ is between $0$ and $\|V(\textbf{e}-V(\textbf{c}))\|_1$; (e) holds because of \eqref{key:point}. This finishes the proof.
\end{proof}

\vspace{0.5cm}

\noindent \emph{\textbf{Proof of Theorem \ref{thm:mle}}}. It is not hard to check that similar proofs as the ones of Lemma 1 and Theorem 2 in \cite{bickel2013asymptotic} work under NSBM\footnote{The tail condition on $\bm{x}$ is used to show that Theorem 1 in \cite{bickel2013asymptotic} holds under NSBM.}. For simplicity, we do not repeat the derivations here. As a result, we can obtain that as $n\rightarrow \infty$,
\begin{align}\label{param:error}
\sqrt{n}(\hat{\bm{\beta}}-\bm{\beta}_0)\rightarrow N(0, I^{-1}(\bm{\beta}_0)), \quad  \sqrt{n^2\rho_n} \log \frac{\hat{B}_{ab}}{B^0_{ab}}=O_p(1), ~~1\leq a, b \leq K.
\end{align}
Based on \eqref{param:error}, we would like to show the strong consistency of $\hat{\bm{c}}$. Define 
\begin{align*}
\mathcal{N}(\bm{e};\{B_{ab}\})&=\sum_{ab}\Big[ \frac{O_{ab}(\bm{e})}{\mu_n}\log \frac{B_{ab}}{\rho_n}+\frac{n_a(\bm{e})n_b(\bm{e})-O_{ab}(\bm{e})}{\mu_n}\log(1-B_{ab}) \Big], \\
\mathcal{C}(\bm{e}; \bm{\beta})&=  \frac{1}{\mu_n} \sum_{i}\Big[ \bm{\beta}^T_{e_i}\bm{x}_i- \log \Big( \sum_{k=1}^K\exp({\bm{\beta}^T_{k}\bm{x}_i}) \Big) \Big].
\end{align*}
We then easily see that
\begin{align*}
\hat{\bm{c}} =\argmax_{\bm{e} \in \{1,\dots, K\}^n} \mathcal{N}(\bm{e}; \{\hat{B}_{ab}\})+\mathcal{C}(\bm{e}; \hat{\bm{\beta}}).
\end{align*}
The subsequent proof is aligned with the ideas of proving strong consistency in \cite{zhao2012consistency}. We first prove that $\hat{\bm{c}}$ is weak consistent. Note that
\begin{align}
\max_{\bm{e}}|\mathcal{N}(\textbf{e}; \{\hat{B}_{ab}\})-\mathcal{H}(T(\bm{e}), \bm{f}^0(\bm{e})) | &\leq \max_{\bm{e} }|\mathcal{N}(\textbf{e}; \{\hat{B}_{ab}\})-\mathcal{N}(\textbf{e}; \{B^0_{ab}\})| \label{chain} \\
&+\max_{\bm{e}} |\mathcal{N}(\bm{e};\{B^0_{ab}\})-\mathcal{H}(\hat{T}(\bm{e}), \hat{\bm{f}}(\bm{e}))|\nonumber \\
&+\max_{\bm{e}}|\mathcal{H}(\hat{T}(\bm{e}), \hat{\bm{f}}(\bm{e}))-\mathcal{H}(T(\bm{e}), \bm{f}^0(\bm{e}))|. \nonumber  
\end{align}
We aim to bound the three terms on the right hand side of the above inequality. For the first one, we have
\begin{align*}
\max_{\bm{e} }|\mathcal{N}(\textbf{e}; \{\hat{B}_{ab}\})-\mathcal{N}(\textbf{e}; \{B^0_{ab}\})| \leq &\sum_{ab} \Bigg[ \max_{\bm{e}} \frac{O_{ab}(\bm{e})}{\mu_n}\cdot \Big |\log \frac{\hat{B}_{ab}(1-B^0_{ab})}{B^0_{ab}(1-\hat{B}_{ab})}\Big|\\
&+  \frac{1}{\rho_n} \Big | \log\frac{1-\hat{B}_{ab}}{1-B^0_{ab}} \Big| \Bigg].
\end{align*}
According to \eqref{param:error}, if we can show $ \max_{\bm{e}} \frac{O_{ab}(\bm{e})}{\mu_n}=O_p(1)$, the above inequality will imply $\exists~a_n \rightarrow 0$ s.t.
\begin{align}\label{big:0}
P(\max_{\bm{e} }|\mathcal{N}(\textbf{e}; \{\hat{B}_{ab}\})-\mathcal{N}(\textbf{e}; \{B^0_{ab}\})| \leq a_n) \rightarrow 1.
\end{align}
For that purpose, we first apply \eqref{concen:one} by choosing $\epsilon_n=(n\rho_n)^{-1/3}$ to have 
\begin{align}\label{big:one}
P\Big(\max_{\bm{e}}\max_{kl}\Big |\frac{O_{kl}(\bm{e})}{\mu_n}-\hat{T}_{kl}(\bm{e})\Big | \leq \epsilon_n \Big ) \rightarrow 1.
\end{align}    
Also notice that $\max_{\bm{e}}|T_{kl}(\bm{e})| \leq \sum_{ab}\bar{B}_{ab}$, and 
\begin{align}\label{big:two}
\max_{\bm{e}}|T_{kl}(\bm{e})-\hat{T}_{kl}(\bm{e})| \leq \sum_{ab} \bar{B}_{ab}( |\hat{\pi}_a-\pi_a|+ |\hat{\pi}_b-\pi_b|) \overset{P}{\rightarrow}0.
\end{align}
Combining \eqref{big:one} and \eqref{big:two} yields the result $ \max_{\bm{e}} \frac{O_{ab}(\bm{e})}{\mu_n}=O_p(1)$. Regarding the second term on the right hand side of \eqref{chain}, it is straightforward to see that
\begin{align*}
&\max_{\bm{e}} |\mathcal{N}(\bm{e};\{B^0_{ab}\})-\mathcal{H}(\hat{T}(\bm{e}), \hat{\bm{f}}(\bm{e}))| \leq \max_{\bm{e}}\max_{kl}\Big |\frac{O_{kl}(\bm{e})}{\mu_n}-\hat{T}_{kl}(\bm{e})\Big |\cdot \sum_{ab} |\log \bar{B}_{ab}| + \\
&\max_{\bm{e}}\max_{kl}\frac{O_{kl}(\bm{e})}{\mu_n}\cdot \sum_{ab} |\log (1-\rho_n \bar{B}_{ab})|+\sum_{ab} \Big |\frac{\log(1-\rho_n \bar{B}_{ab})}{\rho_n}+\bar{B}_{ab}  \Big |.
\end{align*}
The fact that $ \max_{\bm{e}} \max_{kl}\frac{O_{kl}(\bm{e})}{\mu_n}=O_p(1)$ and $\rho_n \rightarrow 0$, combined with \eqref{big:one} enables us to conclude from the last inequality: $\exists~b_n \rightarrow 0$ such that
\begin{align}\label{big:00}
P(\max_{\bm{e}} |\mathcal{N}(\bm{e};\{B^0_{ab}\})-\mathcal{H}(\hat{T}(\bm{e}), \hat{\bm{f}}(\bm{e}))| \leq b_n) \rightarrow 1.
\end{align}
For the third term on the right hand side of \eqref{chain}, it is easily seen that
\begin{align*}
&\max_{\bm{e}}|\mathcal{H}(\hat{T}(\bm{e}), \hat{\bm{f}}(\bm{e}))-\mathcal{H}(T(\bm{e}), \bm{f}^0(\bm{e}))| \leq \sum_{ab} \max_{\bm{e}}|T_{ab}(\bm{e})-\hat{T}_{ab}(\bm{e})| \cdot  |\log \bar{B}_{ab}|+\\
&\sum_{ab} \max_{e} (|f^0_a(\bm{e})-\hat{f}_a(\bm{e})|+|f^0_b(\bm{e})-\hat{f}_b(\bm{e})|)\cdot \bar{B}_{ab}=O(\max_a |\hat{\pi}_a-\pi_a|).
\end{align*}
Hence there exists $c_n \rightarrow 0$ such that
\begin{align}\label{big:000}
P(\max_{\bm{e}}|\mathcal{H}(\hat{T}(\bm{e}), \hat{\bm{f}}(\bm{e}))-\mathcal{H}(T(\bm{e}), \bm{f}^0(\bm{e}))| \leq c_n) \rightarrow 1.
\end{align}
Putting \eqref{chain}, \eqref{big:0}, \eqref{big:00} and \eqref{big:000} together, we obtain that $\exists~ \tilde{\epsilon}_n \rightarrow 0$ so that
\begin{align}\label{key:piece1}
P(\max_{\bm{e}}|\mathcal{N}(\textbf{e}; \{\hat{B}_{ab}\})-\mathcal{H}(T(\bm{e}), \bm{f}^0(\bm{e})) | \leq \tilde{\epsilon}_n/2)\rightarrow 1. 
\end{align}
We now turn to bounding $\max_{\bm{e}} |\mathcal{C}(\bm{e}; \hat{\bm{\beta}})|$. We first decompose it as
\begin{align*}
\mathcal{C}(\bm{e}; \hat{\bm{\beta}})=\frac{n}{\mu_n}\Big[\frac{\mu_n}{n}\mathcal{C}(\bm{c};\hat{\bm{\beta}})+\frac{1}{n} \sum_i (\hat{\bm{\beta}}^T_{e_i}-\hat{\bm{\beta}}^T_{c_i})\bm{x}_i \Big].
\end{align*}
According to \eqref{param:error} and \eqref{uni:one}, it is not hard to get $\frac{\mu_n}{n}\mathcal{C}(\bm{c};\hat{\bm{\beta}})=O_p(1)$. Also a direct Cauchy-Schwartz inequality leads to
\begin{align*}
\max_{\bm{e}}\Big |\frac{1}{n} \sum_i (\hat{\bm{\beta}}^T_{e_i}-\hat{\bm{\beta}}^T_{c_i})\bm{x}_i\Big | \leq 2\|\hat{\bm{\beta}}\|_2 \cdot \frac{1}{n}\sum_i \|\bm{x}_i\|_2=O_p(1).
\end{align*}
Therefore, we have (choosing $\tilde{\epsilon}_n$ large enough)
\begin{align}\label{key:piece2}
P(\max_{\bm{e}} |\mathcal{C}(\bm{e};\hat{\bm{\beta}})| \leq \tilde{\epsilon}_n/2) \rightarrow 1.
\end{align}
The two high probability results \eqref{key:piece1} and \eqref{key:piece2} together give us
\begin{align*}
P(\max_{\bm{e}}|\mathcal{N}(\textbf{e}; \{\hat{B}_{ab}\})+\mathcal{C}(\bm{e};\hat{\bm{\beta}})-\mathcal{H}(T(\bm{e}), \bm{f}^0(\bm{e})) | \leq \tilde{\epsilon}_n)\rightarrow 1. 
\end{align*}
Based on this result, Lemma \ref{strong:curv} implies that $\exists~ \delta_n =O(\tilde{\epsilon}_n)$ s.t
\begin{align*}
P\Big(\max_{\|V(\bm{e})-V(\bm{c})\|_1> \delta_n} [ \mathcal{N}(\textbf{e}; \{\hat{B}_{ab}\})+\mathcal{C}(\bm{e};\hat{\bm{\beta}})] < \mathcal{N}(\textbf{c}; \{\hat{B}_{ab}\})+\mathcal{C}(\bm{c};\hat{\bm{\beta}}) \Big) \rightarrow 1,
\end{align*}
which leads to the weak consistency of $\hat{\bm{c}}$. To obtain strong consistency, it suffices to show
\begin{align*}
P\Big(\max_{0<\|V(\bm{e})-V(\bm{c})\|_1\leq \delta_n} [\mathcal{N}(\textbf{e}; \{\hat{B}_{ab}\})+\mathcal{C}(\bm{e};\hat{\bm{\beta}})] < \mathcal{N}(\textbf{c}; \{\hat{B}_{ab}\})+\mathcal{C}(\bm{c};\hat{\bm{\beta}}) \Big) \rightarrow 1.
\end{align*}
Denote $\mathscr{E}=\{\bm{e}: 0<\|V(\bm{e})-V(\bm{c})\|_1\leq \delta_n\}$. We first bound $\mathcal{N}(\textbf{e}; \{\hat{B}_{ab}\})-\mathcal{N}(\textbf{c}; \{\hat{B}_{ab}\})$.
\begin{align*}
&\mathcal{N}(\textbf{e}; \{\hat{B}_{ab}\})-\mathcal{N}(\textbf{c}; \{\hat{B}_{ab}\})\\
&=\underbrace{ \mathcal{N}(\textbf{e}; \{\hat{B}_{ab}\})-\mathcal{N}(\textbf{e}; \{B^0_{ab}\})+\mathcal{N}(\textbf{c}; \{B^0_{ab}\})-\mathcal{N}(\textbf{c}; \{\hat{B}_{ab}\})}_{\triangleq N_1(\bm{e})} \\
&+\underbrace{ \mathcal{H}(T(\textbf{e}), \textbf{f}^0(\textbf{e}))-\mathcal{H}(T(\textbf{c}), \textbf{f}^0(\textbf{c}))}_{\triangleq N_2(\bm{e})}\\
&+\underbrace{\mathcal{N}(\textbf{e}; \{B^0_{ab}\})- \mathcal{H}(T(\textbf{e}), \textbf{f}^0(\textbf{e}))+ \mathcal{H}(T(\textbf{c}), \textbf{f}^0(\textbf{c}))-\mathcal{N}(\textbf{c}; \{B^0_{ab}\})}_{\triangleq N_3(\bm{e})}.
\end{align*}
For $N_1(\bm{e})$, we have
\begin{align}
\max_{e \in \mathscr{E}}|N_1(\bm{e})| &\leq \max_{e \in \mathscr{E}} \max_{ab} \frac{|O_{ab}(\bm{e})-O_{ab}(\bm{c})|}{\mu_n} \cdot  \sum_{ab}\Big| \log \frac{\hat{B}_{ab}(1-B^0_{ab})}{B^0_{ab}(1-\hat{B}_{ab})}\Big|+  \label{n1:2}\\
&\max_{e \in \mathscr{E}} \max_{ab} \frac{|n_a(\bm{e})n_b(\bm{e})-n_a(\bm{c})n_b(\bm{c})|}{n^2}\cdot \sum_{ab} \Big| \log \frac{1-\hat{B}_{ab}}{1-B^0_{ab}}  \Big|\frac{1}{\rho_n} . \nonumber
\end{align}
From \eqref{param:error} we see that $\log \frac{\hat{B}_{ab}(1-B^0_{ab})}{B^0_{ab}(1-\hat{B}_{ab})}=o_p(1),  \frac{1}{\rho_n}\cdot \log \frac{1-\hat{B}_{ab}}{1-B^0_{ab}} =o_p(1)$. It is also straightforward to confirm,
\begin{align}\label{n1:3}
\max_{ab} \frac{|n_a(\bm{e})n_b(\bm{e})-n_a(\bm{c})n_b(\bm{c})|}{n^2}\leq \frac{2}{n} \sum_i \mathbbm{1}(e_i \neq c_i)\leq \|V(\bm{e})-V(\bm{c})\|_1.
\end{align}
Moreover, note that there exists $C_1>0$ such that
\begin{eqnarray}
&&\max_{\bm{e}\in \mathscr{E}}\max_{kl}|\hat{T}_{kl}(\bm{e})-\hat{T}_{kl}(\bm{c})|\label{tkl:bound2}\\
&\leq& \max_{\bm{e}\in \mathscr{E}}\max_{kl} \max_{ab}\bar{B}_{ab}\cdot \sum_{ka} \sum_{lb}(|V_{ka}(\bm{e})-V_{ka}(\bm{c})|+|V_{lb}(\bm{e})-V_{lb}(\bm{c})|) \nonumber \\
&\leq& C_1 \cdot \max_{\bm{e}\in \mathscr{E}} \|V(\bm{e})-V(\bm{c})\|_1.  \nonumber
\end{eqnarray}
We can then derive from \eqref{concen:two} and \eqref{tkl:bound2} that
\begin{align}\label{n1:1}
P\Big(\max_{\bm{e} \in \mathscr{E}} \Big[\max_{ab} \Big |\frac{O_{ab}(\bm{e})}{\mu_n}-\frac{O_{ab}(\bm{c})}{\mu_n} \Big | -C_2 \cdot \|V(\bm{e})-V(\bm{c})\|_1   \Big] \leq 0\Big) \rightarrow 1,
\end{align}
where $C_2>0$ is a constant. 
Combining \eqref{n1:2}, \eqref{n1:3} and  \eqref{n1:1} yields
\begin{align}\label{n1:fb}
P(\max_{e \in \mathscr{E}}|N_1(\bm{e})| \leq o(1)\cdot \max_{e \in \mathscr{E}}\|V(\bm{e})-V(\bm{c})\|_1 ) \rightarrow 1.
\end{align}
Regarding $N_2(\bm{e})$, Lemma \ref{strong:curv} shows that $\exists~C_3 >0$ s.t.
\begin{align}\label{n2:fb}
N_2(\bm{e}) \leq -C_3 \cdot \|V(\bm{e})-V(\bm{c})\|_1, \mbox{~~~~for~} \bm{e} \in \mathscr{E}.
\end{align}
 For the last term, we can express $N_3(\bm{e})$ as
 \begin{align}
 N_3(\bm{e})&=\sum_{ab} \Big(\frac{O_{ab}(\bm{e})}{\mu_n}-\frac{O_{ab}(\bm{c})}{\mu_n}-\hat{T}_{ab}(\bm{e})+\hat{T}_{ab}(\bm{c}) \Big)\cdot \log \bar{B}_{ab}   \label{n3:00} \\
 &~~+ \sum_{ab} \frac{O_{ab}(\bm{c})-O_{ab}(\bm{e})}{\mu_n}\cdot \log (1-\bar{B}_{ab}\rho_n) \nonumber \\
 &~~+\sum_{ab} \frac{n_a(\bm{e})n_b(\bm{e})-n_a(\bm{c})n_b(\bm{c})}{n^2}\cdot \big [\bar{B}_{ab}+ \frac{1}{\rho_n} \log (1-\bar{B}_{ab}\rho_n) \big ]\nonumber \\
 &~~+ \sum_{ab} (\hat{T}_{ab}(\bm{e})
-T_{ab}(\bm{e})-\hat{T}_{ab}(\bm{c})  +T_{ab}(\bm{c}))\cdot \log \bar{B}_{ab} \nonumber \\
 &~~+ \sum_{ab}(f^0_a(\bm{e})f^0_b(\bm{e})-\hat{f}_a(\bm{e})\hat{f}_b(\bm{e})-f^0_a(\bm{c})f^0_b(\bm{c})+\hat{f}_a(\bm{c})\hat{f}_b(\bm{c})) \cdot \bar{B}_{ab}. \nonumber 
 \end{align}
With a few steps of calculations, it is not hard to obtain
\begin{align*}
\max_{ab}|\hat{T}_{ab}(\bm{e})-T_{ab}(\bm{e})-\hat{T}_{ab}(\bm{c})  +T_{ab}(\bm{c})| \leq \|V(\bm{e})-V(\bm{c})\|_1 \cdot O(\max_a|\hat{\pi}_a-\pi_a|), 
\end{align*}
and
\begin{align*}
&~\max_{ab} |f^0_a(\bm{e})f^0_b(\bm{e})-\hat{f}_a(\bm{e})\hat{f}_b(\bm{e})-f^0_a(\bm{c})f^0_b(\bm{c})+\hat{f}_a(\bm{c})\hat{f}_b(\bm{c})|\\
\leq &~ \|V(\bm{e})-V(\bm{c})\|_1 \cdot O(\max_a|\hat{\pi}_a-\pi_a|).
\end{align*}
 We can then use the above results, together with \eqref{concen:two},  \eqref{n1:3} and \eqref{n1:1} to bound the terms on the right hand side of \eqref{n3:00}. As a result, we are able to show that for any $C_4 >0$,
 \begin{align}\label{n3:fb}
 P(\max_{\bm{e}\in \mathscr{E}}~ [|N_3(\bm{e})|-C_4 \cdot \|V(\bm{e})-V(\bm{c})\|_1]  \leq 0) \rightarrow 1.
 \end{align} 
 The bounds on $N_1(\bm{e}), N_2(\bm{e}), N_3(\bm{e})$ in \eqref{n1:fb}, \eqref{n2:fb} and \eqref{n3:fb} enables us to conclude that $\exists~C_5 >0$,
 \begin{align}\label{nn:bound}
 P(\max_{\bm{e} \in \mathscr{E}}~ [\mathcal{N}(\textbf{e}; \{\hat{B}_{ab}\})-\mathcal{N}(\textbf{c}; \{\hat{B}_{ab}\})+C_5 \cdot \|V(\bm{e})-V(\bm{c})\|_1] \leq 0 ) \rightarrow 1.
 \end{align}
 As a next step, we bound $\mathcal{C}(\bm{e}; \hat{\bm{\beta}}) - \mathcal{C}(\bm{c}; \hat{\bm{\beta}})$. The result \eqref{inter:three} implies that $\frac{n\max_i \|\bm{x}_i\|_2}{\mu_n}=o_p(1)$. This leads to
 \begin{align*}
 \mathcal{C}(\bm{e}; \hat{\bm{\beta}})-\mathcal{C}(\bm{c}; \hat{\bm{\beta}})&=\frac{1}{\mu_n} \sum_{i=1}^n\sum_{k=1}^K \hat{\bm{\beta}}_k^T\bm{x}_i(\mathbbm{1}(e_i=k)-\mathbbm{1}(c_i=k)) \\
 & \leq \frac{2\|\hat{\bm{\beta}}\|_2}{\mu_n} \sum_{i=1}^n \|\bm{x}_i\|_2 \cdot \mathbbm{1}(e_i \neq c_i) \\
 &\leq \frac{2\|\hat{\bm{\beta}}\|_2n\max_i \|\bm{x}_i\|_2}{\mu_n}\cdot \frac{\sum_{i=1}^n\mathbbm{1}(e_i\neq c_i)}{n}\\
 &\overset{(a)}{=} o_p(1) \cdot \|V(\bm{e})-V(\bm{c})\|_1,
 \end{align*}
 where (a) holds because $\|\hat{\bm{\beta}}\|_2=O_p(1), $ and $\frac{1}{n}\sum_{i=1}^n\mathbbm{1}(e_i\neq c_i)\leq \frac{1}{2}\|V(\bm{e})-V(\bm{c})\|_1$. Hence we know
 \begin{align}\label{cc:bound}
 P\Big(\max_{\bm{e} \in \mathscr{E}}~[ \mathcal{C}(\bm{e}; \hat{\bm{\beta}})-\mathcal{C}(\bm{c}; \hat{\bm{\beta}})-\frac{C_5}{2}\cdot \|V(\bm{e})-V(\bm{c})\|_1]   \leq 0 \Big) \rightarrow 1.
 \end{align}
 Finally, from the results of \eqref{nn:bound} and \eqref{cc:bound}, we are able to show that
 \begin{align*}
 P(\max_{\bm{e} \in \mathscr{E}} ~[\mathcal{N}(\textbf{e}; \{\hat{B}_{ab}\})+ \mathcal{C}(\bm{e}; \hat{\bm{\beta}}) ] < \mathcal{N}(\textbf{c}; \{\hat{B}_{ab}\})+ \mathcal{C}(\bm{c}; \hat{\bm{\beta}})) \rightarrow 1.
 \end{align*}
 This completes the proof of strong consistency for $\hat{\bm{c}}$.
 
$\hfill \Box$

\vspace{1cm}

\noindent \emph{\textbf{Proof of Theorem \ref{thm:mfvi}}}. It is not difficult to verify that the same results of Lemma 3 and Theorem 3 in \cite{bickel2013asymptotic} hold under NSBM. We can thus conclude,
\begin{align*}
\sqrt{n}(\check{\bm{\beta}}-\bm{\beta}_0)\rightarrow N({\bf 0}, I^{-1}(\bm{\beta}_0)), \quad  \sqrt{n^2\rho_n} \log \frac{\check{B}_{ab}}{B^0_{ab}}=O_p(1), 1\leq a, b \leq K.
\end{align*}
The proof for the strong consistency of $\check{\bm{c}}$ follows the same lines of arguments as for $\hat{\bm{c}}$ in Theorem \ref{thm:mle}. We hence do not repeat the steps here.

$\hfill \Box$

\vspace{0.5cm}

\noindent \emph{\textbf{Proof of Theorem \ref{thm:pmle}}}. We first focus on analyzing $\tilde{\bm{c}}$. Define
\begin{align*}
\mathcal{N}(\bm{e})&=\frac{1}{\mu_n}\sum_{ab} \Big[ O_{ab}(\bm{e}) \log \frac{O_{ab}(\bm{e})}{n_a(\bm{e})n_b(\bm{e})\rho_n}- O_{ab}(\bm{e}) \Big], \\
\mathcal{C}(\bm{e})&=\frac{1}{\mu_n}\max_{\substack{\bm{\beta} \in \mathbb{R}^{Kp}, \\ \bm{\beta}_K=\bm{0}}} \sum_{i}\Big[ \bm{\beta}^T_{e_i}\bm{x}_i- \log \Big( \sum_{k=1}^K\exp({\bm{\beta}^T_{k}\bm{x}_i}) \Big) \Big].
\end{align*}

According to \eqref{pmle:eq}, $\tilde{\bm{c}}$ can be equivalently expressed as
\begin{align*}
\tilde{\bm{c}}=\argmax_{\bm{e} \in \{1,\dots, K\}^n} [\mathcal{N}(\bm{e})+ \mathcal{C}(\bm{e})].
\end{align*} 

Under the conditions in part (i), we claim that the following results hold: $\exists~ C_1,C_2>0$ such that
\begin{itemize}
\item[$(\mathcal{A})$] $P(\max_{\bm{e}} |\mathcal{N}(\bm{e})+\mathcal{C}(\bm{e})-F(T(\bm{e}), \bm{f}^0(\bm{e}))| < C_1(n\rho_n)^{-1//2}) \rightarrow 1$, ~as~$n \rightarrow \infty$.
\item[$(\mathcal{B})$] for larger $n$, 
$$\min_{\{\bm{e}: \|V(\bm{e})-V(\bm{c})\|_1\geq \frac{3C_1}{C_2}(n\rho_n)^{-1/2} \}} [F(T(\bm{c}), \bm{f}^0(\bm{c}))-F(T(\bm{e}), \bm{f}^0(\bm{e}))] \geq 3C_1 (n\rho_n)^{-1/2}.$$
\end{itemize}
Suppose for now the two results above are correct. Choosing $\delta_n =\frac{3 C_1}{C_2}(n\rho_n)^{-1/2}$, we can obtain that
\begin{align*}
&~ \max_{\{\bm{e}: \|V(\bm{e})-V(\bm{c})\|_1\geq \delta_n \}} [\mathcal{N}(\bm{e})+\mathcal{C}(\bm{e})-\mathcal{N}(\bm{c})-\mathcal{C}(\bm{c})] \\
\leq&~   \max_{\{\bm{e}: \|V(\bm{e})-V(\bm{c})\|_1\geq \delta_n \}} [\mathcal{N}(\bm{e})+\mathcal{C}(\bm{e})-F(T(\bm{e}),\bm{f}^0(\bm{e}))]\\
&~+\max_{\{\bm{e}: \|V(\bm{e})-V(\bm{c})\|_1\geq \delta_n \}} [F(T(\bm{e}),\bm{f}^0(\bm{e}))-F(T(\bm{c}), \bm{f}^0(\bm{c}))]\\
&~+ [F(T(\bm{c}),\bm{f}^0(\bm{c}))-\mathcal{N}(\bm{c})-\mathcal{C}(\bm{c})] \\
\leq&~ 2 \max_{\bm{e}} |\mathcal{N}(\bm{e})+\mathcal{C}(\bm{e})-F(T(\bm{e}),\bm{f}^0(\bm{e}))|-3 C_1 (n\rho_n)^{-1/2}, 
\end{align*}
for large enough $n$. Hence as $n\rightarrow \infty$,
\begin{align*}
&~P\Big(\max_{\{\bm{e}: \|V(\bm{e})-V(\bm{c})\|_1\geq \delta_n \}} [\mathcal{N}(\bm{e})+\mathcal{C}(\bm{e})] < \mathcal{N}(\bm{c})+\mathcal{C}(\bm{c}) \Big) \\
\geq &~ P(\max_{\bm{e}} |\mathcal{N}(\bm{e})+\mathcal{C}(\bm{e})-F(T(\bm{e}), \bm{f}^0(\bm{e}))| < C_1(n\rho_n)^{-1/2})  \rightarrow 1.
\end{align*}
This would lead to the first result in part (i). Regarding the proof of $(\mathcal{A})$, we follow similar arguments in the proof of Theorem 4.1 of \cite{zhao2012consistency}. To save space, we do not detail all the calculations. The key steps are to show that $\exists ~C_3, C_4, C_5 >0,$ such that as $n\rightarrow \infty$
\begin{align}
&P(\max_{\bm{e}} \max_{ab} |O_{ab}(\bm{e})/\mu_n-\hat{T}_{ab}(\bm{e}) | < C_3 (n\rho_n)^{-1/2}) \rightarrow 1, \label{one:one}\\
& P ( \max_{a} |\hat{\pi}_a-\pi_a | <C_4 (n\rho_n)^{-1/2}) \rightarrow 1, \label{one:two} \\
& P(\max_{\bm{e}} |\mathcal{C}(\bm{e})| \leq C_5 (n\rho_n)^{-1/2}) \rightarrow 1. \label{one:three}
\end{align}
The result \eqref{one:one} can be obtained from \eqref{concen:one} by choosing $C_3$ sufficiently large. The condition $\rho_n \rightarrow 0$ combined with a direct application of Hoeffding's inequality gives \eqref{one:two}. And it is straightforward to verify that for any $\bm{e}\in \{1,2,\dots, K\}^n$, 
\begin{align*}
0 \geq \mathcal{C}(\bm{e}) \geq \frac{1}{\mu_n}\sum_{i}\Big[ \bm{0}^T_{e_i}\bm{x}_i- \log \Big( \sum_{k=1}^K\exp({\bm{0}^T_{k}\bm{x}_i}) \Big) \Big]=\frac{-n\log K}{\mu_n}.
\end{align*}
This together with the condition $n \rho_n \rightarrow \infty$ yields \eqref{one:three}. To prove $(\mathcal{B})$, we first use Lemma 1:
\begin{align}
F(T(\bm{c}), \bm{f}^0(\bm{c}))-F(T(\bm{e}), \bm{f}^0(\bm{e})) \geq C_6 \cdot \|V(\bm{e})-V(\bm{c})\|_1, \mbox{~if~}\|V(\bm{e})-V(\bm{c})\|_1\leq C_7, \label{result2:three}
\end{align}
where $C_6, C_7$ are two positive constants. Furthermore, since $V(\bm{c})$ is the unique maximizer of $F(T(\bm{e}), f^0(\bm{e}))$ as a function of $V(\bm{e}) \in \mathcal{V}$ (see the proof of Theorem 3.4 in \cite{zhao2012consistency}), we have
\begin{align}
\min_{\{\bm{e}: \|V(\bm{e})-V(\bm{c})\|_1 > C_7 \}}[F(T(\bm{c}), \bm{f}^0(\bm{c}))-F(T(\bm{e}), \bm{f}^0(\bm{e}))]  >3C_1(n\rho_n)^{-1/2}, \label{result2:four}
\end{align}
for large $n$. On the other hand, \eqref{result2:three} implies that
\begin{align}\label{result2:five}
\min_{\{\bm{e}: ~C_7 \geq \|V(\bm{e})-V(\bm{c})\|_1\geq \frac{3 C_1}{C_6}(n\rho_n)^{-1/2} \}} [F(T(\bm{c}), \bm{f}^0(\bm{c}))-F(T(\bm{e}), \bm{f}^0(\bm{e}))] \geq 3 C_1 (n\rho_n)^{-1/2}.
\end{align}
Finally \eqref{result2:four} and \eqref{result2:five} together finishes the proof of $(\mathcal{B})$ with the choice $C_2=C_6$.

The next step is to prove the strong consistency of $\tilde{\bm{c}}$ in part (ii). To derive strong consistency, we also need to show
\begin{align}\label{strcon:ref}
P\Big(\max_{\{\bm{e}: 0 < \|V(\bm{e})-V(\bm{c})\|_1< \delta_n \}} [\mathcal{N}(\bm{e})+\mathcal{C}(\bm{e})] < \mathcal{N}(\bm{c})+\mathcal{C}(\bm{c}) \Big) \rightarrow 1,
\end{align}
which requires a refined analysis. Denote $\mathcal{E}=\{\bm{e}: 0 < \|V(\bm{e})-V(\bm{c})\|_1< \delta_n\}$. We make use of the existing result: $\exists~C_7 >0$ such that as $n\rightarrow \infty$,
\begin{align*}
P\big(\max_{\bm{e} \in \mathcal{E}} [\mathcal{N}(\bm{e})-\mathcal{N}(\bm{c})+C_7\cdot \|V(\bm{e})-V(\bm{c})\|_1] <0 \big) \rightarrow 1.
\end{align*}
This is obtained by combining (A.13) in \cite{zhao2012consistency} and (1.1) in \cite{bickel2015correction}. If we can show 
\begin{align}\label{logistic:ref}
P\Big(\max_{\bm{e} \in \mathcal{E}} \Big[ \mathcal{C}(\bm{e})-\mathcal{C}(\bm{c}) - C_7\cdot \|V(\bm{e})-V(\bm{c})\|_1 \Big] \leq 0 \Big) \rightarrow 1,
\end{align}
the result \eqref{strcon:ref} will be proved by simply putting together the last two high probability arguments. To derive \eqref{logistic:ref}, we introduce several notations. Denote
\begin{align*}
&\mathcal{R}(\bm{\beta})=\mathbb{E} \Bigg[\sum_{k=1}^K\bm{\beta}_k^T\bm{x}\mathbbm{1}(c=k)-\log \Big(\sum_{k=1}^K \exp({\bm{\beta}_k^T\bm{x}}) \Big) \Bigg], \\
&\mathcal{R}_n(\bm{\beta}, \bm{e})=\frac{1}{n} \sum_{i=1}^n\Bigg[\sum_{k=1}^K\bm{\beta}_k^T\bm{x}_i\mathbbm{1}(e_i=k)-\log\Big(\sum_{k=1}^K\exp({\bm{\beta}_k^T\bm{x}_i}) \Big)\Bigg].
\end{align*}
Observe that $\mathcal{R}_n(\bm{\beta}, \bm{c})$ is the sample version of $\mathcal{R}(\bm{\beta})$. We further define an M-estimator:
\begin{align}\label{m:estimator}
\hat{\bm{\beta}}(\bm{e})=\argmax_{\bm{\beta}} \mathcal{R}_n(\bm{\beta}, \bm{e}).
\end{align}
According to the Convexity Lemma (Lemma 7.75) in \cite{liese2007statistical}, since $-\mathcal{R}_n(\bm{\beta}, \bm{c})$ is a convex stochastic process and $\mathcal{R}_n(\bm{\beta}, \bm{c}) \overset{P}{\rightarrow} \mathcal{R}(\bm{\beta})$, we can obtain 
\begin{align}\label{uni:one}
\sup_{\|\bm{\beta}-\bm{\beta}_0\|_2 \leq \epsilon} | \mathcal{R}_n(\bm{\beta},\bm{c})-\mathcal{R}(\bm{\beta})| \overset{P}{\rightarrow} 0,
\end{align}
where $\epsilon $ is an arbitrary positive constant. Also note that
\begin{align}
&\sup_{\|\bm{\beta}-\bm{\beta}_0\|_2 \leq \epsilon} \sup_{\bm{e} \in \mathcal{E}} |\mathcal{R}_n(\bm{\beta},\bm{e})-\mathcal{R}(\bm{\beta})| \leq \sup_{\|\bm{\beta}-\bm{\beta}_0\|_2 \leq \epsilon} | \mathcal{R}_n(\bm{\beta},\bm{c})-\mathcal{R}(\bm{\beta})|\label{uni:two} \\
&~+\sup_{\|\bm{\beta}-\bm{\beta}_0\|_2 \leq \epsilon} \sup_{\bm{e} \in \mathcal{E} } \Big |\frac{1}{n} \sum_{i=1}^n \sum_{k=1}^K \bm{\beta}_k^T\bm{x}_i (\mathbbm{1}(c_i=k)-\mathbbm{1}(e_i=k)) \Big |,  \nonumber
\end{align}
and
\begin{align}
&~\sup_{\|\bm{\beta}-\bm{\beta}_0\|_2 \leq \epsilon} \sup_{\bm{e} \in \mathcal{E} } \Big |\frac{1}{n} \sum_{i=1}^n \sum_{k=1}^K \bm{\beta}_k^T\bm{x}_i (\mathbbm{1}(c_i=k)-\mathbbm{1}(e_i=k)) \Big |\label{uni:three}\\
\overset{(a)}{\leq} &~ C_8 \cdot \sup_{\bm{e} \in \mathcal{E} } \frac{1}{n} \sum_{i=1}^n \|\bm{x}_i\|_2 \cdot \mathbbm{1}(c_i\neq e_i) \nonumber \\
\overset{(b)}{\leq} &~C_8 \cdot \Big( \frac{1}{n}\sum_{i=1}^n\|\bm{x}_i\|_2^{\alpha} \Big)^{\frac{1}{\alpha}} \cdot \sup_{\bm{e} \in \mathcal{E} }  \Big(\frac{1}{n} \sum_{i=1}^n \mathbbm{1}(c_i\neq e_i) \Big)^{\frac{\alpha-1}{\alpha}} \nonumber \\
\leq &~C_9\cdot  \Big( \frac{1}{n}\sum_{i=1}^n\|\bm{x}_i\|_2^{\alpha} \Big)^{\frac{1}{\alpha}} \cdot (\delta_n)^{\frac{\alpha-1}{\alpha}} \overset{P}{\rightarrow} 0,  \nonumber
\end{align}
where $C_8, C_9 $ are two positive constants; (a) is by Cauchy-Schwarz inequality and (b) is due to H$\ddot{o}$lder's inequality. Putting \eqref{uni:one}, \eqref{uni:two} and \eqref{uni:three} together, we are able to show
\begin{align}\label{uni:new}
\sup_{\|\bm{\beta}-\bm{\beta}_0\|_2 \leq \epsilon} \sup_{\bm{e} \in \mathcal{E}} |\mathcal{R}_n(\bm{\beta},\bm{e})-\mathcal{R}(\bm{\beta})| \overset{P}{\rightarrow } 0.
\end{align}
Having the uniform convergence of $\mathcal{R}_n(\bm{\beta}, \bm{e})$, we next study the convergence of $\hat{\bm{\beta}}(\bm{e})$. However, the uniformity in \eqref{uni:new} only holds over a compact set. The result may not be applied directly. For this reason, we introduce an ancillary variable: $\bar{\bm{\beta}}(\bm{e})=\alpha \hat{\bm{\beta}}(\bm{e})+(1-\alpha)\bm{\beta}_0$, where $\alpha=\frac{\epsilon}{\epsilon + \|\hat{\bm{\beta}}(\bm{e})-\bm{\beta}_0\|_2}$. Clearly,
\begin{align*}
\|\bar{\bm{\beta}}(\bm{e})-\bm{\beta}_0\|_2=\alpha \|\hat{\bm{\beta}}(\bm{e})-\bm{\beta}_0\|_2 \leq \epsilon.
\end{align*}
Besides, since $\mathcal{R}_n(\bm{\beta}, \bm{e})$ is a concave function of $\bm{\beta}$ and $\hat{\bm{\beta}}(\bm{e})$ is its maximizer ,  we have
\begin{align*}
\mathcal{R}_n(\bar{\bm{\beta}}(\bm{e}), \bm{e}) \geq \alpha \mathcal{R}_n(\hat{\bm{\beta}}(\bm{e}), \bm{e})+(1-\alpha) \mathcal{R}_n(\bm{\beta}_0, \bm{e}) \geq  \mathcal{R}_n(\bm{\beta}_0, \bm{e}).
\end{align*}
Based on the last two inequalities and the fact that $\bm{\beta}_0$ is the maximizer of $\mathcal{R}(\bm{\beta})$, we can derive
\begin{align*}
0 &\leq \sup_{\bm{e} \in \mathcal{E}} ~[\mathcal{R}(\bm{\beta}_0)-\mathcal{R}(\bar{\bm{\beta}}(\bm{e}))]\\
&\le \sup_{\bm{e} \in \mathcal{E}}~ [\mathcal{R}(\bm{\beta}_0)-\mathcal{R}_n(\bm{\beta}_0, \bm{e})] + \sup_{\bm{e} \in \mathcal{E}}~[\mathcal{R}_n(\bm{\beta}_0, \bm{e})-\mathcal{R}_n(\bar{\bm{\beta}}(\bm{e}), \bm{e})] \\
&\quad +\sup_{\bm{e} \in \mathcal{E}}~[\mathcal{R}_n(\bar{\bm{\beta}}(\bm{e}),  \bm{e})-\mathcal{R}(\bar{\bm{\beta}}(\bm{e}))] \leq 2~ \sup_{\bm{e} \in \mathcal{E}} \sup_{\|\bm{\beta}-\bm{\beta}_0\|_2 \leq \epsilon} |\mathcal{R}_n(\bm{\beta},\bm{e})-\mathcal{R}(\bm{\beta})|  \overset{P}{\rightarrow } 0.
\end{align*}
Hence,
\begin{align}\label{obj:uni}
 \sup_{\bm{e} \in \mathcal{E}} ~[\mathcal{R}(\bm{\beta}_0)-\mathcal{R}(\bar{\bm{\beta}}(\bm{e}))]  \overset{P}{\rightarrow} 0.
\end{align}
Moreover, it can be directly verified that the maximizer $\bm{\beta}_0$ is unique and isolated. So $\forall ~\tilde{\epsilon}>0, \exists ~\eta >0$ such that 
\begin{align*}
P\Big(\sup_{\bm{e} \in \mathcal{E}}~ \|\bar{\bm{\beta}}(\bm{e})-\bm{\beta}_0\|_2>\tilde{\epsilon} \Big) \leq P\Big(\sup_{\bm{e} \in \mathcal{E}}~[ \mathcal{R}(\bm{\beta}_0)-\mathcal{R}(\bar{\bm{\beta}}(\bm{e}))] >\eta \Big) \rightarrow 0,
\end{align*}
where the last limit is implies by \eqref{obj:uni}. We thus have shown that $\sup_{\bm{e} \in \mathcal{E}}~\|\bar{\bm{\beta}}(\bm{e})-\bm{\beta}_0\|_2 \overset{P}{\rightarrow}  0$. It then leads to the consistency of $\hat{\bm{\beta}}(\bm{e})$:
\begin{align*}
\sup_{\bm{e} \in \mathcal{E}}~\|\hat{\bm{\beta}}(\bm{e})-\bm{\beta}_0\|_2\overset{(c)}=\sup_{\bm{e} \in \mathcal{E}}~ \frac{\epsilon \|\bar{\bm{\beta}}(\bm{e})-\bm{\beta}_0\|_2}{\epsilon - \|\bar{\bm{\beta}}(\bm{e})-\bm{\beta}_0\|_2} \leq \frac{\epsilon \sup_{\bm{e} \in \mathcal{E}}~  \|\bar{\bm{\beta}}(\bm{e})-\bm{\beta}_0\|_2}{\epsilon -\sup_{\bm{e} \in \mathcal{E}}~  \|\bar{\bm{\beta}}(\bm{e})-\bm{\beta}_0\|_2} \overset{P}{\rightarrow} 0,
\end{align*}
where $(c)$ is due to the definition of $\bar{\beta}(\bm{e})$. We are now in the position to derive \eqref{logistic:ref}.
\begin{align*}
&~\mathcal{C}(\bm{e})-\mathcal{C}(\bm{c})\\
=&~\frac{n}{\mu_n}\Big[ \mathcal{R}_n(\hat{\bm{\beta}}(\bm{e}),\bm{e})-\mathcal{R}_n(\hat{\bm{\beta}}(\bm{c}),\bm{c})\Big] \\
=&~\frac{n}{\mu_n}\Big[ \mathcal{R}_n(\hat{\bm{\beta}}(\bm{e}),\bm{e})-\mathcal{R}_n(\hat{\bm{\beta}}(\bm{e}),\bm{c})+  \mathcal{R}_n(\hat{\bm{\beta}}(\bm{e}),\bm{c})-\mathcal{R}_n(\hat{\bm{\beta}}(\bm{c}),\bm{c}) \Big  ] \\
\leq&~ \frac{n}{\mu_n}\Big[ \mathcal{R}_n(\hat{\bm{\beta}}(\bm{e}),\bm{e})-\mathcal{R}_n(\hat{\bm{\beta}}(\bm{e}),\bm{c})\Big] =\frac{1}{\mu_n} \sum_{i=1}^n \sum_{k=1}^K \hat{\bm{\beta}}_k^T(\bm{e})\bm{x}_i \cdot  [ \mathbbm{1}(e_i=k)-\mathbbm{1}(c_i=k)] \\
\leq&~ \frac{2}{\mu_n} \sum_{i=1}^n |\hat{\bm{\beta}}_k^T(\bm{e})\bm{x}_i |\cdot   \mathbbm{1}(e_i \neq c_i) \leq \frac{2 \|\hat{\bm{\beta}}(\bm{e})\|_2}{\mu_n} \sum_{i=1}^n \|\bm{x}_i\|_2 \cdot  \mathbbm{1}(e_i \neq c_i)  \\
\leq&~  \frac{2\|\hat{\bm{\beta}}(\bm{e})\|_2  \cdot \max_i\|\bm{x}_i\|_2}{\mu_n} \sum_{i=1}^n    \mathbbm{1}(e_i \neq c_i) \leq  \frac{n \|\hat{\bm{\beta}}(\bm{e})\|_2  \cdot \max_i\|\bm{x}_i\|_2}{\mu_n}\|V(\bm{e})-V(\bm{c})\|_1.
\end{align*}
Therefore, we can obtain
\begin{align}\label{inter:one}
&~\sup_{\bm{e} \in \mathcal{E}}~ \Big[\mathcal{C}(\bm{e})-\mathcal{C}(\bm{c}) - C_7 \cdot \|V(\bm{e})-V(\bm{c})\|_1\Big]  \\
\leq&~ \sup_{\bm{e} \in \mathcal{E}}~ \Big(\frac{n \|\hat{\bm{\beta}}(\bm{e})\|_2  \cdot \max_i\|\bm{x}_i\|_2}{\mu_n}-C_7 \Big) \|V(\bm{e})-V(\bm{c})\|_1. \nonumber
\end{align}
Because $\sup_{\bm{e} \in \mathcal{E}}~\|\hat{\bm{\beta}}(\bm{e})-\bm{\beta}_0\|_2 \overset{P}{\rightarrow}0 $, it is straightforward to show
\begin{align}\label{inter:two}
P(\sup_{\bm{e} \in \mathcal{E}}~\|\hat{\bm{\beta}}(\bm{e})\|_2 \leq 2 \|\bm{\beta}_0\|_2 ) \rightarrow 1.
\end{align}
According to the condition $P(\|\bm{x}\|_2 > t)=O(e^{-\kappa_2 t})$ and $\frac{n \rho_n}{\log n} \rightarrow \infty$, we have for any positive constant $c_0$,
\begin{align*}
 n P\Big(\frac{n}{\mu_n}\|\bm{x}\|_2 > c_0 \Big)=O(\exp\{\log n -\kappa_2 c_0n\rho_n\})=o(1),
\end{align*}
resulting in 
\begin{align}\label{inter:three}
\lim_{n\rightarrow \infty }P\Big(\frac{n}{\mu_n}\max_i \|\bm{x}_i\|_2 \leq c_0 \Big) &=\lim_{n \rightarrow \infty}\Big[1-P\Big(\frac{n}{\mu_n}\|\bm{x}\|_2 > c_0 \Big)\Big]^n\\
&=\lim_{n\rightarrow \infty} e^{-nP(\frac{n}{\mu_n}\|\bm{x}\|_2 > c_0 )}=1.\nonumber 
\end{align}
Choosing $c_0=\frac{C_7}{2 \|\bm{\beta}_0\|_2}$ and combining \eqref{inter:two} and \eqref{inter:three} lead to
\begin{align*}
&~ P\Big(\sup_{\bm{e} \in \mathcal{E}}~\Bigg[ \frac{n \|\hat{\bm{\beta}}(\bm{e})\|_2  \cdot \max_i \|\bm{x}_i\|_2}{\mu_n}-C_7 \Bigg] \leq 0 \Big) \\
\geq&~ P(\sup_{\bm{e} \in \mathcal{E}}~ \|\hat{\bm{\beta}}(\bm{e})\|_2 \leq 2 \|\bm{\beta}_0\|_2)-P\Big(\frac{n}{\mu_n}\max_i \|\bm{x}_i\|_2 > \frac{C_7}{2 \|\bm{\beta}_0\|_2} \Big) \rightarrow 1.
\end{align*}
The result \eqref{logistic:ref} can be  proved by combining the previous line with \eqref{inter:one}.

We now turn to deriving the consistency of $\tilde{\beta}$. We first prove the result in part (i). We have showed that
\begin{align*}
\sup_{\bm{e} \in \mathcal{E}}~\|\hat{\bm{\beta}}(\bm{e})-\bm{\beta}_0\|_2 \overset{P}{\rightarrow} 0,  \quad P(\|V(\tilde{\bm{c}})-V(\bm{c})\|_1 < \delta_n) \rightarrow 1.
\end{align*}
Hence, for any given $\epsilon >0$,
\begin{align*}
P(\|\tilde{\bm{\beta}}-\bm{\beta}_0\|_2<\epsilon)&=P(\|\hat{\bm{\beta}}(\tilde{\bm{c}})-\bm{\beta}_0\|_2<\epsilon) \\
&\geq P\Big(\sup_{\bm{e} \in \mathcal{E}}~\|\hat{\bm{\beta}}(\bm{e})-\bm{\beta}_0\|_2 < \epsilon,\|V(\tilde{\bm{c}})-V(\bm{c})\|_1< \delta_n \Big ) \\
&\geq P\Big(\sup_{\bm{e} \in \mathcal{E}}~\|\hat{\bm{\beta}}(\bm{e})-\bm{\beta}_0\|_2 < \epsilon \Big)-P\Big(\|V(\tilde{\bm{c}})-V(\bm{c})\|_1  \geq \delta_n \Big ) \rightarrow 1,
\end{align*}
which shows that $\|\tilde{\bm{\beta}}-\bm{\beta}_0\|_2\overset{P}{\rightarrow} 0$. To derive the convergence rate, denote $p_k(\bm{x};\bm{\beta})=\frac{\exp({\bm{\beta}_k^T\bm{x}})}{\sum_{k=1}^K\exp({\bm{\beta}_k^T\bm{x}})}$. We may suppress the dependency on $\bm{x}$ or $\bm{\beta}$ when it is clear from the context. By a first order Taylor expansion of $\nabla \mathcal{R}_n(\hat{\beta}(\tilde{\bm{c}}), \tilde{\bm{c}})$ around $\bm{\beta}_0$, we have\footnote{For notational simplicity, we have included $k=K$ in the subsequent arguments, though we should keep in mind that $\bm{\beta}_K\equiv \bm{0}$.}
\begin{align}
\bm{0}=\frac{1}{n}\sum_{i=1}^n [\mathbbm{1}(\tilde{c}_i=k)-p_k(\bm{x}_i;\bm{\beta}_0)]\bm{x}_i +\frac{1}{n}\sum_{i=1}^nH_k(\bm{x}_i; \acute{\bm{\beta}}^k) (\hat{\bm{\beta}}(\tilde{\bm{c}})-\bm{\beta}_0), \quad   1\leq k \leq K,  \label{taylor:exp}
\end{align}
where $$H_k(\bm{x}; \bm{\beta})\triangleq (p_kp_1\bm{x}\bm{x}^T,  \dots, p_kp_{k-1}\bm{x}\bm{x}^T, -p_k(1-p_k)\bm{x}\bm{x}^T, p_kp_{k+1}\bm{x}\bm{x}^T, \dots, p_kp_{K}\bm{x}\bm{x}^T)$$ and $\acute{\bm{\beta}}^k$ is between $\hat{\bm{\beta}}(\tilde{\bm{c}})$ and $\bm{\beta}_0$. Since $\|\hat{\bm{\beta}}(\tilde{\bm{c}})-\bm{\beta}_0\|_2 \overset{P}{\rightarrow} 0$\footnote{We recall that $\hat{\bm{\beta}}(\tilde{\bm{c}})=\tilde{\bbeta}$, according to our definition \eqref{m:estimator}. }, $\|\acute{\bm{\beta}}^k-\bm{\beta}_0\|_2 \overset{P}{\rightarrow}0$. We can then apply Theorem 9.4 in \cite{keener2010} to conclude that
\begin{align}\label{combine:zero}
\frac{1}{n}\sum_{i=1}^nH_k(\bm{x}_i;\acute{\bm{\beta}}^k)\overset{P}{\rightarrow} \mathbb{E}H_k(\bm{x}; \bm{\beta}_0), \quad 1 \leq k \leq K.
\end{align}

Next we analyze the term $\frac{1}{n}\sum_{i=1}^n [\mathbbm{1}(\tilde{c}_i=k)-p_k(\bm{x}_i; \bm{\beta}_0)] \bm{x}_i $ in  \eqref{taylor:exp}. First note that $\frac{1}{\sqrt{n}}\sum_{i=1}^n [\mathbbm{1}(c_i=k)-p_k(\bm{x}_i; \bm{\beta}_0)] \bm{x}_i =O_p(1)$. Also,
\begin{align*}
&~\Big \|\frac{1}{n}\sum_{i=1}^n [\mathbbm{1}(\tilde{c}_i=k)-\mathbbm{1}(c_i=k)] \bm{x}_i \Big \|_2 \leq \frac{1}{n} \sum_{i=1}^n \mathbbm{1}(\tilde{c_i} \neq c_i) \cdot \|\bm{x}_i\|_2\\
\leq &~ \Big(\frac{1}{n}\sum_{i=1}^n \mathbbm{1}(\tilde{c}_i\neq c_i) \Big)^{\frac{\alpha-1}{\alpha}} \cdot \Big(\frac{1}{n}\sum_{i=1}^n \|\bm{x}_i\|_2^{\alpha} \Big)^{\frac{1}{\alpha}}
\end{align*}
Therefore, we can obtain 
\begin{align}\label{piece:final}
\frac{1}{n}\sum_{i=1}^n [\mathbbm{1}(\tilde{c}_i=k)-p_k(\bm{x}_i; \bm{\beta}_0)] \bm{x}_i =O_p((n\rho_n)^{\frac{1-\alpha}{2\alpha}}).
\end{align}
Based on \eqref{taylor:exp}, \eqref{combine:zero} and \eqref{piece:final}, it is straightforward to derive the second result of part (i). Regarding the asymptotic normality of $\tilde{\bm{\beta}}$ in part (ii), we have already proved that $P(\tilde{\bm{c}} = \bm{c}) \rightarrow 1$. From standard asymptotic results, we know that $\sqrt{n}(\hat{\bm{\beta}}(\bm{c})-\bm{\beta}_0)\overset{d}{\rightarrow} N({\bf 0}, I^{-1}(\bm{\beta}_0))$. The proof can be finished by showing that $\sqrt{n}(\hat{\bm{\beta}}(\bm{c})-\tilde{\bm{\beta}}) \overset{P}{\rightarrow} 0$, with one line of arguments as follows.
\begin{align*}
P(|\sqrt{n}(\hat{\bm{\beta}}(\bm{c})-\tilde{\bm{\beta}})| <\epsilon)=P(|\sqrt{n}(\hat{\bm{\beta}}(\bm{c})-\hat{\bm{\beta}}(\tilde{\bm{c}}))| <\epsilon) \geq P(\tilde{\bm{c}} = \bm{c}) \rightarrow 1.
\end{align*}

$\hfill \Box$

\vspace{0.5cm}

\noindent \emph{\textbf{Proof of Theorem \ref{thm:sdp}}}. The proof is motivated by the arguments presented in \cite{guedon2015community}. To make the analysis compact, we will use the concentration inequality proved there:
\begin{align}\label{con:sdp}
P\Big(\bm{z}^T[A-\mathbb{E}(A\mid \bm{c})]\bm{y} > \frac{n(n-1)t}{2}   \mid \bm{c} \Big)\leq \mbox{exp}\Big(-\frac{n(n-1)t^2}{16\bar{p}(\bm{c})+8t/3}\Big),
\end{align}
where $\bar{p}(\bm{c})=\frac{2}{n(n-1)}\sum_{i<j}\mbox{Var}(A_{ij} \mid \bm{c}), \bm{z},\bm{y}\in \{-1,1\}^n$. The above result is a direct application of Bernstein's inequality (See Page 14 in \cite{guedon2015community}). Let $M(\bm{c}) \in \mathbb{R}^{n \times K}$ with $M_{ik}(\bm{c})=1$ if $c_i=k$; $M_{ik}(\bm{c})=0$ otherwise; $\bar{Z}(\bm{c})=M(\bm{c}) M^T(\bm{c})$ and
\begin{align*}
\mathcal{M}_{\bm{c}}=\Big \{Z \in \mathbb{R}^{n \times n}: Z\succeq 0, 0\leq Z_{ij}\leq 1, \sum_{ij}Z_{ij}=\sum_{k=1}^K\big(\sum_{i=1}^n\mathbbm{1}(c_i=k) \big)^2\Big \}.
\end{align*}
It is straightforward to verify that $\mathbb{E}(A\mid \bm{c})=\rho_nM(\bm{c})\bar{B}M^T(\bm{c})$ and $\bar{Z}_{ij}(\bm{c})=1$ if and only if $c_i=c_j$.
 
Before proceeding, define $S_n(Z) =  \langle A +\gamma_nXX^T, Z \rangle$ and  $S(Z) = \langle \mathbb{E}(A \mid \bm{c})+\gamma_n\mathbb{E}(X\mid \bm{c})\mathbb{E}(X^T\mid \bm{c}),  Z \rangle$. Then, by the definition in \eqref{sdp:formula}, we have $\hat Z = \arg\max_{Z\in  \mathcal{M}_{\bm{c}}} S_n(Z)$. 
Hence for any $Z \in \mathcal{M}_{\bm{c}}$, we have
\begin{align}\label{strong:convexity}
&~S(\bar{Z}(\bm{c})) - S(Z)\\
&=\langle \mathbb{E}(A \mid \bm{c})+\gamma_n\mathbb{E}(X\mid \bm{c})\mathbb{E}(X^T\mid \bm{c}),  \bar{Z}(\bm{c})-Z \rangle  \nonumber  \\
 &=\sum_{ij} (\rho_n\bar{B}_{c_ic_j}+\gamma_n\mathbb{E}(\bm{x}_i^T\mid c_i)\mathbb{E}(\bm{x}_j \mid c_j))\cdot (\bar{Z}_{ij}(\bm{c})-Z_{ij})\cdot \mathbbm{1}(c_i=c_j) \nonumber \\
 &\quad -\sum_{ij} (\rho_n\bar{B}_{c_ic_j}+\gamma_n\mathbb{E}(\bm{x}_i^T\mid c_i)\mathbb{E}(\bm{x}_j \mid c_j))\cdot (Z_{ij}-\bar{Z}_{ij}(\bm{c})) \cdot \mathbbm{1}(c_i \neq c_j)  \nonumber \\
 &\overset{(a)}{\geq} U \cdot \sum_{ij}(\bar{Z}_{ij}(\bm{c})-Z_{ij})\cdot \mathbbm{1}(c_i=c_j)-L \cdot \sum_{ij} (Z_{ij}-\bar{Z}_{ij}(\bm{c})) \cdot \mathbbm{1}(c_i \neq c_j) \nonumber \\
 &\overset{(b)}{\geq} \frac{U-L}{2} \cdot \|\bar{Z}(\bm{c})-Z\|_1, \nonumber
\end{align}
where $U=\min_{1 \leq k \leq K}\{\rho_n\bar{B}_{kk}+\gamma_n\mathbb{E}(\bm{x}^T\mid c=k)\mathbb{E}(\bm{x}\mid c=k)\}, L=\max_{a\neq b}\{\rho_n\bar{B}_{ab}+\gamma_n\mathbb{E}(\bm{x}^T\mid c=a)\mathbb{E}(\bm{x}\mid c=b)\}$; in (a) we have used the fact that $\bar{Z}_{ij}(\bm{c})\geq Z_{ij}$ if $c_i=c_j$ and $\bar{Z}_{ij}(\bm{c})\leq Z_{ij}$ otherwise; (b) holds because  $\sum_{ij}(\bar{Z}_{ij}(\bm{c})-Z_{ij})\cdot \mathbbm{1}(c_i=c_j)=\sum_{ij}(Z_{ij}-\bar{Z}_{ij}(\bm{c}))\cdot \mathbbm{1}(c_i \neq c_j)=\frac{1}{2}\|\bar{Z}(\bm{c})-Z\|_1$. Note that the conditions we assumed imply $U>L>0$ when $n$ is large enough. 
In addition, we observe that 
\begin{align}\label{eq:tri-Zbar-Zhat}
S(\bar{Z}(\bm{c})) - S(\hat Z)&\le [S(\bar{Z}(\bm{c})) - S_n(\bar{Z}(\bm{c}))] + [S_n(\bar{Z}(\bm{c}))  - S_n(\hat Z)] + [S_n(\hat Z) - S(\hat Z)]\\
&\le 2\max_{Z \in \mathcal{M}_{\bm{c}}}|S_n(Z)-S(Z)|\nonumber,
\end{align}
since $S_n(\bar{Z}(\bm{c}))  - S_n(\hat Z)\le 0$.

 Moreover, by applying the Grothendieck's inequality described in Section 3.1 of \citep{guedon2015community},
 we can have
\begin{align}
&~\max_{Z \in \mathcal{M}_{\bm{c}}}|S_n(Z)-S(Z)|\label{uniform:bound}\\
= &~\max_{Z \in \mathcal{M}_{\bm{c}}} |\langle A- \mathbb{E}(A \mid \bm{c})+\gamma_n XX^T-\gamma_n\mathbb{E}(X\mid \bm{c})\mathbb{E}(X^T\mid \bm{c}), Z  \rangle| \nonumber  \\
\leq &~2 \max_{\bm{z}, \bm{y} \in \{-1,1\}^n}\bm{z}^T[A- \mathbb{E}(A \mid \bm{c})+\gamma_n XX^T-\gamma_n\mathbb{E}(X\mid \bm{c})\mathbb{E}(X^T\mid \bm{c})]\bm{y}.  \nonumber
\end{align}
Combining \eqref{strong:convexity}, \eqref{eq:tri-Zbar-Zhat}, and \eqref{uniform:bound}, it is not hard to obtain
\begin{align}\label{first:step}
\|\bar{Z}(\bm{c})-\hat{Z}\|_1 \leq \frac{8}{U-L} \cdot \max_{\bm{z}, \bm{y} \in \{-1,1\}^n}&\bm{z}^T[A- \mathbb{E}(A \mid \bm{c})+\gamma_n XX^T\\
&-\gamma_n\mathbb{E}(X\mid \bm{c})\mathbb{E}(X^T\mid \bm{c})]\bm{y}.\nonumber 
\end{align}

We now bound the term on the right hand side of the above inequality. Note that for sufficiently large $n$
\begin{align*}
\bar{p}(\bm{c})=\frac{2}{n(n-1)}\sum_{i<j} \rho_n\bar{B}_{c_ic_j}(1-\rho_n\bar{B}_{c_ic_j})\geq \frac{1}{n(n-1)}\min_{ab}\bar{B}_{ab}\sum_{i<j}\rho_n=\frac{\min_{ab}\bar{B}_{ab}}{2}\cdot \rho_n.
\end{align*}
Therefore, according to \eqref{con:sdp}, choosing $t=n^{-1/3}\rho_n^{2/3}$ yields that there exists $C_1>0$ such that when $n$ is large
\begin{align*}
P\Big(\bm{z}^T[A-\mathbb{E}(A\mid \bm{c})]\bm{y} > n^{5/3}\rho_n^{2/3}  \mid \bm{c} \Big)\leq \mbox{exp}\big(-C_1 n\cdot(n\rho_n)^{1/3} \big),
\end{align*}
for any $\bm{z}, \bm{y}\in \{-1,1\}^n$. Using the union bound, we can conclude 
\begin{align}\label{unibound:one}
P\Big(\max_{\bm{z}, \bm{y}\in \{-1,1\}^n}\bm{z}^T[A-\mathbb{E}(A\mid \bm{c})]\bm{y} > n^{5/3}\rho_n^{2/3}  \Big)\leq 2^{2n} \cdot \mbox{exp}\big(-C_1 n\cdot(n\rho_n)^{1/3} \big) \rightarrow 0.
\end{align}
Regarding the bound on $\gamma_n XX^T-\gamma_n\mathbb{E}(X\mid \bm{c})\mathbb{E}(X^T\mid \bm{c})$, we first denote the $i$-th column of $X$ by $\bm{d}_i$. Since $\bm{d}_i \mid \bm{c}$ has independent sub-Gaussian elements\footnote{This can be directly shown from the condition that $\|\bm{x}\|_2$ is sub-Gaussian.}, we can apply Hoeffding's inequality to obtain
\begin{align*}
P\big(|\bm{y}^T(\bm{d}_i-\mathbb{E}(\bm{d}_i\mid \bm{c}))| > nt \mid \bm{c} \big) \leq 2e^{-C_2nt^2},
\end{align*}
where $C_2>0$ is a constant and $\bm{y}\in \{1,-1\}^n$. Note that we can choose $C_2$ small enough to guarantee that the above  inequality holds for all $\bm{c}$. Accordingly, we have
\begin{align}\label{inter:media}
P\big(|\bm{y}^T(\bm{d}_i-\mathbb{E}(\bm{d}_i\mid \bm{c}))| > nt  \big) \leq 2 e^{-C_2nt^2}.
\end{align}
This implies the following bound for large $t$, $\bm{y},\bm{z} \in \{-1,1\}^n$,
\begin{align}\label{inter:media2}
&~P(|\bm{z}^T[\bm{d}_i\bm{d}_i^T-\mathbb{E}(\bm{d}_i\mid \bm{c})\mathbb{E}(\bm{d}^T_i\mid \bm{c})]\bm{y}| > n^2t^2)\\
\leq &~P(|\bm{z}^T(\bm{d}_i-\mathbb{E}(\bm{d}_i\mid \bm{c}))\cdot \bm{y}^T(\bm{d}_i-\mathbb{E}(\bm{d}_i\mid \bm{c}))| >n^2t^2/3 ) \nonumber \\
&~+ P(|\bm{z}^T(\bm{d}_i-\mathbb{E}(\bm{d}_i\mid \bm{c}))\cdot \bm{y}^T\mathbb{E}(\bm{d}_i\mid \bm{c})| >n^2t^2/3 )\nonumber\\
&~+P(|\bm{z}^T\mathbb{E}(\bm{d}_i\mid \bm{c})\cdot \bm{y}^T(\bm{d}_i-\mathbb{E}(\bm{d}_i\mid \bm{c}))| >n^2t^2/3 ) \nonumber \\
\overset{(c)}{\leq} &~4e^{-C_3nt^2}+2e^{-C_4nt^4}+2e^{-C_4nt^4}\leq 8e^{-C_5nt^2}. \nonumber
\end{align}
where $C_3, C_4, C_5$ are positive constants; we have used \eqref{inter:media} and the fact that \\$\max_{\bm{y}\in \{-1,1\}^n}|\bm{y}^T\mathbb{E}(\bm{d}_i\mid \bm{c})| =O(n) $ to obtain (c). Applying the inequality \eqref{inter:media2} enables us to conclude that
\begin{align*}
&~P(\gamma_n |\bm{z}^T(XX^T-\mathbb{E}(X\mid \bm{c})\mathbb{E}(X^T\mid \bm{c}))\bm{y}| > n^2t^2)\\
 \leq&~ p\cdot P(|\bm{z}^T[\bm{d}_i\bm{d}_i^T-\mathbb{E}(\bm{d}_i\mid \bm{c})\mathbb{E}(\bm{d}^T_i\mid \bm{c})]\bm{y}| > n^2t^2/(p\gamma_n))  \leq 8pe^{-C_5 n t^2/(p\gamma_n)}.
\end{align*}
Choosing $t=C_6\sqrt{\gamma_n}$ and using the union bound, we then have
\begin{align}\label{unibound:two}
&P\Big(\max_{\bm{z},\bm{y} \in \{-1,1\}^n} \gamma_n |\bm{z}^T(XX^T-\mathbb{E}(X\mid \bm{c})\mathbb{E}(X^T\mid \bm{c}))\bm{y}|    > C_6n^2\gamma_n \Big)\\
& \leq 2^{2n}\cdot 8pe^{-C_5 C_6^2np^{-1}} \rightarrow 0 \nonumber,
\end{align}
where the last limit holds for sufficiently large $C_6$. Putting \eqref{first:step}, \eqref{unibound:one} and \eqref{unibound:two} together gives us that 
\begin{align*}
\frac{\|\bar{Z}(\bm{c})-\hat{Z}\|_1}{n^2} \leq \frac{8n^{-1/3}\rho_n^{2/3}+8C_6\gamma_n}{U-L},
\end{align*}
holds with probability approaching 1. Since $U-L$ is of order $\rho_n$ and $\gamma_n=o(\rho_n), n\rho_n\rightarrow \infty$, we obtain that $\frac{\|\bar{Z}(\bm{c})-\hat{Z}\|_1}{n^2} \overset{P}{\rightarrow}0$.

 Finally, we use similar arguments as in \cite{rohe2011spectral} to analyze the K-means step and show the mis-classification rate vanishes. Denote the $K$ centroids output from K-means run on $\hat{Z}$ by $\hat{\bm{\mu}}_1,\dots, \hat{\bm{\mu}}_K$ and on $\bar{Z}(\bm{c})$ by $\bm{\mu}_1, \dots, \bm{\mu}_K$; $\hat{C}=(\hat{\bm{\mu}}_{\bar{c}_1}, \dots, \hat{\bm{\mu}}_{\bar{c}_n})^T; \mathscr{C}=\{C\in \mathbb{R}^{n \times n}: C~ \mbox{has~exactly~} K \mbox{~non-identical~rows}\}$. We then know $\bar{Z}(\bm{c})=(\bm{\mu}_{c_1},\dots, \bm{\mu}_{c_n})^T$ and also
\begin{align}
\hat{C}=\argmin_{C \in \mathscr{C}} \|C-\hat{Z}\|_F^2. \label{optima:kmeans}
\end{align}
Note that the elements of $\bmu_i$ are either 0 or 1 and mutually orthogonal. Since the size of each community is proportional to $n$, we have with high probability, there exists a constant $C_6>0$ such that $\|\bm{\mu}_i-\bm{\mu}_j\|_2\geq C_6\sqrt{n}$ for $1 \leq i \neq j \leq K$.

We now define a set of nodes that are ``incorrectly" identified by the K-means:
\begin{align*}
\mathscr{N}=\{i\in \{1, 2,\dots, n\}:  \|\hat{\bm{\mu}}_{\bar{c}_i}-\bm{\mu}_{c_i}\|_2 \geq C_6 \sqrt{n}/2  \}.
\end{align*} 
Hence, we have
\begin{align*}
\frac{|\mathscr{N}|}{n} &\leq \frac{1}{n}\sum_{i \in \mathscr{N}} \|\hat{\bm{\mu}}_{\bar{c}_i}-\bm{\mu}_{c_i}\|^2_2 \cdot \frac{4}{C_6^2n} \leq \frac{4}{C_6^2n^2}\sum_{i } \|\hat{\bm{\mu}}_{\bar{c}_i}-\bm{\mu}_{c_i}\|^2_2 \\
&= \frac{4}{C_6^2n^2} \|\hat{C}-\bar{Z}(\bm{c})\|_F^2 \overset{(d)}{\leq}  \frac{16}{C_6^2n^2} \|\hat{Z}-\bar{Z}(\bm{c})\|_F^2\leq  \frac{16}{C_6^2n^2} \|\hat{Z}-\bar{Z}(\bm{c})\|_1,
\end{align*}
where (d) is due to \eqref{optima:kmeans}. This combined with the fact that $\frac{\|\bar{Z}(\bm{c})-\hat{Z}\|_1}{n^2} \overset{P}{\rightarrow}0$ shows that $\frac{|\mathscr{N}|}{n}  \overset{P}{\rightarrow} 0$. On the other hand, for any $i \notin \mathscr{N}$, it is straightforward to verify that for any $1\leq k\neq c_i \leq K$, we have 
\begin{align*}
\|\hat{\bm{\mu}}_{\bar{c}_i}-\bm{\mu}_{k}\|_2&\ge \|{\bm{\mu}}_{{c}_i}-\bm{\mu}_{k}\|_2- \|\hat{\bm{\mu}}_{\bar{c}_i}-{\bm{\mu}}_{{c}_i}\|_2\\
&\ge  C_6 \sqrt{n}- C_6 \sqrt{n}/2= C_6 \sqrt{n}/2> \|\hat{\bm{\mu}}_{\bar{c}_i}-\bm{\mu}_{c_i}\|_2. 
\end{align*}

We construct a new community assignment:
\begin{align*}
\bar{\bar{c}}_i= \argmin_{a \in \{1, \dots, K\}}  \|\hat{\bm{\mu}}_{\bar{c}_i}-\bm{\mu}_{a}\|_2.
\end{align*}
Note that the communities that $\{\bar{\bar{c}}_{i}\}$ represents might be different from $\{\bar{c}_i\}$, since $\{\bar{\bar{c}}_i\}$ could cover less than $K$ communities. We can view $\{\bar{c}_i\}$ as a possibly finer partition of $\{\bar{\bar{c}}_i\}$. Clearly, for any $i \notin \mathscr{N}$, $\bar{\bar{c}}_i=c_i$. We thus have 
\begin{align}\label{contradiction}
\frac{1}{n}\sum_{i}\mathbbm{1}(\bar{\bar{c}}_i \neq c_i) \leq\frac{1}{n} |\mathscr{N}| \overset{P}{\rightarrow}0.
\end{align}
This implies that $\{\bar{\bar{c}}_i\}=\{\bar{c}_i\}$ with probability approaching 1. Otherwise, the estimate $\{\bar{\bar{c}}_i\}$ yields less than $K$ communities with non-vanishing probability, which is impossible for $\{\bar{\bar{c}}_i\}$ to achieve the result in \eqref{contradiction}. This completes the proof.

$\hfill \Box$

\bibliography{mc}

\begin{thebibliography}{56}
\expandafter\ifx\csname natexlab\endcsname\relax\def\natexlab#1{#1}\fi
\expandafter\ifx\csname url\endcsname\relax
  \def\url#1{\texttt{#1}}\fi
\expandafter\ifx\csname urlprefix\endcsname\relax\def\urlprefix{URL }\fi
\providecommand{\eprint}[2][]{\url{#2}}

\bibitem[{Abbe and Sandon(2015)}]{abbe2015community}
\textsc{Abbe, E.} and \textsc{Sandon, C.} (2015).
\newblock Community detection in general stochastic block models: fundamental
  limits and efficient recovery algorithms.
\newblock \textit{arXiv:1503.00609}.

\bibitem[{Airoldi et~al.(2009)Airoldi, Blei, Fienberg and
  Xing}]{airoldi2009mixed}
\textsc{Airoldi, E.~M.}, \textsc{Blei, D.~M.}, \textsc{Fienberg, S.~E.} and
  \textsc{Xing, E.~P.} (2009).
\newblock Mixed membership stochastic blockmodels.
\newblock In \textit{Advances in Neural Information Processing Systems}.
  33--40.

\bibitem[{Akoglu et~al.(2012)Akoglu, Tong, Meeder and
  Faloutsos}]{akoglu2012pics}
\textsc{Akoglu, L.}, \textsc{Tong, H.}, \textsc{Meeder, B.} and
  \textsc{Faloutsos, C.} (2012).
\newblock Pics: Parameter-free identification of cohesive subgroups in large
  attributed graphs.
\newblock In \textit{SDM}. Citeseer, 439--450.

\bibitem[{Amini et~al.(2013)Amini, Chen, Bickel and Levina}]{amini2013pseudo}
\textsc{Amini, A.~A.}, \textsc{Chen, A.}, \textsc{Bickel, P.~J.} and
  \textsc{Levina, E.} (2013).
\newblock Pseudo-likelihood methods for community detection in large sparse
  networks.
\newblock \textit{The Annals of Statistics}, \textbf{41} 2097--2122.

\bibitem[{Amini and Levina(2014)}]{amini2014semidefinite}
\textsc{Amini, A.~A.} and \textsc{Levina, E.} (2014).
\newblock On semidefinite relaxations for the block model.
\newblock \textit{arXiv:1406.5647}.

\bibitem[{Ana and Jain(2003)}]{ana2003robust}
\textsc{Ana, L.} and \textsc{Jain, A.~K.} (2003).
\newblock Robust data clustering.
\newblock In \textit{Computer Vision and Pattern Recognition, 2003.
  Proceedings. 2003 IEEE Computer Society Conference on}, vol.~2. IEEE,
  II--128.

\bibitem[{Anandkumar et~al.(2014)Anandkumar, Ge, Hsu and
  Kakade}]{anandkumar2014tensor}
\textsc{Anandkumar, A.}, \textsc{Ge, R.}, \textsc{Hsu, D.} and \textsc{Kakade,
  S.~M.} (2014).
\newblock A tensor approach to learning mixed membership community models.
\newblock \textit{The Journal of Machine Learning Research}, \textbf{15}
  2239--2312.

\bibitem[{Bickel et~al.(2013)Bickel, Choi, Chang and
  Zhang}]{bickel2013asymptotic}
\textsc{Bickel, P.}, \textsc{Choi, D.}, \textsc{Chang, X.} and \textsc{Zhang,
  H.} (2013).
\newblock Asymptotic normality of maximum likelihood and its variational
  approximation for stochastic blockmodels.
\newblock \textit{The Annals of Statistics}, \textbf{41} 1922--1943.

\bibitem[{Bickel and Chen(2009)}]{bickel2009nonparametric}
\textsc{Bickel, P.~J.} and \textsc{Chen, A.} (2009).
\newblock A nonparametric view of network models and newman--girvan and other
  modularities.
\newblock \textit{Proceedings of the National Academy of Sciences},
  \textbf{106} 21068--21073.

\bibitem[{Bickel et~al.(2015)Bickel, Chen, Zhao, Levina and
  Zhu}]{bickel2015correction}
\textsc{Bickel, P.~J.}, \textsc{Chen, A.}, \textsc{Zhao, Y.}, \textsc{Levina,
  E.} and \textsc{Zhu, J.} (2015).
\newblock Correction to the proof of consistency of community detection.
\newblock \textit{The Annals of Statistics}, \textbf{43} 462--466.

\bibitem[{Binkiewicz et~al.(2014)Binkiewicz, Vogelstein and
  Rohe}]{binkiewicz2014covariate}
\textsc{Binkiewicz, N.}, \textsc{Vogelstein, J.~T.} and \textsc{Rohe, K.}
  (2014).
\newblock Covariate assisted spectral clustering.
\newblock \textit{arXiv preprint arXiv:1411.2158}.

\bibitem[{Boyd et~al.(2011)Boyd, Parikh, Chu, Peleato and
  Eckstein}]{boyd2011distributed}
\textsc{Boyd, S.}, \textsc{Parikh, N.}, \textsc{Chu, E.}, \textsc{Peleato, B.}
  and \textsc{Eckstein, J.} (2011).
\newblock Distributed optimization and statistical learning via the alternating
  direction method of multipliers.
\newblock \textit{Foundations and Trends{\textregistered} in Machine Learning},
  \textbf{3} 1--122.

\bibitem[{Cai and Li(2015)}]{cai2015robust}
\textsc{Cai, T.~T.} and \textsc{Li, X.} (2015).
\newblock Robust and computationally feasible community detection in the
  presence of arbitrary outlier nodes.
\newblock \textit{The Annals of Statistics}, \textbf{43} 1027--1059.

\bibitem[{Celisse et~al.(2012)Celisse, Daudin and
  Pierre}]{celisse2012consistency}
\textsc{Celisse, A.}, \textsc{Daudin, J.-J.} and \textsc{Pierre, L.} (2012).
\newblock Consistency of maximum-likelihood and variational estimators in the
  stochastic block model.
\newblock \textit{Electronic Journal of Statistics}, \textbf{6} 1847--1899.

\bibitem[{Chang and Blei(2010)}]{chang2010hierarchical}
\textsc{Chang, J.} and \textsc{Blei, D.~M.} (2010).
\newblock Hierarchical relational models for document networks.
\newblock \textit{The Annals of Applied Statistics} 124--150.

\bibitem[{Chen et~al.(2015)Chen, Li and Xu}]{chen2015convexified}
\textsc{Chen, Y.}, \textsc{Li, X.} and \textsc{Xu, J.} (2015).
\newblock Convexified modularity maximization for degree-corrected stochastic
  block models.
\newblock \textit{arXiv:1512.08425}.

\bibitem[{Chen et~al.(2012)Chen, Sanghavi and Xu}]{chen2012clustering}
\textsc{Chen, Y.}, \textsc{Sanghavi, S.} and \textsc{Xu, H.} (2012).
\newblock Clustering sparse graphs.
\newblock In \textit{Advances in neural information processing systems}.
  2204--2212.

\bibitem[{Choi et~al.(2012)Choi, Wolfe and Airoldi}]{choi2012stochastic}
\textsc{Choi, D.~S.}, \textsc{Wolfe, P.~J.} and \textsc{Airoldi, E.~M.} (2012).
\newblock Stochastic blockmodels with a growing number of classes.
\newblock \textit{Biometrika} asr053.

\bibitem[{Cross and Parker(2004)}]{cross2004hidden}
\textsc{Cross, R.~L.} and \textsc{Parker, A.} (2004).
\newblock \textit{The hidden power of social networks: Understanding how work
  really gets done in organizations}.
\newblock Harvard Business Review Press.

\bibitem[{Dasgupta et~al.(2004)Dasgupta, Hopcroft and
  McSherry}]{dasgupta2004spectral}
\textsc{Dasgupta, A.}, \textsc{Hopcroft, J.~E.} and \textsc{McSherry, F.}
  (2004).
\newblock Spectral analysis of random graphs with skewed degree distributions.
\newblock In \textit{Foundations of Computer Science, 2004. Proceedings. 45th
  Annual IEEE Symposium on}. IEEE, 602--610.

\bibitem[{Daudin et~al.(2008)Daudin, Picard and Robin}]{daudin2008mixture}
\textsc{Daudin, J.-J.}, \textsc{Picard, F.} and \textsc{Robin, S.} (2008).
\newblock A mixture model for random graphs.
\newblock \textit{Statistics and computing}, \textbf{18} 173--183.

\bibitem[{Decelle et~al.(2011)Decelle, Krzakala, Moore and
  Zdeborov{\'a}}]{decelle2011asymptotic}
\textsc{Decelle, A.}, \textsc{Krzakala, F.}, \textsc{Moore, C.} and
  \textsc{Zdeborov{\'a}, L.} (2011).
\newblock Asymptotic analysis of the stochastic block model for modular
  networks and its algorithmic applications.
\newblock \textit{Physical Review E}, \textbf{84} 066106.

\bibitem[{Dempster et~al.(1977)Dempster, Laird and Rubin}]{dempster1977maximum}
\textsc{Dempster, A.~P.}, \textsc{Laird, N.~M.} and \textsc{Rubin, D.~B.}
  (1977).
\newblock Maximum likelihood from incomplete data via the em algorithm.
\newblock \textit{Journal of the royal statistical society. Series B
  (methodological)} 1--38.

\bibitem[{Fortunato(2010)}]{fortunato2010community}
\textsc{Fortunato, S.} (2010).
\newblock Community detection in graphs.
\newblock \textit{Physics Reports}, \textbf{486} 75--174.

\bibitem[{Gao et~al.(2015)Gao, Ma, Zhang and Zhou}]{gao2015achieving}
\textsc{Gao, C.}, \textsc{Ma, Z.}, \textsc{Zhang, A.~Y.} and \textsc{Zhou,
  H.~H.} (2015).
\newblock Achieving optimal misclassification proportion in stochastic block
  model.
\newblock \textit{arXiv:1505.03772}.

\bibitem[{Gu{\'e}don and Vershynin(2016)}]{guedon2015community}
\textsc{Gu{\'e}don, O.} and \textsc{Vershynin, R.} (2016).
\newblock Community detection in sparse networks via grothendieckÕs
  inequality.
\newblock \textit{Probability Theory and Related Fields}, \textbf{165}
  1025--1049.

\bibitem[{Holland et~al.(1983)Holland, Laskey and
  Leinhardt}]{holland1983stochastic}
\textsc{Holland, P.~W.}, \textsc{Laskey, K.~B.} and \textsc{Leinhardt, S.}
  (1983).
\newblock Stochastic blockmodels: First steps.
\newblock \textit{Social networks}, \textbf{5} 109--137.

\bibitem[{Hubert and Arabie(1985)}]{hubert1985comparing}
\textsc{Hubert, L.} and \textsc{Arabie, P.} (1985).
\newblock Comparing partitions.
\newblock \textit{Journal of classification}, \textbf{2} 193--218.

\bibitem[{Jin(2015)}]{jin2015fast}
\textsc{Jin, J.} (2015).
\newblock Fast community detection by score.
\newblock \textit{The Annals of Statistics}, \textbf{43} 57--89.

\bibitem[{Jordan et~al.(1999)Jordan, Ghahramani, Jaakkola and
  Saul}]{jordan1999introduction}
\textsc{Jordan, M.~I.}, \textsc{Ghahramani, Z.}, \textsc{Jaakkola, T.~S.} and
  \textsc{Saul, L.~K.} (1999).
\newblock An introduction to variational methods for graphical models.
\newblock \textit{Machine learning}, \textbf{37} 183--233.

\bibitem[{Joseph and Yu(2013)}]{joseph2013impact}
\textsc{Joseph, A.} and \textsc{Yu, B.} (2013).
\newblock Impact of regularization on spectral clustering.
\newblock \textit{arXiv:1312.1733}.

\bibitem[{Karrer and Newman(2011)}]{karrer2011stochastic}
\textsc{Karrer, B.} and \textsc{Newman, M.~E.} (2011).
\newblock Stochastic blockmodels and community structure in networks.
\newblock \textit{Physical Review E}, \textbf{83} 016107.

\bibitem[{Keener(2010)}]{keener2010}
\textsc{Keener, R.~W.} (2010).
\newblock \textit{Theoretical Statistics: Topics for a Core Course (Springer
  Texts in Statistics)}, vol.~1.
\newblock Springer.

\bibitem[{Krzakala et~al.(2013)Krzakala, Moore, Mossel, Neeman, Sly,
  Zdeborov{\'a} and Zhang}]{krzakala2013spectral}
\textsc{Krzakala, F.}, \textsc{Moore, C.}, \textsc{Mossel, E.}, \textsc{Neeman,
  J.}, \textsc{Sly, A.}, \textsc{Zdeborov{\'a}, L.} and \textsc{Zhang, P.}
  (2013).
\newblock Spectral redemption in clustering sparse networks.
\newblock \textit{Proceedings of the National Academy of Sciences},
  \textbf{110} 20935--20940.

\bibitem[{Le and Levina(2015)}]{le2015estimating}
\textsc{Le, C.~M.} and \textsc{Levina, E.} (2015).
\newblock Estimating the number of communities in networks by spectral methods.
\newblock \textit{arXiv preprint arXiv:1507.00827}.

\bibitem[{Lei(2016)}]{lei2016goodness}
\textsc{Lei, J.} (2016).
\newblock A goodness-of-fit test for stochastic block models.
\newblock \textit{The Annals of Statistics}, \textbf{44} 401--424.

\bibitem[{Lei and Rinaldo(2014)}]{lei2014consistency}
\textsc{Lei, J.} and \textsc{Rinaldo, A.} (2014).
\newblock Consistency of spectral clustering in stochastic block models.
\newblock \textit{The Annals of Statistics}, \textbf{43} 215--237.

\bibitem[{Liese and Miescke(2007)}]{liese2007statistical}
\textsc{Liese, F.} and \textsc{Miescke, K.-J.} (2007).
\newblock Statistical decision theory.
\newblock In \textit{Statistical Decision Theory}. Springer, 1--52.

\bibitem[{Montanari and Sen(2015)}]{montanari2015semidefinite}
\textsc{Montanari, A.} and \textsc{Sen, S.} (2015).
\newblock Semidefinite programs on sparse random graphs and their application
  to community detection.
\newblock \textit{arXiv:1504.05910}.

\bibitem[{Nallapati and Cohen(2008)}]{nallapati2008link}
\textsc{Nallapati, R.} and \textsc{Cohen, W.~W.} (2008).
\newblock Link-plsa-lda: A new unsupervised model for topics and influence of
  blogs.
\newblock In \textit{ICWSM}.

\bibitem[{Newman and Clauset(2015)}]{newman2015structure}
\textsc{Newman, M.} and \textsc{Clauset, A.} (2015).
\newblock Structure and inference in annotated networks.
\newblock \textit{arXiv:1507.04001}.

\bibitem[{Newman(2003)}]{newman2003structure}
\textsc{Newman, M.~E.} (2003).
\newblock The structure and function of complex networks.
\newblock \textit{SIAM review}, \textbf{45} 167--256.

\bibitem[{Newman(2006)}]{newman2006modularity}
\textsc{Newman, M.~E.} (2006).
\newblock Modularity and community structure in networks.
\newblock \textit{Proceedings of the National Academy of Sciences},
  \textbf{103} 8577--8582.

\bibitem[{Qin and Rohe(2013)}]{qin2013regularized}
\textsc{Qin, T.} and \textsc{Rohe, K.} (2013).
\newblock Regularized spectral clustering under the degree-corrected stochastic
  blockmodel.
\newblock In \textit{Advances in Neural Information Processing Systems}.
  3120--3128.

\bibitem[{Rohe et~al.(2011)Rohe, Chatterjee and Yu}]{rohe2011spectral}
\textsc{Rohe, K.}, \textsc{Chatterjee, S.} and \textsc{Yu, B.} (2011).
\newblock Spectral clustering and the high-dimensional stochastic blockmodel.
\newblock \textit{The Annals of Statistics}, \textbf{39} 1878--1915.

\bibitem[{Ruan et~al.(2013)Ruan, Fuhry and Parthasarathy}]{ruan2013efficient}
\textsc{Ruan, Y.}, \textsc{Fuhry, D.} and \textsc{Parthasarathy, S.} (2013).
\newblock Efficient community detection in large networks using content and
  links.
\newblock In \textit{Proceedings of the 22nd international conference on world
  wide web}. International World Wide Web Conferences Steering Committee,
  1089--1098.

\bibitem[{Saade et~al.(2014)Saade, Krzakala and
  Zdeborov{\'a}}]{saade2014spectral}
\textsc{Saade, A.}, \textsc{Krzakala, F.} and \textsc{Zdeborov{\'a}, L.}
  (2014).
\newblock Spectral clustering of graphs with the bethe hessian.
\newblock In \textit{Advances in Neural Information Processing Systems}.
  406--414.

\bibitem[{Saldana et~al.(2015)Saldana, Yu and Feng}]{saldana2015many}
\textsc{Saldana, D.~F.}, \textsc{Yu, Y.} and \textsc{Feng, Y.} (2015).
\newblock How many communities are there?
\newblock \textit{Journal of Computational and Graphical Statistics}.
\newblock To appear.

\bibitem[{Steinhaeuser and Chawla(2010)}]{steinhaeuser2010identifying}
\textsc{Steinhaeuser, K.} and \textsc{Chawla, N.~V.} (2010).
\newblock Identifying and evaluating community structure in complex networks.
\newblock \textit{Pattern Recognition Letters}, \textbf{31} 413--421.

\bibitem[{T{\"u}t{\"u}nc{\"u} et~al.(2003)T{\"u}t{\"u}nc{\"u}, Toh and
  Todd}]{tutuncu2003solving}
\textsc{T{\"u}t{\"u}nc{\"u}, R.~H.}, \textsc{Toh, K.~C.} and \textsc{Todd,
  M.~J.} (2003).
\newblock Solving semidefinite-quadratic-linear programs using sdpt3.
\newblock \textit{Mathematical programming}, \textbf{95} 189--217.

\bibitem[{Wang and Bickel(2015)}]{wang2015likelihood}
\textsc{Wang, Y.} and \textsc{Bickel, P.~J.} (2015).
\newblock Likelihood-based model selection for stochastic block models.
\newblock \textit{arXiv preprint arXiv:1502.02069}.

\bibitem[{Yang et~al.(2013)Yang, McAuley and Leskovec}]{yang2013community}
\textsc{Yang, J.}, \textsc{McAuley, J.} and \textsc{Leskovec, J.} (2013).
\newblock Community detection in networks with node attributes.
\newblock In \textit{Data Mining (ICDM), 2013 IEEE 13th International
  Conference on}. IEEE, 1151--1156.

\bibitem[{Zhang and Zhou(2015)}]{zhang2015minimax}
\textsc{Zhang, A.~Y.} and \textsc{Zhou, H.~H.} (2015).
\newblock Minimax rates of community detection in stochastic block models.
\newblock \textit{arXiv preprint arXiv:1507.05313}.

\bibitem[{Zhang et~al.(2013)Zhang, Levina and Zhu}]{zhang2013community}
\textsc{Zhang, Y.}, \textsc{Levina, E.} and \textsc{Zhu, J.} (2013).
\newblock Community detection in networks with node features.
\newblock In \textit{Advances in Neural Information Processing Systems}.

\bibitem[{Zhang et~al.(2014)Zhang, Levina and Zhu}]{zhang2014detecting}
\textsc{Zhang, Y.}, \textsc{Levina, E.} and \textsc{Zhu, J.} (2014).
\newblock Detecting overlapping communities in networks with spectral methods.
\newblock \textit{arXiv:1412.3432}.

\bibitem[{Zhao et~al.(2012)Zhao, Levina and Zhu}]{zhao2012consistency}
\textsc{Zhao, Y.}, \textsc{Levina, E.} and \textsc{Zhu, J.} (2012).
\newblock Consistency of community detection in networks under degree-corrected
  stochastic block models.
\newblock \textit{The Annals of Statistics}, \textbf{40} 2266--2292.

\end{thebibliography}
\bibliographystyle{ims}

\end{document}